\documentclass[10pt]{article}

\usepackage{amsmath, latexsym, amssymb, amscd,graphicx,amsthm,amstext,amsxtra,amsopn,amsbsy,young,enumitem}
\usepackage[vcentermath]{youngtab}

\newtheorem{thm}{\bf Theorem}[section]
\newtheorem{ex}[thm]{\bf Example}

\newtheorem{cor}[thm]{\bf Corollary}
\newtheorem{defin}[thm]{\textsl{\bf Definition}{}}
\newtheorem{lem}[thm]{\bf Lemma}

\newtheorem{remark}[thm]{\bf Remark}

\newcommand{\dd}{\mathrm{d}}

\newcommand{\quant}{\vec{\zeta}}
\newcommand{\quantsub}{\zeta}
\newcommand{\quantzero}{\vec{0}}
\newcommand{\glob}{\theta}
\newcommand{\Paramglob}{\Theta}
\newcommand{\clas}{\vec{u}}

\newcommand{\cartan}{\vec{\xi}}
\newcommand{\cartansub}{\xi}
\newcommand{\cartanzero}{\vec{0}}
\newcommand{\Vb}{{f_{\bf b}}}
\newcommand{\Vm}{{f_{\bf m}}}
\newcommand{\Vzero}{{f_{\bf 0}}}
\newcommand{\Vl}{{f_{\bf l}}}
\newcommand{\Va}{{f_{\bf a}}} % Il va falloir modifier certaines choses autour de colmod si on change
\newcommand{\Vbp}{{f_{\bf b'}}}
\newcommand{\Vap}{{f_{\bf a'}}}
\newcommand{\lc}{{l(c)}}
\newcommand{\tcb}{{t_{\bf b}^c}}
\newcommand{\tca}{{t_{\bf a}^c}} % Un align* a remplacer si on change a notation
\newcommand{\ssc}{{s_c}}
\newcommand{\colmod}{{\kappa}}

\DeclareMathOperator{\Tr}{Tr}

\newcommand{\ml}{M_n(\lambda)}

\newcommand{\la}{\lambda}

\newcommand{\prodtwo}[2]{\prod_{\substack{#1 \\ #2}}}
\newcommand{\sumtwo}[2]{\sum_{\substack{#1 \\ #2}}} % sum with 2 lines^M
\numberwithin{equation}{section}

\begin{document}

\title{Local asymptotic normality for finite dimensional quantum systems}

%\author{Jonas Kahn^{1} and M\u{a}d\u{a}lin Gu\c{t}\u{a}^{2}}

\author{Jonas Kahn$^{1}$ and M\u{a}d\u{a}lin Gu\c{t}\u{a}$^{2}$\\\\
$^{1}$ Universit\' e Paris-Sud 11, D\' epartement de Math\' ematiques, \\
B\^{a}t 425, 91405 Orsay Cedex, France\\\\
$^{2}$ University of Nottingham, School of Mathematical Sciences, \\University Park, NG7 2RD, Nottingham, U.K. }

%Fautes actuelles: ligne 1680

%NOTATION |MATHBF M|

%VERIFIER QUE $\delta$ EST DÉFINI.

%AJOUTER LE NECESESSAIRE POUR LAISSER Z GRANDIR LENTEMENT

%GAFFE: JE ME MARCHE DESSUS AVEC LA NOTATION 

\maketitle

\begin{abstract}
We extend our previous results on local asymptotic normality (LAN) for qubits \cite{Guta&Kahn,Guta&Janssens&Kahn} to quantum 
systems of arbitrary finite dimension $d$.   LAN means that the quantum statistical 
model consisting of $n$ identically prepared $d$-dimensional systems with joint state 
$\rho^{\otimes n}$ converges as $n\to\infty$ to a statistical model consisting of classical and quantum Gaussian variables with fixed and known covariance matrix, and unknown means related to the parameters of the density matrix $\rho$.  Remarkably, the limit model splits into a product of a classical Gaussian with mean equal to the 
diagonal parameters, and independent harmonic oscillators prepared in thermal equilibrium states displaced by an amount proportional to the off-diagonal elements. 

As in the qubits case \cite{Guta&Janssens&Kahn}, LAN is the main ingredient in devising a general two step adaptive procedure for the optimal estimation of completely unknown $d$-dimensional quantum states. 
This measurement strategy shall be described in a forthcoming paper \cite{Guta&Kahn4}.
 
\end{abstract}

\newpage 

\tableofcontents

\newpage

\section{Introduction}

Quantum statistics deals with problems of statistical inference arising in quantum mechanics. The first significant results in this area 
appeared in the seventies and tackled issues such as quantum Cram\' er-Rao bounds for unbiased estimators, optimal estimation for families of states possessing a group symmetry, estimation of Gaussian states, optimal discrimination between non-commuting states. It is impossible to list all contributions but the following references  may give the flavour of these developments \cite{Helstrom69,Yuen&Lax,Yuen&Lax&Kennedy,Belavkin2,Belavkin,Holevo}. 
The more recent theoretical advances   
\cite{Hayashi.editor,Hayashi.book,Paris.editor,Barndorff-Nielsen&Gill&Jupp,Gill&Guta&Artiles,Audenaert&Szkola} are closely related to the rapid development of quantum information and quantum 
engineering, and are often accompanied by practical implementations \cite{Mabuchi,Hannemann&Wunderlich,Smith&Silberfarb&Deutsch&Jessen,Breitenbach&Schiller&Mlynek}.
% In quantum optics a measurement method called quantum homodyne tomography \cite{Vogel&Risken,D'Ariano.2,Leonhardt.Munroe} allows the estimation with arbitrary precision \cite{Gill&Guta&Artiles,Butucea&Guta&Artiles} of the state of a monochromatic beam of light, by repeatedly measuring a sufficiently large number of identically prepared beams \cite{Smithey,Breitenbach&Schiller&Mlynek,Zavatta}.

An important topic in quantum statistics is that of optimal estimation of an unknown state using the results of measurements performed on $n$ identically prepared quantum systems \cite{Massar&Popescu,Cirac,Vidal,Gill&Massar,Keyl&Werner,Bagan&Baig&Tapia,Hayashi&Matsumoto,Hayashi&Matsumoto2,Bagan&Gill,Gill}. In the case of  two dimensional systems, or qubits, the problem has been solved explicitly in the {\it Bayesian set-up}, in the particular case of an invariant prior and figure of merit based on the fidelity distance between states \cite{Bagan&Gill}. However the method used there does not work for more general priors, loss functions, or higher dimensions. In the {\it pointwise approach}, Hayashi and Matsumoto \cite{Hayashi&Matsumoto} showed that the Holevo bound \cite{Holevo} for the variance of locally unbiased estimators can be achieved asymptotically, and provided a sequence of measurements with this property. Their results, building on earlier work \cite{Hayashi.japanese,Hayashi.conference}, 
indicate for the first time the emergence of a Gaussian limit in the problem of optimal state estimation for qubits. The extension to $d$-dimensional case is analysed by Matsumoto in \cite{Matsumoto_SUD}.

In \cite{Guta&Kahn,Guta&Janssens&Kahn} we performed a detailed analysis of this phenomenon (again for qubits), and showed that we deal with the quantum generalization of an important concept in mathematical statistics called {\it local asymptotic normality}. As a corollary, we devised a two steps adaptive measurement strategy for state estimation which is asymptotically optimal for a large class of loss functions and priors, and could be practically implemented using continuous-time measurements. In `classical statistics', the idea of approximating a sequence of statistical models by a family of Gaussian distributions appeared in \cite{Wald}, and was fully developed by Le Cam \cite{LeCam} who coined the term ``local asymptotic normality''. Among the many applications we mention its role in asymptotic optimality theory  and in proving the asymptotic normality of certain estimators such as the maximum likelihood estimator. The aim of this paper is 
to extend the results of \cite{Guta&Kahn,Guta&Janssens&Kahn} to systems of 
{\it arbitrary dimension $d<\infty$}, and thus provide the main tool for solving the {\it open problem} of optimal state estimation for $d$-dimensional quantum systems 
\cite{Guta&Kahn4}.

Before stating the main result of the paper we will explain briefly the meaning of local asymptotic normality for two dimensional systems 
\cite{Guta&Kahn,Guta&Janssens&Kahn}. We are given $n$ qubits identically prepared in an unknown state $\rho$. Asymptotic normality means that for large $n$ we can encode the statistical information contained in the state $\rho^{\otimes n}$ into a Gaussian model consisting of a classical random variable with distribution $N(u,I^{-1})$, and a quantum harmonic oscillator prepared in a (Gaussian) displaced thermal state $\Phi_{\zeta}$. The term {\it local} refers to how $\rho$ is related to the parameters $\theta=(u,\zeta)$, as explained below.
%The parameter $u\in \mathbb{R}$ is relevant for the variation in diagonal elements of 
%$\rho$, while $\zeta\in \mathbb{C}$ is related to small $SU(2)$ rotations around a fixed state. 

%for large $n$ the quantum 
%statistical model described by the joint state $\rho^{\otimes n}$ of $n$ qubits, identically prepared in an unknown state $\rho$, is close to a limit statistical model consisting of a pair of a classical Gaussian random variable $X\equiv N(u,V)$ and a quantum harmonic oscillator prepared in a Gaussian state $\Phi_{z}$, both having known variances but unknown means $(u,z)\in \mathbb{R}\times \mathbb{C}$. 

For a more precise formulation let 
us parametrise the qubit states by their Bloch vectors 
$\rho(\overrightarrow{r}) = \frac{1}{2}(\mathbf{1} + \overrightarrow{r}\overrightarrow{\sigma})$ where $\overrightarrow{\sigma} = (\sigma_{x}, \sigma_{y}, \sigma_{z})$ are the Pauli matrices. The neighborhood of the state $\rho_{0}$ with 
$\overrightarrow{r_{0}}= (0,0, 2\mu-1)$ and $1/2<\mu<1$, is a 
three-dimensional ball parametrised by the deviation $u\in\mathbb{R}$ of diagonal elements and $\zeta\in\mathbb{C}$ of 
the off-diagonal ones
 \begin{equation}\label{eq.family}
\rho_{\theta} = 
\left( 
\begin{array}{cc}
\mu +u & \zeta^{*}\\
\zeta& 1-\mu -u
\end{array}
\right),\qquad \theta= (u,\zeta)\in\mathbb{R} \times\mathbb{C}.
\end{equation}
Note that $\rho_{0}$ is to be considered fixed and known but otherwise arbitrary, and can be taken to be diagonal without any loss of generality. Consider now $n$ identically prepared qubits whose individual states are in a neighborhood of  $\rho_{0}$ of size $1/\sqrt{n}$, so that their joint state is 
$
\rho^{n}_{\theta} :=  \left[\rho_{\theta/\sqrt{n}}\right]^{\otimes n}$ for some unknown 
$\theta$. We would like to understand the structure of the family (statistical experiment)
\begin{equation}\label{eq.qn}
\mathcal{Q}_{n} := \{ \rho^{n}_{\theta} : \| \theta\| \leq C\} ,
\end{equation}
{\it as a whole}, more precisely what is its asymptotic behavior as $n\to \infty$ ?

For this we consider a quantum harmonic oscillator with position and momentum operators satisfying the commutation relations $[{\bf Q},{\bf P}]= i\mathbf{1}$. 
We denote by $\{|n\rangle , n\geq 0\}$ the eigenbasis of 
the number operator and define the thermal equilibrium state
$$
\Phi= (1-e^{-\beta}) \sum_{k=0}^{\infty} e^{-k\beta} | k\rangle\langle k|, \qquad 
e^{-\beta}= \frac{1-\mu}{\mu},
$$
which has centered Gaussian distributions for both ${\bf Q}$ and ${\bf P}$ with 
variance $ 1/(4\mu-2) >1/2$. We define a family of displaced thermal equilibrium states
\begin{equation}\label{eq.displacedthermal}
\Phi^{\zeta}: =  D^{\zeta} ( \Phi ):= 
 W( \zeta/\sqrt{2\mu-1}) \, \Phi \,  W( \zeta/\sqrt{2\mu-1})^{*},
\end{equation}
where $W(\zeta) := \exp(\zeta a^{*}- \zeta a)$ is the unitary displacement operator 
with $\zeta\in \mathbb{C}$. Additionally we consider a classical {\it Gaussian shift} 
model consisting of the family of normal distributions $N(u, \mu(1-\mu))$ with unknown center $u$ and fixed known variance. The classical-quantum statistical experiment to which we 
alluded above is 
\begin{equation}\label{eq.r}
\mathcal{R} := \{ \Phi^{\theta} : =N (u, \mu(1-\mu)) \otimes \Phi^{\zeta} : 
\|\theta\|\leq C\}
\end{equation}
where the unknown parameters $\theta =(u,\zeta)\in\mathbb{R}\times \mathbb{C}$ are the same as those of $\mathcal{Q}_{n}$.
\begin{thm}\label{th.qubits}\cite{Guta&Kahn,Guta&Janssens&Kahn}
Let $\mathcal{Q}_{n}$ be the quantum statistical experiment \eqref{eq.qn}  
and let $\mathcal{R}$ be the classical-quantum experiment \eqref{eq.r}. 
Then for each $n$ there exist quantum channels 
(normalized completely positive maps)
\begin{eqnarray*}
T_{n} &:& M\left( (\mathbb{C}^{2})^{\otimes n} \right)\to 
                L^{1}(\mathbb{R})\otimes \mathcal{T}(L^{2} (\mathbb{R})) , \\
S_{n} &: &  L^{1}(\mathbb{R})  \otimes \mathcal{T}( L^{2}( \mathbb{R})) \to 
M\left((\mathbb{C}^{2})^{\otimes n}\right),
\end{eqnarray*}
with $\mathcal{T}(L^{2}(\mathbb{R}))$ the trace-class operators, such that 
\begin{eqnarray*}\label{eq.channel.conv.}
&&
\lim_{n\to \infty}\, 
\sup_{\|\theta \|\leq C}\,
 \| \Phi_{\theta} - T_{n} \left(  \rho^{n}_{\theta}\right)  \|_{1} =0, \\
&&
\lim_{n\to \infty} \,
\sup_{\|\theta\|\leq C}
\, 
\| \rho^{n}_\theta - S_{n} \left( \Phi_{\theta} \right)  \|_{1} =0, \\
\end{eqnarray*}
for an arbitrary constant $C>0$.
 \end{thm}
The theorem shows that from a statistical point of view the joint qubits states are asymptotically indistinguishable from the limit Gaussian system. At the first sight one might object that the local nature of the result prevents us from drawing any conclusions for the original model of a completely unknown state $\rho$.
%The fact that the result holds for a local neighbourhood of a fixed state $\rho_{0}$, shrinking at rate $1/\sqrt{n}$ is not a limitation, on the contrary it shows that the interesting statistical phenomena happen at the same scaling as that of the Central Limit Theorem. The reduction from a completely unknown state to a state in a fixed local neighbourhood can be done by a rough estimation using a small number of the quantum systems, which is the first step of our adaptive measurement strategy.
However this is not a limitation, but reflects the correct normalisation of the parameters with $n\to\infty$. Indeed as $n$ grows we have more information about the state which
can be pinned down to a region of size slightly larger that $1/\sqrt{n}$ by performing rough measurements on a small proportion of the systems. After this `localisation' step, we can use more sophisticated techniques to better estimate the state within the local neighborhood of the first step estimator, and it is here where we use the local asymptotic normality result. Indeed, since locally the states are uniformly close to displaced Gaussian states we can pull back the optimal (heterodyne) measurement for estimating the latter to get an asymptotically optimal measurement  for the former. Based on this insight we have proposed a realistic measurement set-up for this purpose using 
an atom-field interaction and continuous measurements in the field \cite{Guta&Janssens&Kahn}.

This paper deals with the extension of the previous result to $d$-dimensional systems. 
Like in the two-dimensional case we parametrise the neighbourhood of a fixed 
(diagonal) state $\rho_{0}$ by a vector $\vec{u}\in\mathbb{R}^{d-1}$ of diagonal parameters and $d(d-1)/2$ complex parameters $\quant =(\zeta_{j,k}: j<k)$, one for each off-diagonal matrix element 
(cf. \eqref{rho.theta.tilde} and \eqref{rho.theta}). 
We consider the same $1/\sqrt{n}-$scaling and look at the family 
$$
\mathcal{Q}_{n} = 
\left\{ \left[\rho_{\theta/\sqrt{n}}\right]^{\otimes n} :  
\theta = (\vec{u},\vec{\zeta}) \in \Theta_{n} \subset \mathbb{R}^{d-1}\otimes \mathbb{C}^{d(d-1)/2} \right\} ,
$$
where $\Theta_{n}$ is a ball of local parameters whose size is allowed to grow slowly with $n$. 

As in the $2$-dimensional case, the limit model is the product of a classical statistical model depending on the parameters $\vec{u}$ and a quantum model depending on $\vec{\zeta}$. Moreover the quantum part splits into a tensor product of displaced 
thermal states of quantum oscillators, one for each off-diagonal matrix element $\zeta_{j,k}$ with $j<k$. Thus
$$
\Phi^{\theta} = N(\vec{u}, I_{\rho_{0}}^{-1})\otimes \bigotimes_{j<k} \Phi^{\zeta_{j,k}}_{j,k} , \qquad
\theta = (\vec{u}, \vec{\zeta}).
$$ 
Here, $I_{\rho_{0}}$ is the Fisher information matrix of the multinomial model with parameters $(\mu_{1}, \dots ,\mu_{d})$ described in Example \ref{ex.limit.diagonal}, and 
$\Phi^{\zeta_{j,k}}_{j,k}$ is the displaced thermal equilibrium state defined in 
\eqref{eq.displaced.thermal} with inverse temperature $\beta= \ln (\mu_{j}/\mu_{k})$. 

Theorem \ref{main} is the main result of the paper and shows the convergence of 
$\mathcal{Q}_{n}$ to the Gaussian model 
$$
\mathcal{R}_{n} = \left\{ \Phi^{\theta} : \theta\in \Theta_{n}\subset \mathbb{R}^{d-1}\otimes \mathbb{C}^{d(d-1)/2}\right\},
$$
in the spirit of Theorem \ref{th.qubits}. On the technical side, the uniform convergence holds over local neighbourhoods $\Theta_{n}$ which are allowed to 
grow with $n$ rather that being fixed balls.  This is essential for constructing the two stage optimal measurement: first localise within a neighbourhood $\Theta_{n}$, and 
then apply the optimal Gaussian measurement. The details of this construction are similar to the two dimensional case and will be given in a subsequent paper 
\cite{Guta&Kahn4}.

Despite the similarity to the two dimensional case, the proof of the $d$-dimensional result has additional features which may be responsible for the fact that the optimal estimation problem has remained unsolved until now. The proof is based on the following observations:
\begin{itemize}
\item
the $n$ systems space $(\mathbb{C}^{d})^{\otimes n}$ decomposes into a direct sum 
of irreducible representations of $SU(d)$, each representation being labelled by a Young diagram $\lambda$ (cf. Theorem \ref{th.sud.sn});
\item
the joint state $\rho_{\theta/\sqrt{n}}^{\otimes n}$ has the block diagonal form \eqref{prerhon}, the block weights $\lambda \to p^{\theta,n}_{\lambda}$ depend only on the diagonal parameters $\vec{u}$ and are closely related to the multinomial distribution of Example \ref{ex.limit.diagonal}. This classical statistical model converges to the $(d-1)$-dimensional Gaussian shift model $N(\vec{u}, I^{-1}_{\rho_{0}})$;
\item
there exists an isometry $V_{\lambda}$ mapping basis vectors $\left\vert {\bf m}, \lambda\right\rangle$ of the irreducible representation $\mathcal{H}_{\lambda}$ {\it almost} into number vectors $\left\vert {\bf m}\right\rangle$ of the multimode Fock space, where 
${\bf m}= \{m_{j,k}: j<k \}$ is the collection of number eigenvalues for all oscillators. 
\item 
given a typical $\lambda$, the conditional block-state $\rho^{\theta,n}_{\lambda}$ can be 
mapped with $V_{\lambda}$ into a multimode state which is close (in trace norm) to the Gaussian product state $\otimes_{j<k} \Phi_{j,k}^{\zeta_{j,k}}$. 
This can be done {\it uniformly} over the typical diagrams whose normalised shapes have $1/\sqrt{n}$ fluctuations around  $(\mu_{1}, \mu_{2}, \dots, \mu_{d})$, and over parameters $\theta\in \Theta_{n}$.
%\item
%the action of small $SU(d)$ rotations $U(\vec{\zeta})^{\otimes n}$ on the state 
%$\rho^{\theta,n}$ 
%\item
\end{itemize}

The first item is the well known Weyl duality which is extensively used in quantum statistics for i.i.d. states. The probability distribution of the second point has also been analysed the context of large deviations \cite{Keyl&Werner} for the estimation of the state eigenvalues. The third point shows that the basis $|{\bf m}, \lambda\rangle$ is almost orthogonal for indices ${\bf m}$ which are not too big. This basis is obtained by projecting tensors of the form $f_{\bf a}:= f_{a(1)}\otimes\dots \otimes f_{a(n)}$ onto a subspace of  $(\mathbb{C}^{d})^{\otimes n}$ which is isomorphic to $\mathcal{H}_{\lambda}$ (cf. Theorem \ref{th.basis.irrep}). Let us place the indices $\{a(i): i=1\dots n\}$ in the boxes of the diagram $\lambda$ along rows, starting from the left end of the first row, to obtain a tableau $t_{\bf a}$. It turns out that we only need to consider $f_{\bf a}$ for which $t_{\bf a}$ is a semistandard tableau (nondecreasing along rows, increasing along columns). Then the label ${\bf m} :=\{m_{i,j} : j>i\}$ is the collection of integers 
$m_{i,j}$ equal to the number of $j's$ on the row $i$, and is in one to one correspondence with ${\bf a}$. The following is an example of such semistandard 
tableau
$$
t_{\bf m}= \young(1111111122233,222223,333) ~, \qquad {\rm with}~ m_{1,2}=3, m_{1,3}=2,m_{2,3}=1.
$$
The relatively large number of $i$'s in the row $i$ is intentional, since it turns out that the 
`relevant' vectors, i.e. those carrying the states $\rho^{\theta,n}_{\lambda}$, 
have indices $m_{i,j}$ small compared with the length of the rows ($\lambda_{i} \approx n\mu_{i}$ for typical representations $\lambda$). More precisely, in section \ref{preuve_quasi_orth} we prove the following quasi-orthogonality result which allows us to carry the block states over to the oscillator space: if  ${\bf m} \neq {\bf l}$ and $| {\bf  l} | \leq | {\bf m} | \leq n^{\eta}$  then 
$$
\vert \langle {\bf m}, \lambda | {\bf l}, \lambda \rangle \vert = 
O(n^{(9\eta-2) |{\bf m} -{\bf l}|/12} )  \xrightarrow[n\to\infty]{} 0 \qquad {\rm for~} \eta< 2/9. 
$$

The proof of the fourth point involves a detailed analysis of the state 
$\rho^{\theta,n}_{\lambda}$ through its coefficients in the basis $\left\vert {\bf m},\lambda \right\rangle$ of $\mathcal{H}_{\lambda}$. When $\theta=0$ the state is diagonal and 
its coefficients approach uniformly those of the multidimensional thermal state 
$\Phi^{\vec{0}}=\otimes_{j<k} \Phi_{j,k} $ as shown in Lemma \ref{len0}. 
The next step is to apply $SU(d)$ rotations and obtain the states $\rho^{\theta,n}_{\lambda}$. In Lemmas \ref{ldisplacement} and \ref{lgrouplimit} it is shown that the unitary operations ${\rm Ad} [U_{\lambda}(\zeta/\sqrt{n})]$ act on $\rho^{0,n}_{\lambda}$ in the same way as the displacement operator $D^{\vec\zeta}$ acts on the thermal state $\Phi^{\vec{0}}$. A remarkable fact is that in the limit the different off-diagonal parameters `separate' into a product of shift experiments for quantum oscillators, 
one for each off-diagonal index $(j<k)$. This could be guessed from the Quantum Central Limit Theorem \ref{th.clt} which is related to the restriction of our result to $\theta=0$.

 %%%%%%%%%%%%%%%%%%%%%%%%%%

Due to the apparent intricacy of the main result, the paper is organised according to the `onion peeling' principle. We start in section \ref{given} with general classical 
statistical notions which motivate our investigation in quantum statistics. In particular we explain the relevance of the Le Cam distance between statistical models as a statistically meaningful way to describe convergence. Section \ref{idea} presents the classical version of local asymptotic normality with the multinomial model as example. 

In section \ref{sec.qlan} we introduce the quantum statistical model consisting of 
$n$ identical quantum systems with joint state $\rho^{\theta,n}$ described 
by diagonal and rotation parameters. We also introduce the multimode Gaussian states appearing in the limit. With this we can formulate the main result, Theorem \ref{main}.

In section \ref{prepreuve} 
we introduce the basis $|{\bf m}, \lambda\rangle$ and the isometry 
$V_{\lambda}$ allowing us to define the channels $T_{n}$ and $S_{n}$ connecting the two statistical models. 

In section \ref{sec.main.steps} we break the proof of the main theorem into manageable lemmas, essentially by using triangle inequalities. Each lemma deals with a different aspect of the convergence and has an interest in its own. 

Finally, the technical proofs are collected in section \ref{sec.technical.proofs}. Notably, subsection \ref{technical_tools} and Lemma \ref{lemtools} contain the combinatorial substance of the paper. 

Our investigation relies on the theory of representations of $SU(d)$. We refer to \cite{Fulton,Goodman&Wallach,Fulton&Harris} for proofs of standard results and more 
details.

As in the two-dimensional case \cite{Guta&Janssens&Kahn}, local asymptotic normality provides a two stage adaptive measurement strategy which is asymptotically optimal for both Bayesian and pointwise viewpoints, and for a large range of `distances' on the state space \cite{Guta&Kahn4}.

%%%%%%%%%%%%%%%%%%%%%%%%%%%%%%%%%%%%%%
\section{Classical and quantum statistical experiments}
\label{given}

In this section we introduce some basic notions from classical statistics with the aim of defining the Le Cam distance between statistical models and local asymptotic normality. In parallel, we will define the quantum analogues  and point out their relevance in quantum statistics. The reader may find the conceptual framework helpful in understanding the quantum version of the result, but otherwise the section can be skipped at the first reading. 

Let $X$ be a random variable with values in the measure space 
$(\mathcal{X},\Sigma_{\mathcal{X}})$, and let us assume that its probability distribution 
$P$ belongs to some family $\{ P_{\theta} : \theta\in\Theta\}$ where the parameter 
$\theta$ is unknown. Statistical inference deals with the question of how to use the available data $X$ in order to draw conclusions about some property of $\theta$. We shall call the family 
\begin{equation}\label{eq.experiment}
\mathcal{E} :=  \{ P_{\theta} : \theta\in \Theta \},
\end{equation}
a {\it statistical experiment or statistical model} over $(\mathcal{X}, \Sigma_\mathcal{X})$ 
\cite{LeCam}.

In quantum statistics the data is replaced by a quantum system prepared in a state 
$\phi$  which belongs to a family $\{ \phi_{\theta} :\theta\in \Theta\}$  of states over an algebra of observables. 
In order to make a statistical inference about $\theta$ one first has to measure the system, and then apply statistical techniques to draw conclusions from the data consisting of the measurement outcomes. An important difference with the classical case is that the experimenter has the possibility to choose the measurement set-up $M$, and 
each set-up will lead to a different classical model 
$\{ P^{(M)}_{\theta}  : \theta\in \Theta\}$, where 
$P^{(M)}_{\theta}$ is the distribution of outcomes when performing the measurement $M$ on the system prepared in state $\phi_{\theta}$. 
%Typically, the posterior state of the quantum system does not carry any more information about $\theta$, hence it is important to choose beforehand the most informative measurement for each particular decision problem. 

The guiding idea of this paper is to investigate the structure of the family of 
quantum states   
$$
\mathcal{Q} :=  \{ \phi_{\theta} : \theta\in\Theta \},
$$
which will be called a {\it quantum statistical experiment}. We shall show that in an important asymptotic set-up, namely that of a large number of identically prepared systems, the joint state can be approximated by a multidimensional quantum Gaussian state, for {\it all} possible preparations of the individual systems. This will bring a drastic simplification in the problem of optimal estimation for $d$-dimensional quantum systems, which can then be solved in the asymptotic framework \cite{Guta&Kahn4}.     

%
%In this paper however, we are not primarily interested in finding the optimal 
%measurement  decision problem but rather in the structure of the quantum statistical experiment itself: 
%$$
%\mathcal{E} :=  \{  \rho_{\theta} : \theta\in\Theta\}.
%$$

%This may be a completely unknown state of a d-dimensional system (qudit), or the joint state 
%$\rho^{\otimes n}$ of $n$ identically prepared qudits with unknown individual states.

%%%%%%%%%%%%%%%%%%%%%%%%%%%%%%%%%%%%%%%%%
\subsection{Classical and quantum randomizations}
\label{do}

Any statistical decision (e.g. estimator, test) can be seen as data processing  using a {\it Markov kernel}. 
Suppose we are given a random variable $X$ taking values in $(\mathcal{X},\Sigma_{\mathcal{X}})$ and we want to produce a `decision' $y\in \mathcal{Y}$  based on the data $X$. The space $\mathcal{Y}$ may  be for example the parameter space $\Theta$ in the case of estimation, or just the set $\{0,1\}$ in the case of testing between two hypotheses. For every value $x\in\mathcal{X}$ we choose $y$ randomly with probability distribution  given by $K_{x} (dy)$. Assuming that $K:\mathcal{X}\times \Sigma_{\mathcal{Y} }\to [0,1] $ is measurable with respect to $x$ for all fixed $A\in \Sigma_{\mathcal{Y} }$, we can regard $K$ as a map from probability distributions over $(\mathcal{X},\Sigma_{\mathcal{X}})$ to probability distributions over  $(\mathcal{Y},\Sigma_{\mathcal{Y} })$ with 
\begin{equation}\label{eq.markov.kernel}
K(P) (A) = \int K_{x}(A) P(dx), \quad A\in \Sigma_{\mathcal{Y}}.
\end{equation}
A {\it statistic} $S: \mathcal{X}\to \mathcal{Y}$ is a particular example of such a procedure, where $K_{x}$ is simply the delta measure at $S(x)$. 
%Another particular case is that of a {\it randomized statistic} $R:\mathcal{X} \times [0,1] \to \mathcal{Y}$ where an additional independent variable $U$ uniformly distributed over $[0,1]$ is used to compute the decision $R(X,U)$. 

Besides statistical decisions, there is another important reason why one would like to apply such treatment to the data, namely to summarize it in a more convenient and informative way for future purposes as illustrated in the following simple example. Consider $n$  independent identically distributed random variables $X_{1},\dots , X_{n}$ with values in $\{0,1\}$ and distribution $P_{\theta} := ( 1-\theta,\theta)$ with $\theta\in\Theta:= (0,1)$. The associated statistical experiment is
$$
\mathcal{E}_{n}:=\{ P_{\theta}^{n} : \theta\in\Theta\}.
$$
It is easy to see that $\bar{X}_{n} = \frac{1}{n} \sum_{i=1}^{n} X_{i}$ is an
unbiased estimator of $\theta$ and moreover it is a {\it sufficient statistic} for
$\mathcal{E}_{n}$, \emph{i.e.} the conditional distribution $P_{\theta}^{n}(\cdot | \bar{X}_{n} = \bar{x})$ does not depend on $\theta$! In other words the dependence on $\theta$ of the total sample $(X_{1},X_{2},\dots , X_{n})$ is completely captured by the statistic
$\bar{X}_{n}$ which can be used as such for any statistical decision problem concerning $\mathcal{E}_{n}$. If we denote by $\bar{P}^{n}_{\theta}$ the distribution of $\bar{X}_{n}$ then the experiment
$$
\bar{\mathcal{E}}_{n} = \{ \bar{P}^{n}_{\theta} : \theta\in\Theta\},
$$ 
is statistically equivalent to $\mathcal{E}_{n}$. To convince ourselves that $\bar{X}_{n}$ does contain the same statistical information as $(X_{1},\dots , X_{n})$, we show that we can obtain the latter from the former by means of a randomized statistic. 
Indeed for every fixed value $\bar{x}$ of $\bar{X}_{n}$ there exists a measurable function
$$
f_{\bar{x}}: [0,1] \to \{0,1 \}^{n},
$$
such that  the distribution of $f_{\bar{x}} (U)$ is 
$P^{n}_{\theta}( \cdot| \bar{X}_{n} = \bar{x})$. In other words
$$
\lambda( f_{\bar{x}}^{-1} (x_{1} ,\dots , x_{n}) ) =
P^{n}_{\theta}(  x_{1}, \dots ,x_{n} | \bar{X}_{n} = \bar{x}),
$$
where $\lambda$ is the Lebesgue measure on $[0,1]$. Then
$
F(\bar{X}_{n}, U) := f_{\bar{X}_{n}} (U),
$
has distribution $P_{\theta}^{n}$. To summarize, statistics, randomized statistics and  Markov kernels, are ways to transform the available data for a specific purpose. The Markov kernel $K$ defined in  \eqref{eq.markov.kernel} maps the  experiment 
 $\mathcal{E}$ of equation \eqref{eq.experiment} into the experiment
$$
\mathcal{F}:= \{ Q_{\theta} :  \theta\in \Theta\} ,
$$
over $(\mathcal{Y}, \Sigma_\mathcal{Y})$ with $Q_{\theta} = K(P_{\theta})$. For mathematical convenience it is useful to represent such transformations in terms of linear maps between linear spaces. 
%We will show that for ``all practical purposes'' we can work with 
\begin{defin}\label{def.stochoperator}
A positive linear map
$$
T_{*} : L^{1} (\mathcal{X}, \Sigma_{\mathcal{X}} , P)\to 
L^{1}(\mathcal{Y}, \Sigma_{\mathcal{Y}} , Q)
$$
is called a {\it stochastic operator} or {\it transition} if $\| T_{*} (g)\|_{1} = \|g \|_{1}$ for every $g\in L_{+}^{1}(\mathcal{X})$.
\end{defin}
\begin{defin}\label{def.markovop}
A positive linear map
$$
T: L^{\infty} (\mathcal{Y}, \Sigma_{\mathcal{Y}} , Q)
\to L^{\infty}(\mathcal{X}, \Sigma_{\mathcal{X}} , P)
$$
is called a {\it Markov operator} if
$T \mathbf{1} =\mathbf{1}$, and if for any $f_{n}\downarrow 0$ in $L^{\infty}(\mathcal{Y}) $ we have $T f_{n}\downarrow 0$.
\end{defin}
A pair $(T_{*}, T)$ as above is called a dual pair if
$$
\int f T (g) dP = \int T_{*}(f) g dQ,
$$
for all $f\in L^{1}(\mathcal{X},\Sigma_{\mathcal{X}} , P)$ and $g\in L^{\infty}(\mathcal{Y},\Sigma_{\mathcal{Y}} , Q)$. It is a theorem that for any stochastic operator $T_{*}$ there exists a unique dual Markov operator $T$ and vice versa.

What is the relation between  Markov operators and Markov kernels ? Roughly 
speaking, any Markov kernel defines a Markov operator when we restrict to families 
of dominated probability measures.
% (see Theorem \ref{th.markov.kernel.randomization}). 
Let us assume that all distributions $ P_{\theta}$ of the experiment 
$\mathcal{E}$ defined in \eqref{eq.experiment} are absolutely continuous with respect to a fixed probability distribution  $P$, such that there exist densities 
$p_{\theta} : = dP_{\theta}/ dP : \mathcal{X}\to \mathbb{R}_{+}$. Such an experiment is called {\it dominated} and in concrete situations this condition is usually satisfied.  Let $K_{x}(dy)$ be a Markov kernel \eqref{eq.markov.kernel} 
such that $Q_{\theta}= K(P_{\theta})$, then we define associated Markov operator 
$(T(f))(x) := \int f(y) k_{x}(dy)$ and have
\begin{equation}\label{eq.markov.operator.randomization}
Q_{\theta} = P_{\theta}\circ T, \qquad \forall\theta.
\end{equation}

%Moreover, $P$ can be chosen to be a countable 
%convex combination of $P_{\theta}$'s and such that 
%$\{P_{\theta}\}\sim P$, i.e. for any $A\in\Sigma$, 
%$P(A) =0$ if and only if $P_{\theta} (A)=0$ for all $\theta$ 
%(see Lemma 20.3 in \cite{Strasser}).
%\begin{thm}\label{th.markov.kernel.randomization}\cite{Strasser}
%Let $\mathcal{E}: = ( P_{\theta}: \theta\in \Theta)$ be an experiment over $(\mathcal{X}, \Sigma_\mathcal{X})$ and $\mathcal{F}: = ( Q_{\theta}: \theta\in \Theta )$ an experiment over $(\mathcal{Y}, \Sigma_\mathcal{Y})$ and suppose that both $\mathcal{E}, \mathcal{F}$ are dominated, with  $\{P_{\theta}\}\sim P$ and $\{Q_{\theta}\}\sim Q$. 
%Let $K$ be the Markov kernel defined in \eqref{eq.markov.kernel} such that 
%$$
%K(P_{\theta}) = Q_{\theta}.
%$$
% Then there exists a Markov operator 
%$T: L^{\infty} (\mathcal{Y}, \Sigma_{\mathcal{Y}} , Q)\to L^{\infty}(\mathcal{X}, \Sigma_{\mathcal{X}} , P)$ such that 
%\begin{equation}\label{eq.markov.operator.randomization}
%Q_{\theta} = P_{\theta}\circ T, \qquad \forall\theta.
%\end{equation}
%\end{thm}
When the probability distributions of two experiments are related to each other as in \eqref{eq.markov.operator.randomization}, we say that $\mathcal{F}$ is a {\it randomization} of $\mathcal{E}$. From the duality between $T$ and $T_{*}$ we obtain an equivalent characterization in terms of the stochastic operator 
$T_{*} : L^{1} (\mathcal{X}, \Sigma_{\mathcal{X}} , P)\to L^{1}(\mathcal{Y}, \Sigma_{\mathcal{Y}} , Q) $ such that 
$$
T_{*}( dP_{\theta}/dP) = dQ_{\theta}/dQ, \qquad \forall\theta\, .
$$
The concept of randomization is weaker than that of Markov kernel transformation, but  under the additional condition that $(\mathcal{Y}, \Sigma_{\mathcal{Y}})$ is locally compact space with countable base and Borel $\sigma$-field, it can be shown that any randomization can be implemented by a Markov kernel \cite{Strasser}.

What is the analogue of randomizations in the quantum case ? In the language of operator algebras $L^{\infty}(\mathcal{X},\Sigma_{\mathcal{X}} , P)$ is a commutative von Neumann algebra and $L^{1}(\mathcal{X},\Sigma_{\mathcal{X}} , P)$ is the space 
of (densities of) {\it normal} linear functionals on it. The stochastic operator $T_{*}$ is the classical version of {\it quantum channel}, \emph{i.e.} a completely positive normalized (trace-preserving) map 
$$
T_{*} : \mathcal{A}_{*}\to \mathcal{B}_{*} 
$$ 
where $\mathcal{A}_{*},\mathcal{B}_{*}$ are the spaces of normal states on the von Neumann algebra $\mathcal{A}$ and respectively $\mathcal{B}$. Any normal state $\phi$ on $\mathcal{A}$ has a density $\rho$ with respect to the trace such that $\phi (A)= \mathrm{Tr}(\rho A)$ for all $A\in\mathcal{A}$. The dual of $T_{*}$  is
$$
T : \mathcal{B}\to \mathcal{A},
$$
which is a unital completely positive map and has the property that 
$T_{*}(\phi)(b) = \phi(T(b))$ for all $b\in\mathcal{B}$ and $\phi\in\mathcal{A}_{*}$. 
We interpret such quantum channels as possible physical transformations from input to output states.

A particular class of channels is that of measurements. In this case the input is the state of a quantum system described by an algebra $\mathcal{A}$, and the output  is a probability distribution over the space of outcomes 
$(\mathcal{X}, \Sigma_\mathcal{X})$. Any measurement is described by a positive linear map
$$
M : L^{\infty}(\mathcal{X}, \Sigma_\mathcal{X}, P) \to \mathcal{A},
$$
which is completely specified by the image of characteristic functions of measurable 
sets, also called {\it positive  operator valued measure} (POVM). This map $M :\Sigma_{\mathcal{X}} \to \mathcal{A}$  has following properties
\begin{enumerate}
\item
Positive: $M(A) \geq 0, \qquad \forall \,A\in\Sigma_{\mathcal{X}}$ ;
\item
Countably additive: $\sum_{i=1}^{\infty} M(A_{i}) = M(\cup_{i} A_{i}) , \quad A_{i}\cap A_{j} =\emptyset , i\neq j$;
\item
Normalized: $M(\mathcal{X}) =\mathbf{1}$.
\end{enumerate}
The corresponding channel acting on states is a positive map 
$M_{*}:\mathcal{A}_{*} \to L^{1}(\mathcal{X}, \Sigma_{\mathcal{X}},P)$ given by
$$
M(\phi) (A) = \phi (M(A)) =\mathrm{Tr}(\rho M(A)),
$$
where $\rho$ is the density matrix of $\phi$.
By applying the channel $M$ to the quantum statistical experiment  consisting of the family of states $\mathcal{Q} = (\phi_{\theta} :\theta\in\Theta)$ on $\mathcal{A}$  we obtain a classical statistical experiment
$$
\mathcal{Q}_{M} :=\{ M(\phi_{\theta})  :\theta\in\Theta \} ,
$$
over the outcomes space $(\mathcal{X} ,\Sigma_{\mathcal{X}})$. 

As in the classical case, quantum channels can be seen as ways to compare quantum experiments. The first steps in this direction were made by Petz \cite{Petz86,Petz&Jencova,Ohya&Petz} who developed the theory of 
{\it quantum sufficiency} dealing with the problem of characterizing when a sub-algebra of observables contains the same statistical information about a family of states, as the original algebra. More generally, two experiments 
$ \mathcal{Q}:=\{\mathcal{A}, \phi_{\theta}:\theta\in\Theta\}$ and 
$\mathcal{R}:=\{ \mathcal{B} , \psi_{\theta} : \theta\in\Theta\}$ are called {\it statistically equivalent} if there exist channels $T:\mathcal{A}\to \mathcal{B}$ and $S:\mathcal{B}\to \mathcal{A}$ such that 
$$
\psi_{\theta} \circ T= \phi_{\theta} \qquad {\rm and }\qquad 
\phi_{\theta} \circ S= \psi_{\theta} \qquad \forall \theta.
$$

As consequence, for any measurement 
$M: L^{\infty}(\mathcal{X} ,\Sigma_{\mathcal{X}}, P )\to \mathcal{A}$ there exists a measurement $T\circ M: L^{\infty}(\mathcal{X} ,\Sigma_{\mathcal{X}}, P ) \to \mathcal{B}$ such that the resulting classical experiments coincide $\mathcal{Q}_{M}= \mathcal{R}_{T\circ M}$. Thus for any statistical problem, and any procedure concerning the experiment $\mathcal{Q}$ there exists a procedure for $\mathcal{R}$ with the 
same risk (average cost), and vice versa. 

%%%%%%%%%%%%%%%%%%%%%%%%%%%%%%%%%%%%%%%%%%%%

\subsection{The Le Cam distance and its statistical meaning}
\label{Dexp}

%A central feature of the latter relation \eqref{app} is that it depends only on the experiments (what we are given), and not on the specific statistical problem (what we aim at). It gives a kind of distance between experiments, and (approximately) solves at one stroke in one experiment all the precise statistical problems we know how to solve in the other. 

%But we were starting from a rather special case when the experiments were defined on the same space and $\| p_{\theta} - s_{\theta} \|_1 \leq \epsilon$ for all $\theta$. We may go further. 

We have seen that two experiments are statistically equivalent when they can be transformed into each other be means of quantum channels. When this cannot be 
done exactly, we would like to have a measure of how close the two experiments 
are when we allow any channel transformation. We define the \emph{deficiency} of 
$ \mathcal{R}$  with respect to $ \mathcal{Q}$ as 
\begin{equation}\label{eq.deficiency}
\delta(\mathcal{R}, \mathcal{Q}) = \inf_{T} \sup_{\theta} \| \phi_{\theta} - \psi_{\theta}  \circ T\|
\end{equation}
where the infimum is taken over all channels $T:\mathcal{A}\to\mathcal{B}$. 
The norm distance between two states on $\mathcal{A}$ is defined as 
$$
\| \phi_{1} - \phi_{2}\| :=\, \sup \{ |\phi_{1}(a) -\phi_{2}(a)| : a\in\mathcal{A} , \|a\|\leq 1\},
$$
and for $\mathcal{A}= \mathcal{B}(\mathcal{H})$ it is equal to $\|\rho_{1}-\rho_{2}\|_{1}:=\mathrm{Tr}(|\rho_{1}-\rho_{2}|)$, where $\rho_{i}$ is the density matrix of the state $\phi_{i}$.
When 
$\delta(\mathcal{R}, \mathcal{Q})=0$ we say that $\mathcal{R}$ is more informative 
than $\mathcal{Q}$. Note that $\delta(\mathcal{R}, \mathcal{Q})$ is not symmetric but satisfies a triangle inequality of the form 
$\delta(\mathcal{R}, \mathcal{Q}) +\delta(\mathcal{Q}, \mathcal{T}) \geq \delta(\mathcal{R}, \mathcal{T})$.  By symmetrizing we obtain a proper distance over the space 
of equivalence classes of experiments, called Le Cam's distance \cite{LeCam}
\begin{equation}\label{eq.lecam.distance}
\Delta(\mathcal{Q}, \mathcal{R}): = \rm{max} \left( \delta(\mathcal{Q}, \mathcal{R}) \, ,\,\delta(\mathcal{R}, \mathcal{Q}) \right).
\end{equation}

What is the statistical meaning of the Le Cam distance ? We shall show that if 
$\delta(\mathcal{R} , \mathcal{Q}) \leq \epsilon$ then for any statistical decision 
problem with loss function between $0$ and $1$, any measurement procedure for 
$\mathcal{Q}$ can be matched by a measurement procedure for $\mathcal{R}$ 
whose risk will be at most $\epsilon$ larger than the previous one.

A decision problem is specified by a {\it decision space} $(\mathcal{X},\Sigma_{\mathcal{X}})$ and a {\it loss function} $W_{\theta}: \mathcal{X}\to [0,1]$ for each 
$\theta\in\Theta$. We are given a quantum system prepared in the state 
$\phi_{\theta}\in\mathcal{A}_{*}$ with unknown parameter $\theta\in\Theta$ and would like to perform a measurement with outcomes in $\mathcal{X}$ such that the expected value of the loss function $W_{\theta}$ is small. Let
$$
M: L^{\infty} (\mathcal{X}, \Sigma_{\mathcal{X}} ,P) \to \mathcal{A},
$$
be such a measurement,  and $P^{(M)}_{\theta} = \phi_{\theta}\circ M $, then  the {\it risk} at $\theta$ is
$$
R(M,\theta):=  \int_{\mathcal{X}} W_{\theta}(x)P^{(M)}_{\theta}(d x).
$$

Since the point $\theta$ is unknown one would like to obtain a small risk over all possible  realizations
\[
R_{max} (M) = \sup_{\theta\in \Theta} R(M,\theta).
\] 
The {\it minimax risk} is then $R_{minmax}:= \inf_{M} R_{max} (M)$.  
In the Bayesian framework one considers a prior distribution $\pi$ over $\Theta$ 
and then averages the risk with respect to $\pi$
$$
R_{\pi} (M) = \int_{\Theta} R(M,\theta) \pi(d\theta).
$$
The optimal risk in this case is $R_{\pi} := \inf_{M} R_{\pi} (M)$.

Coming back to the experiments $\mathcal{Q}$ and $\mathcal{R}$ we shall compare 
their achievable risks for a given decision problem as above. Consider the measurement $N:L^{\infty}(\mathcal{X} , \Sigma_{\mathcal{X}} ,P) \to \mathcal{B} $ given by 
$N= T\circ M$ where $T: \mathcal{A}\to \mathcal{B}$ is the channel which achieves 
the infimum in \eqref{eq.deficiency}. Then 
\begin{eqnarray*}
R(N,\theta)& = &\int_{\mathcal{X}} W(\theta,x )P^{(N)}_{\theta}(d x) = 
\psi_{\theta} (T \circ M (W_{\theta}))  \\
&\leq& 
\|\psi_{\theta}\circ T - \phi_{\theta}\| + 
\phi_{\theta}(M(W_{\theta}))\leq
\delta(\mathcal{R}, \mathcal{Q}) + R(M,\theta),
\end{eqnarray*}
where we have used the fact that $0\leq W_{\theta}\leq 1$.
\begin{lem}\label{lemma.risk.deficiency}
For every achievable risk $R(M, \theta)$ for $\mathcal{Q}$ there exists a 
measurement $N:L^{\infty}(\mathcal{X} , \Sigma_{\mathcal{X}} ,P) \to \mathcal{B}$ 
for $\mathcal{R}$ such that
$$
R(N, \theta) \leq R(M, \theta) + \delta (\mathcal{R}, \mathcal{Q}).
$$
In consequence
$$
R_{minmax}(\mathcal{R}) \leq  R_{minmax}(\mathcal{Q}) +  
\delta (\mathcal{R}, \mathcal{Q}).
$$
\end{lem}

\section{Local asymptotic normality in statistics}
\label{idea}

In this section we describe the notion of local asymptotic normality and its significance
in statistics \cite{LeCam,Torgersen,Strasser,vanderVaart}. Suppose that we observe 
$X_{1}, \dots ,X_{n}$ where $X_{i}$ take values in a measurable space 
$(\mathcal{X}, \Sigma_{\mathcal{X}})$ and are are independent, identically distributed with distribution $P_{\theta}$ indexed by a parameter $\theta$ belonging to an open subset $\Theta\subset\mathbb{R}^{m}$.
The full sample is a single observation from the product $P_{\theta}^{n}$ of
$n$ copies of $P_{\theta}$ on the sample space $(\Omega^{n}, \Sigma^{n})$.
Local asymptotic normality means that for large $n$ such statistical experiments can be approximated by Gaussian experiments after a suitable reparametrisation. 
Let $\theta_{0}$ be a fixed point and define a local parameter
$u=\sqrt{n} (\theta-\theta_{0})$ characterizing points in a small neighbourhood 
of $\theta_{0}$, and rewrite $P_{\theta}^{n}$ as
$P_{\theta_{0}+ u/\sqrt{n}}^{n}$ seen as a distribution depending on the
parameter $u$. Local asymptotic normality means that for large $n$ the experiments
$$
\left\{P_{\theta_{0} + u/\sqrt{n}}^{n} : u \in \mathbb{R}^{m}\right\}
\qquad {\rm and} \qquad
\left\{ N( u, I_{\theta_{0}}^{-1}) : u \in \mathbb{R}^{m}\right\},
$$
have the same statistical properties when the models $\theta\mapsto P_{\theta}$ are sufficiently  `smooth'.
The point of this result is that while the original experiment may be difficult to analyse, the limit one is a tractable {\it Gaussian shift} experiment in which we observe a single sample from the normal distribution with unknown mean $u$ and fixed variance 
matrix $I_{\theta_{0}}^{-1}$. Here
$$
\left[I_{\theta_{0}}\right]_{ij} =
\mathbb{E}_{\theta_{0}}
\left[ \ell_{\theta_{0}, i} \ell_{\theta_{0}, j}
 \right],
$$
is the Fisher information matrix at $\theta_{0}$, with $\ell_{\theta,i} := \partial\log p_{\theta}/\partial\theta_{i}$  the score function and 
$p_{\theta}$ is the density of $P_{\theta}$ with respect to a reference 
probability distribution $P$.

There exist two formulations of the result depending on the notion of convergence which one uses. In this paper we only discuss the {\it strong} version based on convergence with respect to the Le Cam distance, and we refer to 
\cite{vanderVaart} for another formulation using the so called weak convergence (convergence in distribution of finite dimensional marginals of the likelihood ratio process), and to \cite{Guta&Jencova} for its generalization to quantum statistical experiments.

Before formulating the theorem, we explain what sufficiently smooth means. The least restrictive condition is that $p_{\theta}$ is {\it differentiable in quadratic mean}, \emph{i.e.} there exists a measurable function $\ell_{\theta}:\mathcal{X}\to\mathbb{R}$ such that 
as $u\to 0$
$$
\int \left[ 
p_{\theta+u}^{1/2} - p_{\theta}^{1/2} - u^{t} \ell_{\theta} p_{\theta}^{1/2} 
\right]^{2} dP \to 0.
$$ 
Note that $\ell_{\theta}$ must still be interpreted as score function since under some regularity conditions we have $\partial p_{\theta}^{1/2} / \partial \theta_{i} = \frac{1}{2}(\partial \log p_{\theta}/\partial \theta_{i}) p^{1/2}_{\theta}$.
\begin{thm}\label{th.lan}
Let $\mathcal{E}:= \{P_{\theta} : \theta \in \Theta \}$ be a statistical experiment with 
$\Theta \subset \mathbb{R}^d$ and $P_{\theta}\ll P$ such that the map $\theta\to p_{\theta}$  is differentiable in quadratic mean. Define
\begin{align*}
\mathcal{E}_n = \{P_{\theta_{0} +u/\sqrt{n} }^{n} :  \|u\|\leq C  \} , \qquad
\mathcal{F} = \{ N(u, I_{\theta_{0}}^{-1} ) :  \| u\|\leq C  \},
\end{align*}
with $I_{\theta_{0}}$ the Fisher information matrix of 
$\mathcal{E}$ at point $\theta_{0}$, and $C$ a positive constant. 
Then $ \Delta(\mathcal{E}_n,\mathcal{F}) \to 0$. In other words, there exist sequences of randomizations $T_n$ and $S_n$ such that:
\begin{align*}
\lim_{n\to\infty}\, \sup_{\|u\|\leq C} \,\left\|
T_n(P_{\theta_{0} +u/\sqrt{n}}^{n}) - N(u, I_{\theta_{0}}^{-1}) \right\| & = 0,
\\
\lim_{n\to\infty}\,\sup_{\| u \|\leq C} \,\left\|
P_{\theta_{0} +u/\sqrt{n}}^{n} - S_n( N(u, I_{\theta_{0}}^{-1})) \right\|  & = 0.
\end{align*}
\end{thm}

%Notice the renormalization of the parameter in $\mathcal{E}_n$. We are dealing with $\theta/\sqrt{n}$. Asymptotically, our knowledge of parameters scales like $\sqrt{n}$, as is most usual in statistics. The consequence of this is that local asymptotic normality is exactly that, \emph{local}. We shall see in section \ref{largerange} how to go to global.

%\noindent{\bf Remark.} 
\begin{remark}
{\rm
Note that the statement of the Theorem is not of Central Limit  type which typically involves convergence {\it in distribution} to a Gaussian distribution at a {\it single} point $\theta_{0}$. Local asymptotic normality states that the convergence is {\it uniform} around the point $\theta_{0}$, and moreover the variance of the limit Gaussian is fixed whereas the variance obtained from the Central Limit Theorem depends on the point $\theta$. Additionally, the randomization transforming the 
data $(X_{1},\dots, X_{n})$ into the Gaussian variable is the same for all 
$\theta=\theta_{0} + u/\sqrt{n}$ and thus does not require  \emph{a priori} the knowledge of $\theta$.
}
\end{remark}

%\noindent{\bf Remark.} 
\begin{remark}
{\rm 
Local asymptotic normality is the basis of many important results in asymptotic optimality theory and explains the asymptotic normality of certain estimators such as the maximum likelihood estimator. The quantum version introduced in the next section plays a similar role for the case of quantum statistical model. An asymptotically optimal estimation strategy based on local asymptotic normality was derived in \cite{Guta&Janssens&Kahn} for two-dimensional systems. 
}
\end{remark}

\begin{remark}\label{rem.l2.space}
{\rm
Let us define the real Hilbert space 
$L^{2}(\theta_{0})= ( \mathbb{R}^{m}, (\cdot, \cdot)_{\theta_{0}} )$ with inner product
$$
(u,v)_{\theta_{0}} =  u^{T} I_{\theta_{0}} v.
$$
By multiplying with $I_{\theta_{0}}$ we see that limit experiment can be equivalently chosen to be $N(I_{\theta_{0}}u , I_{\theta_{0}})$. 
The characteristic function of $X\sim N( I_{\theta_{0}} u , I_{\theta_{0}})$ is 
\begin{equation}\label{eq.characteristic.gaussian.shift}
F_{u}(w):= \mathbb{E}_{\theta_{0}}[ \exp( i w^{T} X )] =
\exp\left(-\frac{1}{2} \| w\|_{\theta_{0}}^{2} + i (w, u)_{\theta_{0}} \right).
\end{equation}
A similar expression will be encountered in section \ref{sec.qlan} for the case of 
quantum Gaussian shift experiment.
}
\end{remark}

\begin{ex}\label{ex.limit.diagonal} Let $P_{\mu} = (\mu_{1},\dots , \mu_{d})$ be a probability distribution with unknown parameters $(\mu_{1},\dots, \mu_{d-1})\in \mathbb{R}^{d-1}_{+}$ satisfying $\mu_{i}>0$ and $\sum_{i\leq d-1} \mu_{i}<1$.
The Fisher information at a point $\mu$ is 
\begin{equation}\label{eq.Fisher.info.classic}
I(\mu)_{ij} = \sum_{k=1}^{d-1} \mu_{k} ( \delta_{ik} \mu_{i}^{-1}\cdot \delta_{jk}\mu_{j}^{-1}) +  (1- \sum_{l=1}^{d-1}\mu_{l})^{-1}= \delta_{ij} \mu_{i}^{-1} +  (1- \sum_{l=1}^{d-1}\mu_{l})^{-1},
\end{equation}
and its inverse is 
\begin{equation}\label{eq.covariance.diagonal}
V(\mu)_{ij} :=  [I(\mu)^{-1}]_{ij} = \delta_{ij} \mu_{i}  -\mu_{i}\mu_{j}.
\end{equation}
Thus the limit experiment in this case is $\mathcal{F}:= (N(u, V(\mu)) : u \in\mathbb{R}^{d-1},  \|u\| \leq C)$. 
\end{ex}
This experiment will appear again in Theorem \ref{main}, as the classical part of the limit Gaussian shift experiment. 
%%%%%%%%%%%%%%%%%%%%%%%%%%%%%%%%%%%%%
%
%{\bf This may not be needed:} Let us consider as loss function the square of the 
%$\ell^{2}$ distance $\| \mu-\nu\|_{2}^{2} = \sum_{i\leq d} (\mu_{i} -\nu_{i})^{2}$, then in the limit experiment this corresponds to 
%$$
%W(u, v) = \sum_{i=1}^{d-1} (u_{i} - v_{i})^{2} + (\sum_{i=1}^{d-1} (u_{i}-v_{i}))^{2}. 
%$$ 
%The optimal estimator of $u$ for this loss function is the data itself 
%$\hat{u}:=X\sim N(u, V(\mu))$ and the risk is independent of $u$
%\begin{equation}\label{risk.classical.gaussian}
%R= \sum_{i=1}^{d-1} \mu_{i} (1-\mu_{i}) + \sum_{i=1}^{d-1} \mu_{i} (1-\mu_{i}) -
%\sum_{1\leq i\neq j \leq d-1} \mu_{i}\mu_{j}= \sum_{i=1}^{d} \mu_{i} (1-\mu_{i}),
%\end{equation}
%where the last sum contains $d$ terms and we used the fact that $\mu_{d} = 1-\sum_{i\leq d-1} \mu_{i}$.

 %%%%%%%%%%%%%%%%%%%%%%%%%%%%%%%%%%%%%
 
\section{Local asymptotic normality in quantum statistics}
\label{sec.qlan}
%%%%%%%%%%%%%%%%%%%%%%%%%%%%%%%%%%%%%%

In this section we present the main result of the paper. Local asymptotic normality for 
$d$-dimensional quantum systems means roughly the following: the sequence 
$\mathcal{Q}_{n}$ of experiments consisting of joint states $\rho^{\otimes n}$ of $n$ identical quantum systems prepared independently in the same state $\rho$, converges  to a limit experiment $\mathcal{R}$ which is a quantum-classical Gaussian model involving displaced thermal equilibrium states of $d(d-1)/2$ oscillators and a $(d-1)$-dimensional classical Gaussian shift model.  As in the classical case, the result has a local nature reflecting the $1/\sqrt{n}$ rate of convergence of state estimation. A neighbourhood of a fixed diagonal state 
$\rho_{0}={\rm Diag}(\mu_{1}, \dots, \mu_{d})$ is parametrised 
by (changes in the) diagonal parameters $\vec{u}\in \mathbb{R}^{d-1}$ and off-diagonal parameters $\vec{\zeta}\in \mathbb{C}^{d(d-1)/2}$. The latter can be implemented by small unitary rotations. The  limit Gaussian model has a classical part 
$N(\vec{u}, V(\mu))$ with fixed known variance  $V(\mu)$, and a quantum part 
$\otimes_{j<k} \Phi_{j,k}^{\zeta_{j,k}}$ with each  $\Phi_{j,k}^{\zeta_{j,k}}$ being a thermal equilibrium state with $\beta_{j,k} = \ln(\mu_{j}/\mu_{k})$, 
displaced in phase space by an amount proportional to $\zeta_{j,k}$. 

The reason for choosing the above parametrisation is twofold. Firstly, it unveils the important separation between 'classical' and 'quantum' parameters, and the further separation among the different off-diagonal parameters. Secondly, it is very convenient for the proof. However as we will see in ???, the limit experiment can be formulated in a `coordinate-free' way in terms of quasifree states on $CCR$-algebra. Although it is 
not needed in the main theorem, we include this formulation linking our result to the 
Quantum Central Limit Theorem. We stress again that local asymptotic normality is not a consequence of the Central Limit Theorem, indeed the latter is not even an ingredient in the proof but gives an indication as to what is the limit state when all parameters are zero.

%%%%%%%%%%%%%%%%%%%%%%%%%%%%%%%%%%%%%%%%%%%
\subsection{The $n$-tuple of $d$-dimensional systems}
%%%%%%%%%%%%%%%%%%%%%%%%%%%%%%%%%%%%%%%%%%

As explained in section \ref{idea} for the classical case, our theory will be local in nature, so we will be interested in a (shrinking) neighbourhood of an arbitrary but fixed faithful state
\begin{equation}
\label{rho0}
\rho_0 = 
\begin{bmatrix} 
\mu_1 & 0 & \dots & 0
\\
0 & \mu_2 & \ddots& \vdots 
\\
\vdots & \ddots & \ddots & 0 \\
0 & \dots & 0 &\mu_d 
\end{bmatrix} 
\qquad \qquad \mathrm{with }\ \mu_1 > \mu_2 > \dots > \mu_d > 0,
\end{equation}
which for technical reasons is chosen to have different eigenvalues. A  sufficiently small neighbourhood of $\rho_{0}$ in the state space can be 
parametrised by $\theta := (\vec{u}, \vec{\zeta})$ as follows
%Up to the second order in  $\theta / \sqrt{\delta }$, they are of the form:
\begin{equation}\label{rho.theta.tilde}
\tilde{\rho}_{\theta} 
:=
\begin{bmatrix} 
\mu_1 + u_1 & \zeta_{1,2}^* & \dots & \zeta_{1,d}^*
\\
\zeta_{1,2} & \mu_2  + u_2 & \ddots& \vdots 
\\
\vdots & \ddots & \ddots & \zeta_{d-1,d}^* \\
\zeta_{1,d} & \dots & \zeta_{d-1,d} &\mu_d - \sum_{i=1}^{d-1} u_i
\end{bmatrix},
\qquad   u_i\in \mathbb{R}, ~ \zeta_{j,k}\in \mathbb{C}.
\end{equation}
Indeed, note that if $\theta$ is small enough then $\tilde{\rho}_{\theta}$ is a density matrix. 

Let $\delta:= \inf_{1\leq i \leq d} \mu_i - \mu_{i+1}$, with $\mu_{d+1} = 0$, be the separation between the eigenvalues. In the first order in $\theta / \sqrt{\delta}$, the family $\tilde{\rho}_{\theta}$ is obtained by first perturbing the diagonal elements of $\rho_{0}$ with $\vec{u}$ and then performing a small unitary transformation with 
\begin{equation}
U(\quant)  := 
\exp \left[ i \left( 
  \sum_{1\leq j<k\leq d} \frac{{\rm Re}(\quantsub_{j,k}) T_{j,k} 
+ {\rm Im}(\quantsub_{j,k}) T_{k,j}}{\sqrt{\mu_j - \mu_k }}\right)  \right] \\
\end{equation}
where $T_{j,k}$ are generators of the Lie algebra of $SU(d)$ defined in \eqref{generators_algebra}. The advantage of the latter parametrisation is that we can fully exploit the machinery of irreducible group representations. %(see section \ref{Sud}). 
For this reason, in all subsequent computations we will work with the `unitary' family 
\begin{equation}\label{rho.theta}
\rho_{\theta} 
:= U(\vec{\zeta}) 
\begin{bmatrix} 
\mu_1 + u_1 & 0 & \dots & 0
\\
0 & \mu_2  + u_2 & \ddots& \vdots 
\\
\vdots & \ddots & \ddots & 0 \\
0 & \dots & 0 & \mu_d - \sum_{i=1}^{d-1} u_i
\end{bmatrix} U^*(\vec{\zeta}),
\qquad   u_i\in \mathbb{R},~ \zeta_{j,k}\in \mathbb{C}.
\end{equation}
but we keep in mind the relationship with \eqref{rho.theta.tilde}.

As in the classical case, the parameter $\theta$  will be scaled by the factor 
$1/\sqrt{n}$ meaning that we zoom in around $\rho_{0}$ with the rate equal to the 
typical estimation rate based on $n$ samples. 
Let $\rho^{\theta,n}:= \rho_{\theta/\sqrt{n}}^{\otimes n}$ and let $\mathcal{Q}_{n}$ be the sequence of statistical experiments 
\begin{equation}
\label{Q_exp}
\mathcal{Q}_{n} := \left\{\rho^{\theta,n} :\theta\in\Theta_{n}\right\},
\end{equation}
consisting of $n$ systems, each one prepared in a state $\rho_{\theta/\sqrt{n}}$ situated in a local neighborhood of $\rho_{0}$. 
The local parameter $\theta=(\overrightarrow{u},\overrightarrow{\zeta})$ 
belongs to a neighborhood $\Theta_{n}$ of the origin of 
$\mathbb{R}^{d-1} \times \mathbb{C}^{d(d-1)/2}$ which is allowed to grow slowly 
with $n$ in a way that will be made precise later. 

One of the principal tools in our result is the representation theory of the special unitary group  $SU(d)$. Due to lack of space we shall not include any proofs and refer to \cite{Fulton,Goodman&Wallach,Fulton&Harris} for details. In particular we will be working 
with the well known tensor representation which will be analysed in increasing depth across the following sections.

%which we review in Appendix \ref{Sud}. 
The space $(\mathbb{C}^{d})^{\otimes n}$ carries two commuting group 
representations: that of  $SU(d)$ given by 
\begin{equation}\label{eq.unitary.rep.tensor}
\pi_{n}(U): |\psi_{1}\rangle \otimes \dots  \otimes |\psi_{n}\rangle
\mapsto U|\psi_{1}\rangle \otimes \dots  \otimes U|\psi_{n}\rangle, \qquad U\in SU(d),
\end{equation}
and that of the permutation group $S(n)$ given by 
\begin{equation}\label{eq.unitary.rep.permutations}
\tilde{\pi}_{d}(\tau):  |\psi_{1} \rangle \otimes \dots \otimes |\psi_{n}\rangle \mapsto 
|\psi_{\tau^{-1}(1)}\rangle \otimes\dots \otimes |\psi_{\tau^{-1}(n)}\rangle, 
\qquad \tau \in S(n).
\end{equation}
Since the two group representations commute with each other, 
the representation space decomposes into a direct sum of tensor products 
of irreducible representations. It turns out that the irreducible representations 
of $SU(d)$ and $S(n) $ are indexed by {\it Young diagrams} with $d$ rows for the former 
and $n$ boxes for the latter. A Young diagram is defined by a tuple of ordered integers 
$\lambda = (\lambda_{1}  \geq \lambda_{2} \dots\geq  \lambda_{k})$ with $\lambda_{i}$ the number of boxes on row $i$ (see Figure \ref{fig.young.diagram}).
\begin{figure}[h!]
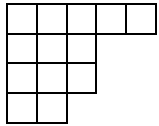

\begin{displaymath}
\yng(5,3,3,2) 
\end{displaymath}
\caption{Young diagram with $\lambda=(5,3,3,2)$.}
\label{fig.young.diagram}
\end{figure}
As we shall see later this pictorial representation will be very useful in understanding 
the structure of the irreducible representations $(\mathcal{H}_{\lambda}, \pi_{\lambda})$ of $SU(d)$. 

The following theorem called {\it Schur-Weyl duality} shows that the only tensor 
products appearing in the above mentioned direct sum  are those of irreducible representations indexed by the same $\lambda$, and in particular  the algebras  
generated by $\pi_{n}(u)$ and  respectively $\tilde{\pi}_{d}(\tau)$ are each other's commutant!

% as characterised in Theorem \ref{th.commutant.algebra}.
\begin{thm}\label{th.sud.sn}
Let $\pi_{n}$ and $\tilde{\pi}_{d}$ be the representations of $SU(d)$ and respectively 
$S(n)$ on $(\mathbb{C}^{d})^{\otimes n}$. Then the representation space decomposes into a direct sum of tensor products of irreducible representations of $SU(d)$ and 
$S(n)$ indexed by Young diagrams with $d$ lines and $n$ boxes:
\begin{eqnarray*}
(\mathbb{C}^{d})^{\otimes n} &\cong& 
\bigoplus_{\lambda}
\mathcal{H}_{\lambda}\otimes \mathcal{K}_{\lambda} ,\\
\pi_{n} &\equiv& 
\bigoplus_{\lambda} 
\pi_{\lambda} \otimes \mathbf{1}_{\mathcal{K}_{\lambda}},\\
\tilde{\pi}_{d} &\equiv& 
\bigoplus_{\lambda} 
\mathbf{1}_{\mathcal{H}_{\lambda}} \otimes \tilde{\pi}_{\lambda}.
\end{eqnarray*}
\end{thm}
In particular $\rho^{\theta,n}= \rho_{\theta/\sqrt{n}}^{\otimes n} $ and $\tilde{\pi}_d(\tau)$ commute for all $\tau$.
Hence we have the block diagonal form for the joint states
\begin{equation}
\label{prerhon}
\rho^{\theta, n} = \bigoplus_{\lambda} 
p^{\theta,n}_{\lambda} \rho^{\theta, n}_{\lambda } \otimes 
\frac{\mathbf{1}_{\mathcal{K}_{\lambda}}}{M_{n}(\lambda) },
\end{equation}
where $M_{n}(\lambda) $ is the dimension of $\mathcal{K}_{\lambda}$, 
$p^{\theta,n}_{\lambda}$ is a probability distribution over the Young diagrams, and 
$\rho^{\theta, n}_{\lambda }$ is a density matrix on $\mathcal{H}_{\lambda}$. 
From \eqref{rho.theta} and the Schur-Weyl duality, we get the expression of the block states  
\begin{equation}\label{eq.block.states}
\rho^{\theta,n}_{\lambda} =U_{\lambda}(\vec{\zeta}/\sqrt{n})\, \rho^{u, 0,n}_{\lambda} \,U_{\lambda}(\vec{\zeta}/\sqrt{n})^{*}.
\end{equation}

We interpret the decomposition \eqref{prerhon} as follows: by doing a `which block' measurement we obtain information about $\theta$ through the probability density 
$p^{\theta,n}_{\lambda}$. In fact it is easy to see that $p^{\theta,n}_{\lambda}$  does not depend on $\vec{\zeta}$, so it only gives information about the diagonal parameters 
$\vec{u}$. Later on we shall see that the model $p^{\theta,n}$ has the same limit as the classical multinomial model described in  Example \ref{ex.limit.diagonal}. Once this information has been obtained, one still possesses a conditional quantum state $\rho^{\theta,n}_{\lambda}$. It turns out that this state carries information about the rotation parameters $\vec{\zeta}$, and we will show that the statistical model described by the conditional state converges to a `purely quantum' Gaussian shift experiment.

%%%%%%%%%%%%%%%%%%%%%%%%%%%%%%%%%%%%
\subsection{Displaced thermal equilibrium states of a harmonic oscillator}
\label{sec.displaced.thermal.eq.states}
%%%%%%%%%%%%%%%%%%%%%%%%%%%%%%%%%%%%%%
The ground state of a quantum harmonic oscillator or the laser state of a monochromatic light pulse are well known examples of quantum Gaussian states. Both physical systems are described by the same algebra of observables generated 
by the canonical `position' and `momentum' observables ${\bf Q}$ and ${\bf P}$  satisfying the Heisenberg commutation relation
\begin{equation}\label{eq.heisenber.comm.}
{\bf Q}{\bf P}- {\bf P}{\bf Q} = i \mathbf{1}.
\end{equation}

%We will now give a brief description of the (unique up to unitary equivalence) irreducible 
%regular representation of $CCR(\mathbb{C})$ and of its Gaussian states. 
%The symplectic form on $\mathbb{C}$, seen as a two dimensional real space, is $\sigma(z,z^{\prime})=\mathrm{Im}(\bar{z}z^{\prime})$. The canonical coordinates are 
%${\bf Q}=B(1/\sqrt{2})$ and ${\bf P}= B(i/\sqrt{2})$ and hence satisfy 
%\begin{equation}\label{eq.heisenber.comm.}
%{\bf Q}{\bf P}- {\bf P}{\bf Q} = i \mathbf{1}.
%\end{equation}
%The ground state of a quantum harmonic oscillator or the laser state of a monochromatic light pulse are well known examples of quantum Gaussian states. Both physical systems are described by the same algebra of observables generated 
%by the canonical `position' and `momentum' observables ${\bf Q}$ and ${\bf P}$  satisfying the Heisenberg commutation relation
%\begin{equation}\label{eq.heisenber.comm.}
%{\bf Q}{\bf P}- {\bf P}{\bf Q} = i \mathbf{1}.
%\end{equation}
These observables can be represented on the Hilbert space $L^{2}(\mathbb{R})$ as 
\begin{equation}
({\bf Q}\psi)(x) = x\psi(x),\qquad
({\bf P}\psi)(x) = -i \frac{d\psi}{dx}(x), \qquad \psi\in L^{2}(\mathbb{R}).
\end{equation} 
The space $L^{2}(\mathbb{R})$ has a special orthonormal basis 
$\{  \left\lvert  0 \right\rangle,  \left\lvert  1 \right\rangle, \dots\}$ with the vector 
$ \left\lvert  m \right\rangle$ given by 
$$
%\psi_{k}(x) := 
H_{m}(x) e^{-x^{2}/2} / (\sqrt{\pi}2^{m}m!)^{1/2},
$$ 
where $H_{m}$ are the Hermite polynomials. These are the eigenvectors of the number operator ${\bf N}:= \frac{1}{2}({\bf Q}^{2}+{\bf P}^{2}-\mathbf{1})$ counting the number of `excitations' of the oscillator or the number of photons in the case of the light beam, such that ${\bf N}  \left\lvert  m \right\rangle=m  \left\lvert  m \right\rangle$. 

The creation and annihilation operators 
$$
{\bf a}^{*}= ({\bf Q}- i{\bf P} )/\sqrt{2}, \qquad {\bf a}=  ({\bf Q}+ i{\bf P} )/\sqrt{2}
$$ 
satisfy $ [{\bf a}, {\bf a}^{*}]=\mathbf{1}$ and act as `ladder' operators on the basis 
$\psi_{k}$: 
\begin{eqnarray*}
{\bf a} \left\lvert  m \right\rangle= \sqrt{m}  \left\lvert  m-1 \right\rangle,\qquad
{\bf a}^{*} \left\lvert  m \right\rangle = \sqrt{m+1}  \left\lvert  m +1\right\rangle.
\end{eqnarray*}
In particular the following identity holds: ${\bf N}={\bf a}^{*}{\bf a}$.

It can be easily checked that both ${\bf Q}$ and ${\bf P}$ have Gaussian distribution with respect to the vacuum state $ \left\lvert  0 \right\rangle$. In fact they are `jointly Gaussian' 
$$
\left\langle 0 \right\vert \exp(iu{\bf Q} + iv{\bf P}) | \left\vert 0\right\rangle= 
\exp\left(-\frac{1}{4}(u^{2}+v^{2})\right).
$$
We will often use the complex form of the unitary {\it Weyl operators} 
$$
W(z):= \exp( z{\bf a}^{*}- \bar{z} {\bf a})= \exp(ip_{0}{\bf Q} - iq_{0} {\bf P}) , \qquad 
z= (q_{0}+i p_{0})/\sqrt{2}\in \mathbb{C}, 
$$
which satisfy the Weyl relations
$$
W(z)^{*}W(z^{\prime})W(z) = \exp\left( 2i {\rm Im}(\bar{z}^{\prime}z) \right)W(z^{\prime}). 
$$ 
The coherent (vector) states $\left\vert z\right\rangle $ are obtained by displacing the vacuum state with Weyl operators 
\begin{equation}
\left\vert z\right\rangle := W(z)\left\vert 0\right\rangle = \exp(-|z|^{2}/2) \sum_{m=0}^{\infty} \frac{z^{m}}{\sqrt{m!}} 
\left\vert m\right\rangle.
\end{equation}
They are Gaussian states with the same variance as the vacuum, and means  
$\left\langle z\right\vert {\bf Q} \left\vert z\right\rangle =  \sqrt{2}{\rm Re}(z)$ and 
$\left\langle z\right\vert {\bf P} \left\vert z\right\rangle =  \sqrt{2}{\rm Im}(z)$:
$$
\left\langle z\right\vert W(z^{\prime}) \left\vert z\right\rangle = 
\exp\left(-\frac{1}{2}  |z-z^{\prime}|^{2} +2i {\rm Im}(\bar{z}^{\prime}z) \right).
$$

Besides, coherent states, an important role in our discussion will be played by the thermal equilibrium states. For every $\beta>0$ we define the Gaussian state
\begin{equation}\label{eq.phi.beta}
\phi_{\beta}(W(z))= \exp\left(-\frac{|z|^{2}}{2\tanh (\beta/2)}\right).
%\quad {\rm or}\quad
%\phi_{\beta}\left(\exp(i( u{\bf Q} + v{\bf P}) )\right) = 
%\exp\left(-\frac{u^{2}+v^{2}}{4 \tanh(\beta/2) }\right) 
\end{equation}
Its density matrix consisting of a mixture of $k$-photon states with geometrical weights
\begin{equation}\label{eq.thermal.state}
\Phi_{\beta} =  (1-e^{-\beta}) \sum_{k=0}^{\infty} e^{-k \beta } \left\vert k\right\rangle\left \langle k \right \vert. 
\end{equation}
and can also be obtained by 'smearing' the coherent states with a Gaussian kernel:
\begin{equation}\label{phi_integrale_gaussienne}
\Phi_{\beta}= \frac{e^{\beta}-1}{\pi} \int_{\mathbb{C}} \exp\left(-(e^{\beta}-1) |z|^{2}\right) 
\left\vert z\right\rangle \left\langle z\right\vert dz .
\end{equation}
The thermal equilibrium states can be shifted in `phase space' by means 
of  displacement operations $D^{z}$ which act by adjoining with unitaries 
$ W(z)$, i.e. 
$$
D^{z}(\cdot):= {\rm Ad}[W(z)](\cdot) = W(z)^{*}\cdot W(z). 
$$
The result is a Gaussian state $\phi^{z}_{\beta}$ with the same variance as $\phi_{\beta}$ and the same means as $\vert z\rangle$:
\begin{equation}\label{eq.displacement.Fock}
\phi^{z}_{\beta} (W(z^{\prime})) :=
%\phi_{\beta} (D^{z}(W( z^{\prime}) )  )=
\exp\left(- \frac{|z|^{2}}{2\tanh (\beta/2)}+ 2i{\rm Im}(\bar{z}^{\prime}z )\right) ,
\qquad
\Phi_{\beta}^{z} :=  D^{z} (\Phi_{\beta}) :=W(z)^{*} \Phi_{\beta} W(z).
\end{equation}
%%%%%%%%%%%%%%%%%%%%%%%%%%%%%%%%%%%%%%%
\subsection{The multimode Fock space and the limit Gaussian shift experiment}
%%%%%%%%%%%%%%%%%%%%%%%%%%%%%%%%%%%%%%%%%
\label{subsec.multimode.fock}
We now consider $d(d-1)/2$ commuting harmonic oscillators, with a joint state consisting of independent Gaussian states. Let us define the {\it multimode Fock space}
$$
\mathcal{F} := \bigotimes_{ 1\leq j<k\leq d}  L^{2}(\mathbb{R}),
$$
in which we identify the number basis
\begin{equation}\label{eq.def.m}
\left\vert {\bf m} \right\rangle =  \bigotimes_{j<k} \left\vert m_{j,k} \right\rangle, 
\qquad {\bf m}= \left\{ m_{j,k} \in \mathbb{N}: j<k \right\}.
\end{equation}
For each of the oscillators we define the thermal equilibrium state
\begin{equation}\label{eq.phijk}
\Phi_{j,k} := \Phi_{\beta_{j,k}} , \qquad \beta_{j,k} = \ln(\mu_{j}/\mu_{k}),
\end{equation}
where $\{\mu_{1},\dots ,\mu_{d}\}$ are the eigenvalues of the density matrix 
$\rho_{0}$ (cf. \eqref{rho0}). We now use the Weyl operators to displace these states by an amount proportional to the off-diagonal elements $\zeta_{j,k}$ of $\rho^{\theta}$ 
(cf. \eqref{rho.theta.tilde} and \eqref{rho.theta})
$$
\Phi_{j,k}^{\zeta_{j,k}}:= 
W\left(\frac{\zeta_{j,k}}{2\sqrt{\mu_{j} - \mu_{k} }}\right)^{*} \,
\Phi_{j,k} \,
W\left(\frac{\zeta_{j,k}}{2\sqrt{\mu_{j} - \mu_{k} }}\right).
$$ 
We now define the joint state $\phi^{\vec{\zeta}}$ of the oscillators 
with density matrix 
\begin{equation}\label{eq.phi.vec.zeta.}
 \Phi^{\vec{\zeta}}= \bigotimes_{j<k} \Phi_{j,k}^{\zeta_{j,k}} \in \mathcal{T}_{1}(\mathcal{F}),
\end{equation}
where $\mathcal{T}_{1}(\mathcal{F})$ is the space of trace-class operators on 
$\mathcal{F}$.

The states $\Phi^{\vec{\zeta}}$ form the quantum part of the limit Gaussian experiment. The classical part is identical to the $(d-1)$-dimensional Gaussian shift model 
$N(\vec{u} , V(\mu))$ of Example \ref{ex.limit.diagonal}, where 
$\mu=\{\mu_{1}, \dots, \mu_{d}\}$.
\begin{defin}
On the algebra $L^{\infty}(\mathbb{R}^{d-1}) 
\otimes \mathcal{B}(\mathcal{F})$ we define normal state $\phi^{\theta}$ 
with density 
\begin{equation}\label{eq.phi.theta}
\Phi^{\theta} := \mathcal{N}(\vec{u} , V(\mu))\otimes  \Phi^{\vec{\zeta}} \in L^{1}(\mathbb{R}^{d-1}) \otimes \mathcal{T}_{1}(\mathcal{F}),
\end{equation}
where $\mathcal{N}(\vec{u} , V(\mu))$ is the Gaussian density of Example \ref{ex.limit.diagonal}. The quantum-classical Gaussian experiment $\mathcal{R}$ is 
defined by 
$$
\mathcal{R} = \{ \Phi^{\theta} : 
\theta =(\vec{u}, \vec{\zeta})\in \mathbb{R}^{d-1}\times \mathbb{C}^{d(d-1)/2}\}.
$$
\end{defin}

\subsection{The main theorem}
%%%%%%%%%%%%%%%%%%%%%%%%%%%%%%%%%%%%%

We are now ready to formulate the main result of the paper. In view of subsequent application to optimal state estimation, it is essential to consider (slowly) growing domains of the local parameters. For given $\beta,\gamma>0$ we define
\[
\Paramglob_{n,\beta , \gamma } =  \left\{ (\quant, \clas) :  \lVert \quant \rVert_{\infty} \leq n^{\beta } ,  \left\lVert \clas \right\rVert_{\infty} \leq n^{\gamma }\right\} .
\]

Recall that $\delta $ is the separation between the eigenvalues of $\rho_0$ given by equation \eqref{rho0}. Though we use parametrisation \eqref{rho.theta} for density matrices $\rho_{\glob}$, recall that in the first order this is approximated by $\tilde{\rho}_{\theta}$ defined in \eqref{rho.theta.tilde}. In fact it can be shown that the same theorem holds for the latter parametrisation.

\begin{thm}
\label{main}
Let  $\delta > 0$, let $\beta < 1 / 9$ and $\gamma < 1/4$.
Let the quantum experiments 
\begin{equation*}
\mathcal{Q} _n  =  \left\{ \rho^{\glob ,n} : \glob\in \Paramglob_{n, \beta , \gamma } \right\} , \qquad
\mathcal{R}_{n}  = \left\{ \Phi^{\glob} : \glob \in  \Paramglob_{n, \beta , \gamma } \right\},
\end{equation*}
where $\rho^{\glob ,n } = \rho_{\glob / \sqrt{n}}^{\otimes n}$ is the state on 
$M\left((\mathbb{C} ^{d})^{\otimes n}\right)$ given by equation \eqref{rho.theta}, and 
$\Phi^{\glob}$ is given by 
\eqref{eq.phi.theta}.

Then, 
%if $n > n_0 / \delta ^{k}$, with $n_0$ and $k$ depending only on  $\beta$ and $\gamma$,  
there exist channels (completely positive, normalised maps) 
\begin{eqnarray}
T_n &:& 
M(\mathbb{C} ^{d})^{\otimes n} \to L^{1}(\mathbb{R}^{d-1}) \otimes \mathcal{T}_{1}(\mathcal{F})\\
S_n &:& L^{1}(\mathbb{R}^{d-1}) \otimes \mathcal{T}_{1}(\mathcal{F}) \to M(\mathbb{C} ^{d})^{\otimes n}
\end{eqnarray}
with  $\mathcal{T}_{1}(\mathcal{F})$ is the space of trace-class operators on $\mathcal{F}$, such that 
\begin{align}
\label{Tn}
\sup_{\glob \in \Paramglob_{n,\beta , \gamma } }  \left\lVert \Phi^{\glob} - T_n(\rho^{\glob, n}) \right\rVert_1  = O (n^{- \epsilon } / \delta) , \\ 
\label{Sn}
\sup_{\glob \in \Paramglob_{n,\beta , \gamma } }  \left\lVert S_n(\Phi^{\glob}) - \rho^{\glob, n} \right\rVert_1  =O (n^{- \epsilon } / \delta) ,
\end{align}
where $\epsilon > 0$ depends only on $\delta$, $\beta$ and $\gamma$. In particular we have
$$
\lim_{n\to \infty}\Delta(\mathcal{Q} _n, \mathcal{R} _n) =0,
$$ 
where $\Delta(\cdot, \cdot)$ is the Le Cam distance defined in 
\eqref{eq.lecam.distance}. 
\end{thm}

In other words, we get polynomial speed of convergence of the approximation, which is enough to build two-step evaluation strategies in the finite experiments globally asymptotically equivalent to strategies in the limit experiment \cite{Guta&Kahn4}. 
\subsection{The relation between LAN and CLT}
%%%%%%%%%%%%%%%%%%%%%%%%%%%%%%%%%%
One way to think of local asymptotic normality is the following: we would like to 
understand the asymptotic behaviour of the collective (fluctuation) observables 
\eqref{eq.fluctuation.obs} with respect to a {\it whole neighborhood} of the state $\rho$, how the limit distribution changes as we change the reference state $\rho^{\otimes n}$. 

The quantum Central Limit Theorem describes the asymptotic behaviour of the 
same observables with respect to a {\it fixed} state, and is one of the ingredients in the proof of a different version of  LAN based on {\it weak convergence} 
\cite{Guta&Jencova}. However, in the case of strong convergence, which is the object of this paper, CLT does not play any role since we are interested in convergence in norm rather than in distribution, and uniformly over a range of parameters. 

The purpose of the section is to derive a `coordinate free' version of the limit Gaussian experiment using the Central Limit Theorem and the notion of symmetric logarithmic derivative. The reader interested in the proof of main theorem can skip the following pages and  continue with section~\ref{prepreuve}.

%%%%%%%%%%%%%%%%%%%%%%%%%%%%%%%%%%%%%%%%%
\subsubsection{Quantum Central Limit Theorem}
%%%%%%%%%%%%%%%%%%%%%%%%%%%%%%%%%%%%%%%%%%%

Let $\rho$ be a fixed faithful state on $M(\mathbb{C}^{d})$. To $\rho$ we associate an algebra of canonical commutation relations carrying a Gaussian state $\phi$. The Quantum Central Limit Theorem \cite{Petz} says that $\phi$ is the limit distribution of certain multi-particle observables with respect to of product states $\rho^{\otimes n}$.

Let 
$$
(A, B)_{\rho} := \mathrm{Tr}(\rho\, A\circ B), \qquad {\rm where~} A\circ B := \frac{AB+ BA}{2},
$$
be a positive inner product on the real linear space of {\it selfadjoint operators} 
$M(\mathbb{C}^{d})_{sa}$. We define the Hilbert space with inner product $(\cdot, \cdot)_{\rho}$.
$$
L^{2}(\rho)= \{ A\in M(\mathbb{C}^{d})_{sa} : {\rm Tr}(A\rho)=0\}. 
$$

Let $\sigma$ be the {\it symplectic form} on $L^{2}(\rho)$
$$
\sigma(A,B) = \frac{i}{2}\mathrm{Tr}(\rho\, [A, B]).
$$
The $C^{*}$-algebra of canonical commutation relations 
$CCR(L^{2}(\rho), \sigma)$ is generated by the Weyl operators $W(A)$ satisfying the relations
$$
W(A)^{*} = W(-A), \qquad W(A)W(B) = W(A+B)\exp(-i\sigma(A,B)), \quad A,B \in 
L^{2}(\rho).
$$
On $CCR(L^{2}(\rho), \sigma)$ we define the Gaussian (quasifree) state 
\begin{equation}\label{eq.quasifree}
\phi (W(A)) := \exp\left(-\frac{1}{2} \| A\|_{\rho}^{2}\right), \qquad \| A\|_{\rho}^{2}= (A,A)_{\rho}.
\end{equation}
The state $\phi$ is regular, i.e. there exists a representation $(\pi, \mathcal{H})$ 
of the algebra $CCR(L^{2}(\rho), \sigma)$ such that the one parameter family 
$t\mapsto \pi(W(tA))$ is weakly continuous and $\phi$ is a normal state on the von Neumann algebra generated by $\pi(CCR(L^{2}(\rho), \sigma))$. 
This means that there exist selfadjoint 'field operators' $B(A)$ such that 
$\pi(W(tA)) = \exp(it B(A))$, and there exists a density matrix 
$\Phi_{\pi} \in \mathcal{T}_{1}(\mathcal{H})$ such that
$$
\phi (W(A)) = \mathrm{Tr} \left( \exp(iB(A)) \Phi_{\pi} \right) , \qquad A\in L^{2}(\rho).
$$
The representation $(\pi,\mathcal{H})$ can be obtained through the GNS construction, or by `diagonalising' the CCR algebra as we will see in  a moment. From \eqref{eq.quasifree} we deduce that the distribution of $B(A)$ with respect to $\phi$ is a centred normal distribution with variance $\|A\|_{\rho}^{2}$.  From the Weyl relations it follows that the fields satisfy the following canonical commutation relations
$$  
[B(A), B(C)] = 2i\sigma(A, C)\mathbf{1}, \qquad A,C \in L^{2}(\rho).
$$ 
 
 Consider now the tensor product  $\bigotimes_{k=1}^{n}M(\mathbb{C}^{d})$ which is generated by elements of the form
\begin{equation}\label{eq.xk}
A^{(k)} = \mathbf{1}\otimes \dots \otimes A \otimes \dots \otimes \mathbf{1},
\end{equation}
with $A$ acting on the $k$-th position of the tensor product. We are interested in the asymptotics as $n\to\infty$ of the joint distribution under the state $\rho^{\otimes n}$, of `fluctuation' elements of the form
\begin{equation}\label{eq.fluctuation.obs}
F_{n}(A) :=\frac{1}{\sqrt{n}} \sum_{k=1}^{n} A^{(k)}.
\end{equation}
\begin{thm}{\bf [Quantum CLT]}\label{th.clt}
Let $A_{1}, \dots , A_{s}\in L^{2}(\rho)$. Then the following holds
\begin{eqnarray*}
&&
\lim_{n\to\infty} {\rm Tr}
\left(\rho^{\otimes n} \left(\prod_{l=1}^{s} F_{n}(A_{l}) \right)\right) =
\phi \left( \prod_{l=1}^{s}\left( B(A_{l}) \right)\right),\\
&&
\lim_{n\to\infty} {\rm Tr} 
\left( \rho^{\otimes n} \left( \prod_{l=1}^{s} \exp( iF_{n}(A_{l}) ) \right)\right) =
\phi\left( \prod_{l=1}^{s} W( A_{l} )  \right).
\end{eqnarray*}
\end{thm}

Although the algebra $CCR(L^{2}(\rho), \sigma)$ may look rather abstract, its structure can be easily understood by `diagonalising' it. Let us assume that $\rho$ is a diagonal matrix $\rho_{0}= {\rm Diag}(\mu_{1}, \dots , \mu_{d})$.
%\begin{lem}
The Hilbert space $L^{2}(\rho_{0})$ decomposes as direct sum of orthogonal subspaces
$\mathcal{H}_{\rho_{0}} \oplus \mathcal{H}_{\rho_{0}}^{\perp} $ where
\begin{equation}\label{eq.orthog.decomp}
\mathcal{H}_{\rho_{0}} := 
%\left\{ A\in L^{2}(\rho_{0}, \mathbb{R}) : [A,\rho_{0}]=0\right\} = 
{\rm Lin} \{ A: [A,\rho_{0}] =0 , {\rm Tr}(A\rho_{0}) =0\},
\quad
{\rm and} 
\quad
\mathcal{H}_{\rho_{0}}^{\perp}=  {\rm Lin} \{ T_{j,k} , j\neq k\},
\end{equation}
with $T_{j,k}$  the generators of the $\mathfrak{su}(d)$ algebra defined in \eqref{generators_algebra}. 

The elements $W(A)$ with $A\in \mathcal{H}_{\rho_{0}}$ generate the center of the 
algebra which is isomorphic to the algebra  of bounded continuous functions $C_{b}(\mathbb{R}^{d-1})$. Explicitly, we identify the coordinates in 
$\mathbb{R}^{d-1}$ with the basis 
$\{ d_{i}= -\mu\mathbf{1} + E_{i,i}: i=1,\dots d-1\} $ of $\mathcal{H}_{\rho_{0}}$,  
(see \eqref{generators_algebra} for the definition of $E_{i,i})$. Then the covariance matrix for the basis vectors is 
$$
( d_{i}, d_{j})_{\rho_{0}} = 
{\rm Tr}(\rho_{0} d_{i} d_{j} ) = \delta_{i,j}\mu_{i} - \mu_{i}\mu_{j} = 
[V(\mu)]_{i,j}, 
$$
where $V_{\mu}$ is the covariance matrix \eqref{eq.covariance.diagonal}. 

Moreover 
\begin{equation}\label{eq.symplectic.basis}
t_{j,k}:= T_{j,k}/\sqrt{2(\mu_{j}-\mu_{k} )}, \qquad j\neq k,
\end{equation}
form an {\it orthogonal and symplectic basis} of 
$\mathcal{H}_{\rho_{0}}^{\perp}$, i.e.
$$
\sigma(t_{j,k} , t_{k,j}) = -1/2  ,\quad  j<k ,\quad {\rm and}~ ~ \sigma(t_{j,k} , t_{l,m})= 0 \quad {\rm for ~} \{ j,k\} \neq \{l,m\}.
$$
which means that $\{t_{j,k}, t_{k,j}\} $ generate isomorphic algebras of quantum harmonic oscillator which we denote by $CCR(\mathbb{C})$. 
From 
$$
\| t_{j,k}\|_{\rho_{0}}^{2} = {\rm Tr}(\rho_{0} t_{j,k}^{2}) = \frac{\mu_{j} + \mu_{k}}{2(\mu_{j}-\mu_{k})}
$$
and \eqref{eq.phi.beta} we conclude that each of the oscillators is prepared independently in the thermal equilibrium state $\phi_{j,k}= \phi_{\beta_{j,k}}$ with $\beta_{j,k}=\ln (\mu_{j}/\mu_{k})$.

%In conclusion the state $\phi$ decomposes as

% \begin{equation}\label{eq.decomp.algebra}
%CCR(L^{2}(\rho_{0}), \sigma) \cong 
%C_{b} (\mathbb{R}^{d-1})\otimes \bigotimes_{j<k} CCR(\mathbb{C}).
%\end{equation}
Based on the discussion of sections \ref{sec.displaced.thermal.eq.states} and 
\ref{subsec.multimode.fock} we can choose 
$\mathcal{H}:=L^{2}(\mathbb{R}^{d-1})\otimes \mathcal{F}$  and define the regular representation $\pi$ of $CCR(L^{2}(\rho_{0}), \sigma)$ on this space in a straightforward way and its von Neumann completion is $L^{\infty}(\mathbb{R}^{d-1})\otimes \mathcal{B}(\mathcal{F})$. The  state $\phi$ decomposes as
\begin{equation}\label{eq.decom.state}
\phi \cong N( 0, V_{\mu} ) \otimes \bigotimes_{j<k} \phi_{j,k}.
\end{equation}
which is precisely the state $\phi^{\theta}$ for 
$\theta=(\vec{u},\vec{\zeta}) =(\vec{0},\vec{0})$, defined in \eqref{eq.phi.theta}.

\subsubsection{The quantum Gaussian shift experiment through Fisher information}

We complete the family of states $\phi^{\theta}$ of the experiment $\mathcal{R}$ 
by shifting $\phi^{{0}}$ with the help of symmetric logarithmic derivatives. 
As in the classical case, this will be a family of Gaussian states with the same 
covariance, and mean proportional to the local parameter $\theta$. The covariance is related to the Fisher information matrix as described in Remark \ref{rem.l2.space}. Thus we will start by defining the quantum analogues of the score functions and the Fisher information matrix for the full quantum model $\rho_{\theta}$.

Let us define 
the {\it symmetric logarithmic derivatives}  \cite{Helstrom,Holevo} as the solutions in 
$L^{2}(\rho_{0})$ of
$$
\mathcal{L}_{j,k}^{(re)} \circ \rho_{0} = 
\left.\frac{\partial \rho_{\theta}}{\partial \, {\rm Re}\zeta_{j,k}} \right|_{\theta=0} ,
\quad 
\mathcal{L}_{j,k}^{(im)} \circ \rho_{0} = 
\left.\frac{\partial \rho_{\theta}}{\partial \, {\rm Im}\zeta_{j,k}} \right|_{\theta=0},
\quad
\ell_{i}\circ \rho_{0}=  
\left.\frac{\partial \rho_{\theta}}{\partial u_{i}} \right|_{\theta=0}, 
$$
Then with $H_{j,k},E_{i,i}$ defined in \eqref{generators_algebra}
$$
\mathcal{L}_{j,k}^{(re)} = H_{k,j}/(\mu_{j}+\mu_{k}) , 
\quad \mathcal{L}_{j,k}^{(im)} =H_{j,k}/ (\mu_{j}+\mu_{k}) , \quad
\ell_{i} =  E_{i,i}/\mu_{i} - E_{d,d}/\mu_{d}, 
$$
and the quantum Fisher information matrix consists of a `classical block' that 
coincides with that of the classical multinomial model in \eqref{eq.Fisher.info.classic}
$$
[I_{\rho_{0}} ]_{ij}:= (\ell_{i}, \ell_{j})_{\rho_{0}}= [I(\mu)]_{ij},  \qquad 1\leq i,j\leq d-1,
$$
and a `purely quantum' block given by the diagonal matrix
$$
H_{\rho_{0}} =  
{\rm Diag}\left( \| \mathcal{L}_{j,k}\|_{\rho_{0}}^{2},\| \mathcal{L}_{k,j}\|_{\rho_{0}}^{2}   : j<k \right)= {\rm Diag}\left(  (\mu_{j}+ \mu_{k})^{-1} ,  (\mu_{j}+ \mu_{k})^{-1} : j<k \right).
$$
%Let 
%$$
%\mathcal{L}(\theta):= \sum_{j<k} 
%\left( {\rm Re}(\zeta_{j,k}) \mathcal{L}^{(re)}_{j,k}+ {\rm Im}(\zeta_{j,k}) \mathcal{L}^{(im)}_{j,k}\right) + 
%\sum_{i} u_{i} \ell_{i}. 
%$$

\begin{lem}
Let 
$$
\mathcal{L}(\theta):= 
\sum_{j<k} 
\left( {\rm Re}(\zeta_{j,k}) \mathcal{L}^{(re)}_{j,k}+ {\rm Im}(\zeta_{j,k}) \mathcal{L}^{(im)}_{j,k}\right) + 
\sum_{i} u_{i} \ell_{i}, \qquad \theta=(\vec{u}, \vec{\zeta}). 
$$
Consider the representation $(\pi,\mathcal{H})$ of $CCR(L^{2}(\rho_{0}) , \sigma)$ and the normal state $\phi$ on $L^{\infty}(\mathbb{R}^{d-1})\otimes \mathcal{B}(\mathcal{F})$ as defined in the previous section (cf. \ref{eq.decom.state}).  Let $\tilde{\phi}^{\theta}$ be the state defined by
\begin{equation}\label{eq.def.phi.theta}
\phi^{\theta} (W(A)) := 
\exp\left(-\frac{1}{2} \| A,A\|_{\rho_{0}} +i (A,\mathcal{L}(\theta) )_{\rho_{0}}\right) ,
\qquad A\in L^{2}(\rho_{0}).
\end{equation}
Then $\tilde{\phi}^{\theta}$ is normal with respect to the representation 
$(\pi, \mathcal{H})$ and coincides with $\phi^{\theta}$ (cf. \eqref{eq.phi.theta}).

%The quantum Gaussian shift experiment is defined as the family of states on 
%$CCR(L^{2}(\rho_{0}),\sigma)$
%$$
%\mathcal{R}:= 
%\left\{  \phi^{\theta} : 
%\theta= (\vec{u},\vec{\zeta}) \in \mathbb{R}^{d-1}\times \mathbb{R}^{d(d-1)/2}
%\right\}.
\end{lem}

%{\bf Remark.} 
\begin{remark}
{\rm 
The expression \eqref{eq.def.phi.theta} is clearly the quantum analogue of the characteristic function of the classical Gaussian shift experiment \eqref{eq.characteristic.gaussian.shift}. Note in particular that the distribution of $B(A)$ with respect to $\phi_{\theta}$ is the normal with variance $\|A\|_{\rho_{0}}^{2}$ centred at 
$(A,\mathcal{L}(\theta))_{\rho_{0}}$.  
}
\end{remark}

%\begin{remark}
%{\rm 
%We shall use the product form \eqref{eq.decomp.algebra} of $CCR(L^{2}(\rho), \sigma)$ 
%to show that the experiment $\mathcal{R}$ is equivalent to a {\it product}  of independent Gaussian shift experiments, a classical one for the real parameters $\vec{u}$ and a quantum one for each of the off-diagonal complex parameters $\zeta_{j,k}$. In practice this means that the state estimation problem will separate into {\it independent} estimation problems for each parameter. 
%}
%\end{remark}
{\it Proof.} From \eqref{eq.orthog.decomp} - \eqref{eq.def.phi.theta}, 
and by expressing $A$ in the symplectic basis \eqref{eq.symplectic.basis}
$$
A = 
 \sum_{j<k} 
\left( u_{j,k} t_{j,k}+ v_{j,k} t_{k,j}\right) + 
\sum_{i} w_{i} \ell_{i},
$$
we get
\begin{eqnarray}
\|A\|_{\rho_{0}}^{2} &=& 
w^{T} I_{\rho_{0}} w +  
\sum_{j<k} (u_{j,k}^{2}+ v_{j,k}^{2}) \frac{  \mu_{j}+\mu_{k}}{2(\mu_{j}-\mu_{k})},
\\
(A,\mathcal{L}(\theta))_{\rho_{0}} &=& 
w^{T} I_{\rho_{0}} u + 
\sum_{j<k} 
\frac{u_{j,k} {\rm Re}(\zeta_{j,k}) +v_{j,k} {\rm Im} (\zeta_{j,k}) }{\sqrt{2(\mu_{j}-\mu_{k})}},
\end{eqnarray}
which implies that the following decomposition holds
\begin{equation}\label{eq.Phi.theta}
\phi^{\theta} \cong  N(I_{\rho_{0}}u,I_{\rho_{0}}) 
\otimes \bigotimes_{j<k} \phi_{j,k}^{\zeta_{j,k}} :=  N(I_{\rho_{0}}u,I_{\rho_{0}}) \otimes \Phi^{\vec{\zeta}}
\end{equation}
where we used the following expression for the displaces thermal equilibrium states 
$\phi_{j,k}^{\zeta_{j,k}} = \phi_{\beta}^{z}$ defined in \eqref{eq.displacement.Fock}, with 
$\beta= \ln \mu_{j}/\mu_{k}, z=\zeta_{j,k}$ 
 \begin{equation}\label{eq.displaced.thermal}
\phi_{j,k}^{\zeta_{j,k}} 
\left( e^{ i( u {\bf Q}+ v {\bf P} )} \right)
= \exp\left (
-(u^{2}+v^{2})\frac{\mu_{j}+\mu_{k} }{4(\mu_{j}-\mu_{k})} + i 
\frac{u {\rm Re}(\zeta_{j,k}) + v{\rm Im}(\zeta_{j,k}) }{\sqrt{2(\mu_{j}-\mu_{k} )}} 
\right).
\end{equation}
%Indeed we have the explicit form 
%$$
%\Phi_{j,k}^{\zeta_{j,k}} (W(A)) = \Phi_{j,k} \left(
%W \left(\frac{\zeta_{j,k}}{2\sqrt{(\mu_{j}-\mu_{k} )}}\right) 
%W(A) 
%W\left(\frac{\zeta_{j,k}}{2\sqrt{(\mu_{j}-\mu_{k} )}}\right)^{*} 
%\right),
%$$
%where we have used the complex form of the Weyl operators 
%$$
%W\left(\frac{\zeta_{j,k}}{2\sqrt{(\mu_{j}-\mu_{k} )}}\right) =
%W \left( \frac{ {\rm Im}(\zeta_{j,k}) {\bf Q} - {\rm Re}(\zeta_{j,k}) {\bf P}}
%{ \sqrt{2(\mu_{j}-\mu_{k}} )}
%\right).
%$$

\qed

%%%%%%%%%%%%%%%%%%%%%%%%%%%%%%%%%%%%%%%%

\section{Explicit form of the channels and first steps of the proof}
\label{prepreuve}

\subsection{Second look at the irreducible representations of $SU(d)$}
\label{subsec.irrep.1}
Before explaining the steps involved in the proof, let us take a closer look at the block states \eqref{eq.block.states}. Recall that we have the decomposition of Theorem 
\ref{th.sud.sn} over Young diagrams with $n$ boxes and
$$
\rho^{\theta, n} = 
\bigoplus_{\lambda} \rho^{\theta,n}_{\lambda}\otimes \frac{\mathbf{1}_{\mathcal{K}_{\lambda}}}{M_{n}(\lambda)}.
$$ 
Let $\{f_1,\dots, f_{d}\}$  be the eigenvectors of $\rho_{0}$, i.e. the standard basis vectors of $\mathbb{C}^{d}$. Then the eigenvectors of 
$\rho_{0}^{\otimes n}= \rho^{0,n}$ are tensor products 
$$
f_{\bf a}:= f_{a(1)}\otimes f_{a(2)}\otimes \dots \otimes f_{a(n)},
$$
and the eigenvalues $\prod_{k} \lambda_{a(k)}$ do not depend on the order of the vectors in the product. 

\subsubsection{Projecting onto a copy of $\mathcal{H}_{\lambda}$.}

Our aim is to `project' to an irreducible representation $\mathcal{H}_{\lambda}$ and obtain an explicit expression for the eigenvectors of the block components 
$\rho^{\theta,n}_{\lambda}$. Such a projection is not unique, in fact for any rank one operator 
$ \vert v\rangle \langle u \vert \in \mathcal{B}(\mathcal{K}_{\lambda})$ with 
$\langle u\vert v\rangle=1$ we can define a (not necessarily orthogonal) projection 
$y=y^{2}$ on a copy of $\mathcal{H}_{\lambda}$
\begin{eqnarray*}
y_{\lambda}(u,v)  := \mathbf{1}_{\mathcal{H}_{\lambda}} \otimes \vert v\rangle \langle u\vert
:
(\mathbb{C}^{d})^{\otimes n} 
\to 
\mathcal{H}_{\lambda} \otimes  \vert v\rangle .
%\vert \psi \rangle 
%&\mapsto& 
%\left(\mathbf{1}_{\mathcal{K}_{\lambda}} \otimes \vert v \rangle \langle u  \vert \right) \psi
\end{eqnarray*}
However the action of $y_{\lambda}(u,v)$ on basis vectors $f_{\bf a}$ 
depends on a particular identification between $(\mathbb{C}^{d})^{\otimes n}$ and the direct sum in Theorem \ref{th.sud.sn}. Therefore we need a direct way of 
defining such a projection and the key observation is that 
$y_{\lambda}(u,v)$ is a {\it minimal} projection in the algebra ${\rm Alg}( \tilde{\pi}_{d}(\tau) : \tau\in S(n))$, i.e. it cannot be decomposed into a sum of non-zero projections, and vice-versa any minimal projection is of this form. The following recipe (given without proof) shows how to construct minimal projections in the $S(n)$ group algebra. We recall that the group $^{*}$-algebra $\mathcal{A}(S(n))$ is the linear space spanned by the group elements endowed with a product stemming from the group product
$$
a= \sum_{\tau\in S(n)}a(\tau)\tau, \quad 
b= \sum_{\varrho\in S(n)}b(\varrho)\varrho
~\Longrightarrow~
ab= \sum_{\tau,\varrho\in S(n)}a(\tau)b(\varrho) \tau\varrho = 
\sum_{\sigma\in S(n)}\left(\sum_{s\in S(n)} a(\sigma s^{-1}) b(s) \right) \sigma ,
$$
and with adjoint $a^{*}= \sum_{\tau\in S(n)} a(\tau)\tau^{-1} $. 

Let $\lambda$ be a Young diagram with $n$ boxes consider the (standard) Young tableau $t$ in which the boxes are filled with the numbers $\{ 1,\dots, n\}$ in increasing order from left to right along rows, starting with the top 
row and ending with the bottom row, as shown in the left-side tableau of 
Figure \ref{fig.standard.tableau}. 
\begin{figure}
\begin{displaymath}
\young(1234,56,7) \qquad\qquad
\young(11223,233,3)
\end{displaymath}
\caption{Left: a standard Young tableaux. Right: a semi-standard Young tableau for $d=3$}
\label{fig.standard.tableau}
\end{figure}

Define the group algebra elements
$$
P_{\lambda} = \sum_{\sigma\in \mathcal{R}_{\lambda}} \sigma , \qquad
Q_{\lambda} = \sum_{\tau\in \mathcal{C}_{\lambda}}{\rm sgn}(\tau) \tau,
$$
where $\mathcal{R}_{\lambda}$  is the $S(n)$-subgroup of permutation leaving the 
rows of $t$ invariant, and $\mathcal{C}_{\lambda}$  is the subgroup of permutations leaving the columns of $t$ invariant.  Note that $P_{\lambda}$ and $Q_{\lambda}$ are self-adjoint elements of the $S(n)$ group algebra satisfying
\begin{equation}\label{eq.p.q.square}
P_{\lambda}P_{\lambda} = |\mathcal{R}_{\lambda}| P_{\lambda} = 
(\prod_{i=1}^{d} \lambda_i!) P_{\lambda},
\quad\quad  
Q_{\lambda} Q_{\lambda}  =| \mathcal{C}(\lambda)| Q_{\lambda} = 
(\prod_{i=1}^d i^{\lambda_i - \lambda_{i+1}})Q_{\lambda}.
\end{equation}
 The \emph{Young symmetriser} is defined as
$$ 
Y_{\lambda} := Q_{\lambda}P_{\lambda}.
$$ 

\begin{thm}\label{th.minimal.proj.sn}
Up to a scalar normalising factor, the Young symmetriser $Y_{\lambda}$ is minimal projection in $\mathcal{A}(S(n))$  and $y_{\lambda}:= q_{\lambda}p_{\lambda}= 
\tilde{\pi}_{d}(Q_{\lambda})\tilde{\pi}_{d}(P_{\lambda})$ projects onto a copy of $\mathcal{H}_{\lambda}\subset (\mathbb{C}^{d})^{\otimes n}$.
\end{thm}

The action of the Young symmetriser $y_{\lambda}$ on basis vectors $f_{\bf a}\in (\mathbb{C}^{d})^{\otimes n}$ follows easily from the definition of $Y_{\lambda}$. For each $f_{\bf a}$ we fill the boxes of $\lambda$ with the indices $a(k)$ going along rows from left to right, starting with the top row and finishing with the bottom one. For example, if $\lambda ={\tiny \yng(3,2)}$ and $f_{\bf a}=f_{2}\otimes f_{2}\otimes f_{1}\otimes f_{2}\otimes f_{1}$ then  $t_{\bf a}= \tiny{\young(221,21)}$. $S(n)$ has an obvious action on the set of tableaux by permuting the {\it content} of the boxes which are 
numbered from $1$ to $n$ in the standard way as in Figure \ref{fig.standard.tableau}. 
The action of the Young symmetriser $y_{\lambda}= q_{\lambda}p_{\lambda}$ on  $f_{\bf a}$ is deduced from the action on the tableau $t_{\bf a}$ : one first symmetrises with respect to components which are in the same row, and then antisymmetrises with respect to components in the same column. 
For example if 
 ${\tiny \lambda= \yng(2,1)}$ then 
$$
y_{\lambda}(f_{2}\otimes f_{1} \otimes f_{3}) = 
f_{2}\otimes f_{1} \otimes f_{3} +f_{1}\otimes f_{2} \otimes f_{3}- f_{3}\otimes f_{1} \otimes f_{2}-f_{3}\otimes f_{2} \otimes f_{1}.
$$

%We briefly describe the procedure here and refer to \cite{Fulton,Goodman&Wallach,Fulton&Harris} for more details.  

\subsubsection{Finding a basis in $\mathcal{H}_{\lambda}$}

By the previous Theorem the vectors $y_{\lambda} f_{\bf a}$ span $\mathcal{H}_{\lambda}$, but are not linearly independent. We  show now how to select a basis ( subset of linearly independent vectors spanning $\mathcal{H}_{\lambda})$. 
A {\it semistandard} Young tableau is a diagram filled with numbers 
in $\{1,\dots, d\}$ such that the entries are non-decreasing along rows from left to right and increasing along columns from top to bottom, as in the right-side of Figure 
\ref{fig.standard.tableau}.  
\begin{thm}\label{th.basis.irrep}
The vectors $y_{\lambda}f_{\bf a}$ for which $t_{\bf a}$ is a semistandard 
Young tableau form a (non-orthogonal) basis of the irreducible representation 
$(\pi_{\lambda}, \mathcal{H}_{\lambda})$.
\end{thm}
Since  the values in the rows are nondecreasing, there is a one-to-one correspondence between Young tableaux $t_{\bf a}$ and vectors ${\bf m}=(m_{i,j})_{1\leq i < j\leq d}$ where $m_{i,j}$ is the number of $j$'s appearing in line $i$ of the Young tableau $t_{\bf a}$. Note that we need only consider $m_{i,j}$ for $j>i$, as there is no $j$ in line $i$ if $j<i$ (the columns are increasing), and the number of $i$ in line $i$ is $\lambda_i - \sum_{j=i+1}^d m_{i,j}$. 
For example, if $t_{\bf a}= \tiny{\young(11233,23,3)}$ then ${\bf m} = \{ m_{1,2}=1, m_{1,3}=2, m_{2,3}=1\}$.

By a slight abuse of notation we shall denote the corresponding vectors by $y_{\lambda}f_{\bf m}$ and the normalised vectors
\begin{equation}\label{eq.def.m.lambda}
|{\bf m}, \lambda\rangle
:=  \mathcal{N}({\bf m} ,\lambda)
y_{\lambda} f_{\bf m}, 
\end{equation}
where $\mathcal{N}({\bf m}, \lambda)= 1/\|y_{\lambda} f_{\bf m}\| $ . This constant is in general not easy to compute but we will describe its asymptotic properties in section 
\ref{pdisplacement}.

Using \eqref{eq.p.q.square} we have
\begin{equation}
\label{csq1}
\langle y_{\lambda} f_{\bf a} | y_{\lambda} f_{\bf b} \rangle  = \langle
q_{\lambda}p_{\lambda}  f_{\bf a} | q_{\lambda} p_{\lambda} f_{\bf b} \rangle  = 
\langle
p_{\lambda} f_{\bf a} | q_{\lambda}^{2} p_{\lambda} f_{\bf b}\rangle =
(\prod_{i=1}^d
i^{\lambda_i - \lambda_{i+1}}) \langle
p_{\lambda}f_{\bf a} | y_{\lambda} f_{\bf b} \rangle. 
\end{equation}

In order to get further simplifications, we examine some special vector states,
that we shall call by analogy with the Fock spaces \emph{finite-dimensional
coherent states}. 

The first is the special vector $|{\bf 0}, \lambda\rangle$, the \emph{highest weight
vector} of the representation $(\pi_{\lambda},\mathcal{H}_{\lambda})$, which later on 
will play the role of the \emph{finite-dimensional vacuum}. This vector, as we have seen,
corresponds to the semi-standard Young tableau where all the entries in row $i$
are $i$. An immediate consequence is that
\begin{equation}
\label{vpP}
p_{\lambda}| f_{\bf 0} \rangle  =
(\prod_{i=1}^{d} \lambda_i!) |f_{\bf 0}\rangle. 
\end{equation}

Moreover $\langle f_{\bf 0} | q_{\lambda} f_{\bf 0} \rangle = 1 $ since any column permutation produces a vector orthogonal to $f_{\bf 0}$. Thus the normalised vector is: 
\begin{equation}\label{eq.norm.0.lambda}
|{\bf 0}, \lambda\rangle =
\frac1{\prod_{i=1}^{d} \lambda_i! \sqrt{
i^{\lambda_i - \lambda_{i+1}}}} y_{\lambda} |f_{\bf 0}\rangle.
\end{equation}

The finite-dimensional coherent states are defined as 
$\pi_{\lambda}(U) |{\bf 0}_{\lambda}\rangle$ for $U\in
SU(d)$. From $[p_{\lambda}, \pi_{\lambda}(U)]=0$ and (\ref{vpP}), we get
$p_{\lambda} \pi_{\lambda}(U)  |{\bf 0}_{\lambda}\rangle  =
(\prod_{i=1}^{d} \lambda_i!) U |{\bf 0}_{\lambda}\rangle$, thus 
\begin{equation}
\label{fcoherent}
\langle y_{\lambda} f_{\bf m} | \pi_{\lambda}(U) |{\bf 0}, \lambda \rangle =
\sqrt{\prod_{i=1}^d
i^{\lambda_i - \lambda_{i+1}} } 
\langle p_{\lambda} f_{\bf m}| q_{\lambda} \pi_{\lambda}(U) f_{\bf 0} \rangle
\end{equation} 

The latter expression holds for any linear combination of $\Vm$ on the left-hand side, 
in particular $\pi_{\lambda }(V) f_{\bf 0}$ for another unitary operator $V$. In Lemma \ref{lemtools}, we shall examine asymptotics of
(\ref{fcoherent}) for specific sequences of unitaries $U$ when $n\to\infty$. 
One of the main tools will be formula (\ref{formdet}).

The following expressions of the dimensions of $\mathcal{K}_{\lambda}$ and 
$\mathcal{H}_{\lambda}$ are given without proof.

Let $g_{l,m}$ be the {\it hook length} of the box $(l,m)$, defined as one
plus the number of boxes under plus the number of boxes to the right. For example the diagram $(5,3,3)$ has the hook lengths : {\scriptsize $\young(76521,432,321)$}.

The dimension $M_n(\lambda)$ of $\mathcal{K}_{\lambda}$ is
\begin{equation*}
%\label{mult_lambda}
\ml = \frac{n!}{\prodtwo{l=1\dots d}{m= 1\dots \la_l}g_{l,m}},
\end{equation*} 
and can be rewritten in the following
form which is more adapted to our needs:
\begin{equation}
\label{mult_lambda}
\ml = {n \choose \lambda_1,\dots,\lambda_d}{\prodtwo{l=1\dots d}{k= l+1\dots
d}\frac{\lambda_l -\lambda_k + k -l}{\lambda_l + k - l}}.
\end{equation}

The dimension $\mathcal{D}(\lambda)$ of $\mathcal{H}_{\lambda}$ is:
\begin{equation}
\label{dim}
\mathcal{D}(\lambda)= \prodtwo{i=1\dots d}{j=1\dots\la_i} \frac{j+d-i}{g_{i,j}}.
\end{equation}

\bigskip

To summarise, we have defined a non-orthonormal 
basis \{$\left\vert {\bf m}, \lambda \right\rangle \}$ of $\mathcal{H}_{\lambda}$ such 
that $\left\vert {\bf m} ,\lambda \right\rangle $ are eigenvectors of 
$\rho^{\quantzero, \clas, n}$ for all $\lambda $, with eigenvalues:
\begin{equation}
\label{jesaispas}
\langle \mathbf{m}, \lambda  |  \rho^{\quantzero, \clas, n} | \mathbf{m}, \lambda\rangle = \prod_{i=1}^d
(\mu_i^{\clas, n})^{\lambda_i} \prod_{j=i+1}^d \left(\frac{\mu_j^{\clas,n}}{\mu_i^{\clas,n}}\right)^{m_{i,j}},
\end{equation} 
where $\mu_i^{\clas,n} = \mu_i + u_i /\sqrt{n}$ for $1\leq i\leq (d-1)$ and $\mu_d^{\clas,n} = \mu_d - (\sum_i u_i)/\sqrt{n}$.

The next step is to take into account the action of the unitary $U(\vec{\zeta})$. 
%For later purposes we define the more general unitaries
%\begin{eqnarray}
%\label{generalU}
%U(\quant, \cartan) & := & 
%\exp \left[ i \left( \sum_{i=1}^{d-1} \cartansub_i H_i + \sum_{1\leq j<k\leq d} \frac{{\rm Re}(\quantsub_{j,k}) T_{j,k} + {\rm Im}(\quantsub_{j,k}) T_{k,j}}{\mu_j - \mu_k }\right)  \right] ,\nonumber\\
%U(\quant, \cartan, n) & :=& U(\quant/\sqrt{n}, \cartan/\sqrt{n}),\quad 
%U(\quant) := U(\quant,\cartanzero),\quad 
%U(\quant,n) := U(\quant/\sqrt{n}),
%\end{eqnarray}
%where $H_{i}$ and $T_{i,j}$ are the generators of $SU(d)$ defined in \eqref{generators_algebra}. 
We define the automorphism of the $n$-particles algebra
$$
\Delta^{\quant, n}: 
M( (\mathbb{C}^{d})^{\otimes n}) \to M((\mathbb{C}^{d})^{\otimes n}),
$$
by
\begin{equation}
\label{finite_displacement_operator}
A\mapsto 
\Delta^{\quant,n}(A) 
= {\rm Ad}[U(\quant, n)](A):= U(\quant/\sqrt{n})^{\otimes n} \,A \, U^*(\quant/\sqrt{n})^{\otimes n},
\end{equation}
%and similarly for $\Delta^{\quant, n}$. 
In particular we have $\rho^{\quant, \clas, n} = \Delta^{\quant, n}(\rho^{\quantzero, \clas, n})$. 
By Theorem \ref{th.sud.sn} and using the decomposition \eqref{prerhon}, 
we get the blockwise action on irreducible components
$$
\Delta^{\quant, n} ( \rho^{\otimes n})
= 
\bigoplus_{\lambda} 
\Delta^{\quant,  n}_{\lambda} (\rho_{\lambda} ) \otimes {\bf 1}_{\mathcal{K}_{\lambda}},
$$
where $\Delta_{\lambda }^{\quant, n}= {\rm Ad}[U_{\lambda }(\quant,n)]$. In particular we have
%When acting on representations $\lambda $ of $SU(d)$, we naturally define $U_{\lambda }(\quant, \clas, n)$ and consequently $\Delta_{\lambda }^{\quant, \clas,n}$, and so on. Using the decomposition \eqref{rhontheta} of $\rho^{\otimes n}$, we obtain:
\begin{equation}
\label{finite_displacement_operator2}
\rho^{\quant, \clas, n}_{\lambda }  = \Delta^{\quant,n}_{\lambda } (\rho^{\quantzero, \clas, n}_{\lambda }).
\end{equation}

%Notice the similarity with equation \eqref{displacement_Fock}. The finite-dimensional displacement operators on $\lambda$ will be the analog of the displacement operators on the Fock space.

With these notations, we can set about building the channels $T_n$.

\subsection{Description of $T_n$}

We look for channels 
$$
T_n: M((\mathbb{C}^{d})^{\otimes n}) \to L^{1}(\mathbb{R}^{d-1})\otimes \mathcal{T}_{1}(\mathcal{F})
$$
of the form:
\begin{equation}
\label{formeTn}
T_n: \rho^{\glob,n} \longmapsto 
 \sum_{\lambda}  
  p_{\lambda}^{\glob,n}\, \tau_{\lambda}^n
  \otimes
   \left( V_{\lambda} \rho^{\glob,n}_{\lambda} V_{\lambda}^*\right) .
\end{equation}
Here, $V_{\lambda }$ is an isometry from $\mathcal{H} _{\lambda }$ to 
$\mathcal{F}$, i.e. $V_{\lambda }^* V_{\lambda} = \mathbf{1}_{\mathcal{H} _{\lambda }}$. On the classical side, $\tau_{\lambda }^n $ is a probability law on $\mathbb{R}^{d-1}$. We may view $\tau^n$ as a Markov kernel \eqref{eq.markov.kernel} from the set of diagrams $\lambda $ to $\mathbb{R} ^{d-1}$.

The channel $T_{n}$ can be described by the following sequence of operations. 
We first performs a `which block' measurement over the irreducible representations and get a result $\lambda$. Then, on the one hand, we apply  a classical randomization to $\lambda$, and on the other hand we apply a channel depending on our result $\lambda$ to the conditional state $\rho_{\lambda}$. 
 
%It can be proved from the axioms of quantum mechanics that this state is $\rho^{\glob,n}_{\lambda } \otimes \mathbf{1}_{\mathbb{C} ^{M_n(\lambda )}}/(M_n(\lambda ))$.

The underlying ideas are the following. 

1). The probability distribution $p_{\lambda}^{\glob,n}$ is essentially a multinomial depending {\it only} on $\clas$, as it can be deduced from \eqref{jesaispas} and \eqref{mult_lambda}. As we have seen in Example \ref{ex.limit.diagonal}, this converges (in Le Cam sense) to a classical Gaussian shift experiment. Here, in order to obtain the strong norm convergence we need to smooth the discrete distribution into a continuous one with respect to the Lebesgue measure. We choose a particular smoothing distribution which will insure the uniform $L^1$ convergence to the Gaussian model (Lemma \ref{lclassical}). 
\begin{defin}\label{def.markov.tau}
Let $\tau^n_{\lambda}$ be the probability density on $\mathbb{R} ^{d-1}$ defined for all $\lambda $ such that $\sum \lambda _i  = n$, by:
\begin{align}
\label{taulambda}
&\tau^n_{\lambda}(\dd x) = \tau^n_{\lambda}(x)\dd x = \dd x\,  n^{(d-1)/2} \chi( A_{\lambda,n}),
\end{align}
where $A_{\lambda,n}= \{ x \in \mathbb{R}^{d-1} :
\ | n^{1/2}x_i + n\mu_i -  \lambda_i | \leq 1/2 , 1\leq i\leq d-1\}$. We further denote
$$
b_{\lambda}^{\glob,n}= p_{\lambda}^{\glob,n} \tau^n_{\lambda},
$$
depending on $\theta$ only through $\vec{u}$.
\end{defin}

2). For the quantum part, we map the `finite-dimensional vacuum' $|{\bf 0}, \lambda \rangle$ to the Fock space vacuum $\left\lvert {\bf 0} \right\rangle $, and the basis vectors 
$|{\bf m}, \lambda \rangle$ of $\mathcal{H}_{\lambda}$ `near' the basis vectors $|{\bf m}\rangle$ of the Fock space $\mathcal{F}$ (cf. definitions \eqref{eq.def.m.lambda} and respectively \eqref{eq.def.m}). Here we need to tackle the problem that 
$\{|{\bf m}, \lambda \rangle\}$ is not an orthonormal basis but only becomes so asymptotically. The following lemma provides the isometry $V_{\lambda}$ appearing in 
\eqref{formeTn}.
\begin{lem}
\label{V_approx}
Let $\eta < 2/9$. Suppose that $\lambda _i - \lambda _{i+1} \geq \delta n$ for all 
$1\leq i\leq d$, with the convention $\lambda _{d+1} = 0$. Then for 
$n>n_{0}(\eta,\delta,d)$ there exists an isometry $V_{\lambda}: \mathcal{H}_{\lambda}\to \mathcal{F}$ such that, 
$V|{\bf 0},\lambda \rangle = |{\bf 0}\rangle $ and for $|\mathbf{m}| \leq n^{\eta}$, 
\[
\left\langle  \mathbf{m} \right| V_{\lambda}  = \frac1{\sqrt{1 + (\tilde{C}n)^{(9\eta -2) / 12}/\delta ^{1/3}}} \left\langle \mathbf{m}, \lambda \right|
\]
where $\tilde{C}= \tilde{C}(\eta,d)$ is a particular constant. More precisely, 
$n_{0}$ can be taken of the form $(C(d)/\delta^{2})^{1/(1-3\eta)}$.

\end{lem}

{\it Proof.} See section \ref{preuve_quasi_orth}. The main tool is Lemma \ref{non-orth}.

\qed

For Young diagrams which do not satisfy the assumption of the previous Lemma, the isometry $V_{\lambda}$ can be defined arbitrarily. The reason is that fact that those blocks have vanishing collective weight and can be neglected altogether 
(cf. Lemma \ref{lconcentration}).

From this operational description we conclude that $T_{n}$ is a 
proper channel since $\tau^{n}$ is a Markov kernel and $V_{\lambda}$ is an isometry. We then want to prove that $T_{\lambda }(\rho^{\quantzero,\clas,n}_{\lambda })$ is close to $\Phi^{\bf 0}$ and that the finite-dimensional operations 
$\Delta_{\lambda}^{\vec{\zeta},n}$ have almost the same action as the displacement operators $D^{\zeta}$ of the Fock space, cf. \eqref{eq.displacement.Fock}. 
%Formula \eqref{finite_displacement_operator2} would end the proof. 
Finite-dimensional coherent states and formula \ref{phi_integrale_gaussienne} will be the stepping stone to those results.

\section{Main steps of the proof}
\label{sec.main.steps}
\subsection{Why $T_n$ does the work}

We shall break (\ref{Tn}) in small manageable pieces. 
The result and  brief explanatory remarks, repeating those in the
derivation, are given from (\ref{develop}) on. 

We introduce first a few shorthand notations: the restriction of $T_{n}$ to the block 
$\lambda$ is
$$
T_{\lambda}: \rho^{\glob, n}_{\lambda } \mapsto 
V_{\lambda } \rho^{\glob,n}_{\lambda } V_{\lambda }^*,
$$ 
so that
\begin{align*}
T_n : \rho^{\glob, n} \mapsto \sum_{\lambda }
p_{\lambda }^{\glob,n} \tau^n_{\lambda } \otimes  
T_{\lambda }(\rho^{\glob,n}_{\lambda })  =
 \sum_{\lambda }  b_{\lambda }^{\glob,n} \otimes \phi^{\glob,n}_{\lambda }.
\end{align*}
We also define $T_{\lambda}^*: \phi \mapsto
V_{\lambda}^*\phi V_{\lambda}$. and note that $T_{\lambda}^*T_{\lambda}= 
{\rm Id}_{\mathcal{H} _{\lambda}}$.

We expand (\ref{formeTn}) as  
\begin{align*}
 T_n(\rho^{\glob,n}) & = \sum_{\lambda} 
 b_{\lambda}^{\glob,n}\otimes  \phi_{\lambda}^{\glob,n}
 \\
& = N(\clas,V_{\mu}) \otimes \phi^{\quant}
-\left( N(\clas,V_{\mu}) -  \sum_{\lambda}
b_{\lambda}^{\glob,n}\right)  \otimes\phi^{\quant} - 
 \sum_{\lambda}  b_{\lambda}^{\glob,n} \otimes  \left(\phi^{\quant} -
\phi_{\lambda}^{\glob,n}\right)  . 
\end{align*}

Proving (\ref{Tn}) then amounts to proving  
\[
\sup_{\glob \in \Omega_{n,\epsilon}} \left\| \left( N(\clas,V_{\mu}) -  \sum_{\lambda}
b_{\lambda}^{\glob,n}\right)\otimes\phi^{\quant} + 
\sum_{\lambda} b_{\lambda}^{\glob,n}  \otimes\left(\phi^{\quant} - \phi_{\lambda}^{\glob,n}\right)   \right\|_1 \leq
Cn^{-\epsilon/\delta}.
%\xrightarrow[n\to\infty]{} 0.
\]

We now use the triangle inequality to upper bound this norm  by a sum of ``elementary''
terms to be treated separately in the following sections.
\begin{eqnarray*}
%&\left\| T_n(\rho^{\glob,n}) - \phi^{\quant}\otimes \mathcal{N}(\clas,V_{\mu})
%\right\| \\
& &  \left\|\left(\mathcal{N}(\clas, V_{\mu}) -  \sum_{\lambda}
b_{\lambda}^{\glob,n}\right) \otimes \phi^{\quant} + 
\sum_{\lambda} b_{\lambda}^{\glob,n} \otimes \left(\phi^{\quant} -
\phi_{\lambda}^{\glob,n}\right)  \right\|_1 \leq \\
&&  \left\| \left( \mathcal{N}(\clas,V_{\mu}) -
\sum_{\lambda} b_{\lambda}^{\glob,n}\right)\otimes\phi^{\quant} \right\|_1 + 
\sum_{\lambda} \left\| b_{\lambda}^{\glob,n} \otimes \left(\phi^{\quant} -
\phi_{\lambda}^{\glob,n}\right)  \right\|_1 \leq\\
&&  \left\|\phi^{\quant}\right\|_1 \left\|\left( \mathcal{N}(\clas,V_{\mu}) -
\sum_{\lambda}
b_{\lambda}^{\glob,n}\right)\right\|_1 + \sum_{\lambda}
\left\| b_{\lambda}^{\glob,n}\right\|_1
\left\| \left(\phi^{\quant} - \phi_{\lambda}^{\glob,n}\right) \right\|_1  .
\end{eqnarray*} 

Since $\|\phi^{\quant}\|_1 = \| \mathcal{N}(\clas,V_{\mu})\|_1 = \|
\phi_{\lambda}^{\glob,n}\|= 1$, we have $ \left\| \left(\phi^{\quant} -
\phi_{\lambda}^{\glob,n}\right) \right\|_1 \leq 2$. Similarly 
$\sum_{\lambda}\| b^{\glob,n}_{\lambda}   \|_1 = 1$ because 
$\|b^{\glob,n}_{\lambda}   \|_1 = p_{\lambda}^{\glob,n}$.  
We split the sum over $\lambda$ in two
parts, one for which it is expected that $\left\| \left(\phi^{\quant} -
\phi_{\lambda}^{\glob,n}\right) \right\|_1$ is small, and the other on
which the sum of all $ \|b^{\glob,n}_{\lambda}   \|_1 $ is small. Specifically,
define the set of {\it typical Young diagrams} 
\begin{equation}\label{def.Lambda}
%\Lambda_{n,\alpha}:= \{ \lambda:  \delta n \leq \lambda _i - \lambda _{i+1} \leq n 
%(\mu_{i}-\mu_{j} ) (1+ n^{-1+\alpha})
%,1\leq i\leq d\}
\Lambda_{n,\alpha}:= \{ \lambda:    |\lambda _i - n\mu  _{i}|  \leq n^{\alpha}
,1\leq i\leq d\} ,\qquad {\rm for} ~ \alpha>1/2,
\end{equation}
then 
\begin{align}
&\left\| T_n(\rho^{\glob,n}) -  \mathcal{N}(\clas,V_{\mu})\otimes\phi^{\quant}
\right\| \notag \leq\\
\label{demidev}
&  \left\|  \mathcal{N}(\clas,V_{\mu}) -
\sum_{\lambda}
b_{\lambda}^{\glob,n} \right\|_1 +  \sup_{\lambda\in
\Lambda_{n,\alpha}} \left\| \phi^{\quant} -
\phi_{\lambda}^{\glob,n}  \right\|_1 + 2\sum_{\lambda\not\in
\Lambda_{n,\alpha}} \|b^{\glob,n}_{\lambda}   \|_1.
\end{align}

The first term corresponds to
the convergence of the classical experiment in the Le Cam sense.
If the second term is small, then on $\Lambda_{n,\alpha}$, the (purely
quantum) family $\rho^{\glob,n}_{\lambda}$ is near the family $\phi^{\quant}$. 
The last term corresponds to the other representations. If it is small, it means that there is concentration of $p_{\lambda}^{\glob,n}$ around  the representations with
shape $\lambda_i = n\mu_i$. In  other words, the only representations that
matter are those in $\Lambda_{n,\alpha}$, there is almost no mass on the
other representations.

The hardest term to dominate (notice that the two others are classical) is the
second. We transform it until we reach tractable fragments.
\begin{align*}
 \left\| \phi^{\quant} -
\phi_{\lambda}^{\glob,n} \right\|_1 
&
= \left\| \phi^{\quant} -
T_{\lambda}(\rho_{\lambda}^{\glob,n}) \right\|_1  \\
& =  \left\| D^{\quant}(\phi^{\quantzero}) - 
[T_{\lambda}\Delta^{\quant,n}_{\lambda}T_{\lambda}^*](T_{\lambda}(\rho^{\quantzero,\clas,n}_{\lambda})) \right\|_1 \\
& =  \left\| D^{\quant}(\phi^{\quantzero}) - D^{\quant}(T_{\lambda}(\rho^{\quantzero,\clas,
n}_{\lambda})) + D^{\quant}(T_{\lambda}(\rho^{\quantzero,\clas,
n}_{\lambda})) -
[T_{\lambda}\Delta^{\quant,n}_{\lambda}T_{\lambda}^*](T_{\lambda}(\rho^{\quantzero,\clas,
n}_{\lambda})) 
%+ [D^{\quant} - T_{\lambda}\Delta^{\quant,n}_{\lambda}T_{\lambda}^* ](\phi^\quantzero - \phi^\quantzero)
\right\|_1 
\\
\begin{split}
& \leq \left\|D^{\quant}(\phi^{\quantzero}) -
D^{\quant}(T_{\lambda}(\rho^{\quantzero,\clas,n}_{\lambda})) \right\|_1 
+ \left\|[D^{\quant} -
T_{\lambda}\Delta^{\quant,n}_{\lambda}T_{\lambda}^*
](T_{\lambda}(\rho^{{\quantzero},\clas,n}_{\lambda})-\phi^{\quantzero}) \right\|_1
\\ 
& \quad+
\left\|[D^{\quant} - T_{\lambda}\Delta^{\quant,n}_{\lambda}T_{\lambda}^*
](\phi^{\quantzero})
\right\|_1  
\end{split}
\\
& \leq 3 \left\| T_{\lambda}(\rho^{{\quantzero},\clas,n}_{\lambda})-\phi^{\quantzero} \right\|_1 
+ \left\|[D^{\quant} -
T_{\lambda}\Delta^{\quant,n}_{\lambda}T_{\lambda}^*
](\phi^{\quantzero}) \right\|_1
\end{align*} 
where in the last inequality we used the fact that the displacement operators are
isometries. 

Note that the first term does not depend on $\vec{\zeta}$ and the second term is small 
if the displacement operators $\Delta^{\vec{\zeta},n}_{\lambda}$ and $D^{\vec{\zeta}}$ have `similar action' on an appropriate domain. 
Using the integral formula \eqref{phi_integrale_gaussienne} for gaussian states 
$\phi_{\beta}$ and the fact that $\phi^{\quantzero}$, is a tensor product of such states (cf. \eqref{eq.Phi.theta}) we bound the second term by
\begin{align*}
\left\|[D^{\quant} -
T_{\lambda}\Delta^{\quant,n}_{\lambda}T_{\lambda}^*
](\phi^{\quantzero}) \right\|_1 
\leq \int_{ \mathbb{C}^{d(d-1)/2}} f(\vec{z}) \left\| 
[D^{\quant} - T_{\lambda}\Delta^{\quant,n}_{\lambda}T_{\lambda}^*]
( \left\vert\vec{z}\right\rangle \left\langle \vec{z}\right\vert)  \right\|_1 d\vec{z}
\end{align*}
where 
$$
 f(\vec{z}) =
\prod_{i<j}\frac{\mu_i-\mu_j}{\pi\mu_j}\exp\left(-\frac{\mu_i-\mu_j}{\mu_j}|z_{i,j}|^2\right).
$$
and 
$ \left \vert \vec{z} \right \rangle \left \langle \vec{z} \right\vert = 
D^{\vec{z}} (\left\vert{\bf 0}\right\rangle\left\langle {\bf 0}\right\vert)$ 
is the multimode coherent state, so
 $$
 [D^{\quant} - T_{\lambda}\Delta^{\quant,n}_{\lambda}T_{\lambda}^*]
(\left \vert \vec{z} \right \rangle \left \langle \vec{z} \right\vert) = 
[D^{\quant}D^{\vec{z}} - T_{\lambda}\Delta^{\quant,n}_{\lambda}T_{\lambda}^*D^{\vec{z}}]
(\left \vert {\bf 0} \right \rangle \left \langle {\bf 0} \right\vert ).
$$  
Now, $f$ is a probability density, and the norm in the integrand is dominated by two. By splitting the integral we obtain
\[
\left\|[D^{\quant} -
T_{\lambda}\Delta^{\quant,n}_{\lambda}T_{\lambda}^*
](\phi^{\quantzero}) \right\|_1  \leq 
2 \int_{\|\vec{z}\| > n^{\beta}} f(\vec{z}) \dd \vec{z} + \sup_{\|\vec{z}\|\leq n^{\beta}}  
\left\|[D^{\quant}  D^{\vec{z}}-
T_{\lambda}\Delta^{\quant,n}_{\lambda}T_{\lambda}^* D^{\vec{z}}
](\left \vert \bf{0} \right \rangle \left \langle \bf{0} \right\vert)  \right\|_1.
\]
By adding and subtracting additional terms
\begin{align*}
D^{\quant}D^{\vec{z}} -
T_{\lambda}\Delta^{\quant,n}_{\lambda}T_{\lambda}^*D^{\vec{z}}
= &
D^{\quant +\vec{z}} -
T_{\lambda}\Delta^{\quant +\vec{z},n}_{\lambda}T_{\lambda}^* \\
& + T_{\lambda}\Delta^{\quant +\vec{z},n}_{\lambda}T_{\lambda}^*
-
T_{\lambda}\Delta^{\quant,n}_{\lambda}\Delta^{\vec{z},n}_{\lambda}T_{\lambda}^* 
 \\
& + T_{\lambda}\Delta^{\quant,n}_{\lambda}\Delta^{\vec{z},n}_{\lambda}T_{\lambda}^* -
T_{\lambda}\Delta^{\quant,n}_{\lambda}T_{\lambda}^*D^{\vec{z}}.
\end{align*}
we deduce that 
\begin{align*}
\left\|[D^{\quant} -
T_{\lambda}\Delta^{\quant,n}_{\lambda}T_{\lambda}^*
](\left \vert \vec{z} \right \rangle \left \langle \vec{z} \right\vert)  \right\|_1 
&
\leq \left\| [D^{\quant +\vec{z}} -
T_{\lambda}\Delta^{\quant +\vec{z},n}_{\lambda}T_{\lambda}^*](\left \vert {\bf 0} \right \rangle \left \langle {\bf 0} \right\vert) \right\|_1 
\\ 
&
\quad + \left\| [\Delta^{\quant +\vec{z},n}_{\lambda}
-
\Delta^{\quant,n}_{\lambda}\Delta^{\vec{z},n}_{\lambda}](|\bf{0}, \lambda \rangle\langle\bf{0}, \lambda|)
\right\|_1 
\\
& 
\quad + \left\| [\Delta^{\vec{z},n}_{\lambda}T_{\lambda}^* -
T_{\lambda}^*D^{\vec{z}}](\left \vert {\bf 0} \right \rangle \left \langle {\bf 0} \right\vert) \right\|_1
\end{align*}
where the last two terms on the right side have been simplified using properties of 
$T_{\lambda}, T^{*}_{\lambda}, \Delta_{\lambda}^{\vec{\zeta}, n}$. Notice that the first and third norms are essentially the same and the three terms are small if the 
action of $\Delta_{\lambda}^{\vec{\zeta}}$ is mapped into that of the displacement operators $D^{\vec{\zeta}}$.

\smallskip 

%Saying that the first norm is small corresponds to saying  that the
%``finite-dimensional'' displacement operator acts on the void like the
%infinite-dimensional displacement operator. Saying that the second norm is
%small amounts to asserting that the ``finite-dimensional'' displacement
%operators multiply like the infinite-dimensional operators, at least when seen
%through their action on the void. These two points together yield that the
%action on coherent states of ``finite-dimensional'' and infinite-dimensional
%displacement operators are the same: a coherent state is obtained through the
%action of a displacement operator on the void, and the composition of two
%displacement operators is the displacement operator with parameter the sum of
%the two parameters. 

Putting all this together, our `expanded' form for (\ref{Tn}) is 
\begin{align}
\label{develop}
&\sup_{\glob\in \Omega_{n,\beta, \gamma}}\left\| T_n(\rho^{\glob,n}) -
\phi^{\quant}\otimes \mathcal{N}(\clas,V_{\mu})
\right\| \\
\label{classique}
& \leq  \sup_{\glob\in\Omega_{n,\beta, \gamma}}\left\|\left(
\mathcal{N}(\clas,V_{\mu}) -
\sum_{\lambda}
b_{\lambda}^{\glob,n}\right)\right\|_1 \\
\label{concentration}
& + 2\sup_{\glob\in\Omega_{n,\beta,\gamma}}\sum_{\lambda\not\in
\Lambda_{n,\alpha}} \|b^{\glob,n}_{\lambda}   \|_1
 \\
\label{en0}
& +3 \sup_{\glob\in \Omega_{n,\beta,\gamma}} \sup_{\lambda\in\Lambda_{n,\alpha}}\left\|
\phi^{{\quantzero}} - T_{\lambda}(\rho^{{\quantzero},\clas,n}_{\lambda})
\right\|_1 \\
\label{deplacement1}
& + \sup_{\|\vec{z}\|\leq n^{\beta}} \sup_{\glob\in\Omega_{n,\beta, \gamma}}\sup_{\lambda\in
\Lambda_{n,\alpha}} \left\| 
 [D^{\quant +\vec{z}} -
T_{\lambda}\Delta^{\quant +\vec{z},n}_{\lambda}T_{\lambda}^*]
( \left\vert \mathbf{0}\right\rangle \left\langle \mathbf{0}\right\vert) 
\right\|_1
 \\
\label{deplacement2}
& + \sup_{\|\vec{z}\|\leq n^{\beta}} \sup_{\glob\in\Omega_{n,\beta, \gamma}}\sup_{\lambda\in
\Lambda_{n,\alpha}}
\left\| 
 [D^{\vec{z}} - T_{\lambda}\Delta^{{\vec{z}},n}_{\lambda}T_{\lambda}^* ]
( \left\vert \mathbf{0}\right\rangle \left\langle \mathbf{0}\right\vert) 
 \right\|_1 
\\
\label{groupelimite}
& +  \sup_{\|\vec{z}\|\leq n^{\beta}} \sup_{\glob\in\Omega_{n,\beta, \gamma}}\sup_{\lambda\in
\Lambda_{n,\alpha}}\left\|
[\Delta^{\quant +\vec{z},n}_{\lambda}
-
\Delta^{\quant,n}_{\lambda}\Delta^{\vec{z},n}_{\lambda}]
( \left\vert \mathbf{0}, \lambda\right\rangle \left\langle \mathbf{0},\lambda\right\vert)
\right\|_1\\
\label{reste}
& +2 \int_{\|\vec{z}\|\geq n^{\beta}} f(\vec{z}) \dd \vec{z}.
\end{align} 

The last Gaussian tail term is less than $C \exp(-\delta n^{2\beta})$ where $C$ depends only on the dimension $d$. Under the hypothesis $n^{2\beta} > 2/\delta$, this can be bounded again by $O(n^{-2\beta})$.

The following lemmas provide upper bounds for each of the terms. Before each lemma we remind the reader what is the significance of the bound. The proofs are gathered in section \ref{sec.technical.proofs}.

The classical part of the channel is a Markov kernel $\tau$ 
(see definition \ref{def.markov.tau}) mapping the `which block' distribution 
$p_{\lambda}^{\theta,n}$ into the density $b_{\lambda}^{\theta,n}$ on 
$\mathbb{R}^{d-1}$ which is approaches uniformly the 
gaussian shift experiment (\ref{classique}).
Recall that $b_{\lambda}^{\glob,n}$ depends only on $\clas$ and not on $\quant$,
so that we have the same parameter set for the two classical experiments.
\begin{lem}
\label{lclassical}
With the above definitions, for any $\epsilon$, 
%for $n > (C/\delta)^{\frac1{1 - \alpha}} +  (C/\delta)^{\frac2{1 - 2\gamma}} $, for a constant $C$ depending only on the dimension and $\epsilon$, 
we have 
\[
\sup_{\glob\in\Omega_{n,\beta,\gamma}}\left\| \mathcal{N}(\clas,V_{\mu}) -
\sum_{\lambda}
b_{\lambda}^{\glob,n}\right\|_1 
%\leq C
=O \left(n^{-1/2+\epsilon}/\delta,  n^{-1/4 + \gamma} / \delta\right).
\]
\end{lem}

The next lemma deals with \eqref{concentration} by showing 
concentration around Young diagrams $\lambda$ in the `typical subset' 
\eqref{def.Lambda}. This allows we to restrict to this set of diagrams in further estimates. 
\begin{lem}
\label{lconcentration}
Let $\alpha-\gamma-1/2>0$. Then, with the above definitions 
%and for $n >  (2 d/\delta)^{\frac{1}{1 - \alpha}} + (2d)^{\frac{1}{\alpha-\gamma-1/2}},$ 
we have 
\[
\sup_{\glob\in\Omega_{n,\beta,\gamma}}\sum_{\lambda\not\in
\Lambda_{n,\alpha}} \|b^{\glob,n}_{\lambda}   \|_1
=O
%\leq 
%2d
\left(n^{d^{2}} \exp(-  n^{2\alpha - 1}/2)\right),
% \xrightarrow[n\to\infty]{} 0. 
\]
with the $O(\cdot)$ term converging to zero.
\end{lem}

%We shall then generalize
%the result to all $\quant$ by recalling that we obtain $\phi^{\quant}$ and
%$\rho^{\quant,\clas,n}_{\lambda}$ from $\phi^{\quantzero}$ and
%$\rho^{\quant,\clas,n}_{\lambda}$ by actions of displacement operators, and that
%we can decompose them in coherent states. See following points. 

The term \eqref{en0} shows that when the rotation parameter is zero, the block 
states $\rho_{\lambda}^{\vec{0},\vec{u},n}$ are essentially thermal equilibrium states, as one would expect from the quantum Central Limit Theorem \ref{th.clt}. However the convergence here is in norm rather than in distribution, and uniform over the various parameters.
\begin{lem}
\label{len0}
With the above definitions, we have
%for $n^{\eta} > C\ln(n)/\delta$,
\[
\sup_{\glob\in \Omega_{n,\beta,\gamma}} \sup_{\lambda\in\Lambda_{n,
\alpha}}\left\|
\phi^{{\quantzero}} - T_{\lambda}(\rho^{{\quantzero},\clas,n}_{\lambda})
\right\|_1 = O(n^{-1/2 + \gamma + \eta}/\delta, \,n^{(9\eta - 2)/ 24 }/\delta^{1/6}).
\]
\end{lem}

The terms \eqref{deplacement1} and \eqref{deplacement2}  show that the `finite dimensional coherent states' obtained by performing small rotations on the `finite-dimensional vacuum' are uniformly close to their infinite dimensional counterparts, thus justifying the coherent state terminology.
\begin{lem}
\label{ldisplacement}
Let $\epsilon>0$ be such that $2 \beta + \epsilon \leq \eta < 2 / 9$. 
%and $\epsilon n^{\beta + \epsilon} \geq \beta$. 
%and $n^{-1/2 + 3 \beta + 2\epsilon} \leq C\delta^{-3/2}$, and 
%$n^{1-3\eta}>2C\delta^{-1}$ for $C$ depending only on the dimension $d$. 
%Assume moreover that $\|\cartan\| \leq n^{-1/2 + 2 \beta}/\delta$. 

Then, 
\[
\sup_{\|\vec{z}\| \leq n^{\beta}}\,
\sup_{\|\cartan\| \leq n^{-1/2 + 2 \beta}/\delta }\, 
\sup_{\glob\in\Omega_{n,    \beta,\gamma}}\,
\sup_{\lambda\in\Lambda_{n,\alpha}} \,\left\| 
 [D^{\quant +\vec{z}} -
T_{\lambda}\Delta^{\quant +\vec{z},\cartan,n}_{\lambda}T_{\lambda}^*]
(\left\vert \mathbf{0}\right\rangle \left\langle \mathbf{0}\right\vert) 
\right\|_1
= R(n)
\]
with
\begin{multline}
\label{Rn}
%R(n) =  O\left(n^{(9\eta -2) / 24}\delta ^{-1/6},n^{-1/2 + \beta + \eta/2}\delta^{-1/2}, n^{-1/4 + \beta/2} \delta^{-1/4}, n^{-1/2 + \alpha/2 + \beta/2}\delta^{-1/2}, \right.\\
%\left.  n^{-1/2 + \alpha/2 + \eta/2}\delta^{-1/2}, n^{-1/2 + 3\eta/2}\delta^{-1/2}, n^{-\beta/2}\right).
R(n)^{2}=O\left(n^{(9\eta -2) / 12}\delta ^{-1/3},n^{-1 + 2\beta + \eta}\delta^{-1}, n^{-1/2 + 3\beta+2\epsilon} \delta^{-3/2}, n^{-1 + \alpha +2 \beta}\delta^{-1}, \right.\\
\left. n^{-1 + \alpha + \eta}\delta^{-1}, n^{-1 + 3\eta}\delta^{-1}, n^{-\beta}\right)
\end{multline}
For estimating the terms (\ref{deplacement1}, \ref{deplacement2}), the case when $\cartan = \vec{0}$ is sufficient. This more general form is useful for the proof of Lemma \ref{lgrouplimit}. The unitary operation is defined as 
$\Delta_{\lambda}^{\vec{\zeta},\xi,n} := {\rm Ad}[U_{\lambda}(\vec{\zeta},\xi,n) ]$ 
with $U(\vec{\zeta},\xi,n))$ the general $SU(d)$ element of \eqref{generalU}.
\end{lem}
Finally \eqref{groupelimite} shows that the `finite-dimensional' displacement operators multiply as the corresponding displacement operators when acting on the vacuum. 
\begin{lem}
\label{lgrouplimit}
With the above definitions, under the same hypotheses as in Lemma \ref{ldisplacement}, we have
\[
\sup_{\|\vec{z}\| \leq n^{\beta}}\sup_{\glob\in\Omega_{n,\beta,\gamma}} \sup_{\lambda\in\Lambda_{n,
\alpha}}
\left\|
[\Delta^{\quant +\vec{z},n}_{\lambda}
-
\Delta^{\quant,n}_{\lambda}
\Delta^{\vec{z},n}_{\lambda}]
(|\bf{0}, {\lambda}\rangle \langle\bf{0}, {\lambda}|)
\right\|_1 %\right)
= R(n)
\]
with $R(n)$ given by equation \eqref{Rn}.
\end{lem}

From the last three lemmas, together with the bound on the remainder integral \eqref{reste} we obtain the following lemma which can be plugged into the 
bound (\ref{demidev}):
\begin{lem}
\label{toutquantique}
With the above notations 
under the same hypotheses as in Lemma \ref{ldisplacement}, we have
%and $n^{2\beta}> 2d/\delta$,
\[
\sup_{\glob\in\Omega_{n,\beta,\gamma}}\sup_{\lambda\in\Lambda_{n,\alpha}}
\|\phi^{\quant} - \phi^{\glob,n}_{\lambda}\| = R(n) + O(n^{-1/2 + \gamma + \eta}/\delta + n^{(9\eta - 2)/24}/\delta ^{1/6}) 
\]
with $R(n)$ given by equation \eqref{Rn}.
\end{lem}

Gathering all these results yield the following theorem which provides the bound  
\eqref{Tn}.
\begin{thm}\label{thmpoubelle}
For any $\delta > 0$, $1 > \alpha > 1/2$, $\eta < 2/9$, $\epsilon > 0$, $\beta < (\eta + \epsilon) / 2$, $\gamma < 1/4$,  
%and $n$ such that $\epsilon n^{\beta + \epsilon} > \beta$, $n^{1 - \alpha} > C/\delta$, 
%$n^{\eta}/\ln(n) > C/\delta$, $n^{1/2 - \gamma}> C/\delta$, 
the sequence of channels $T_n$ satisfies
\begin{multline}
\sup_{ \glob\in \Omega_{n,\beta,\gamma}}\left\| T_n(\rho^{ \glob,n}) -
\phi  \right\|_1 =O 
(n^{-1/2 + \beta + \eta/2}\delta^{-1/2} +  n^{-1/4 + \beta/2} \delta^{-1/4} +
  n^{-1/2 + \alpha/2 + \eta/2}\delta^{-1/2} + \\
 n^{-1/2 + 3\eta/2}\delta^{-1/2} +  n^{-\beta/2} + n^{-1/2 + \gamma + \eta}/\delta + n^{(9\eta - 2)/24}/\delta ^{1/6}) 
\end{multline}
%where $C$ depends only on the dimension $d$.
With any explicit $\alpha, \beta, \gamma, \delta$, we get an explicit polynomial rate.
\end{thm}

\subsection{Definition of $S_n$ and proof of its efficiency}

The channel $S_{n}$ is essentially the inverse of $T_{n}$ and as we shall 
see, (\ref{Sn}) can be deduced from (\ref{Tn}).

On the classical side we need a Markov kernel completing the equivalence between the
family $p_{\lambda}^{\clas,n}$ and $ \mathcal{N}(\clas,V_{\mu})$. Let $\sigma^n$ be
defined by
\begin{equation}
\label{sigmalambda}
\sigma^n: x \in \mathbb{R}^{d-1} \mapsto \delta_{\lambda_x}  
\end{equation}
where $\lambda_x$ is the Young diagram such that $\sum_{1}^d \lambda_i = n$, 
and $ |n^{1/2}x_i+n\mu_i - \lambda_i| < 1/2 $, for 
$2\leq i\leq d$. No such diagram exists, we set $\lambda_{x}$ to any admissible value, for example $(n,0,\dots,0)$. 
Notice that
with (\ref{taulambda}), $\sigma^n \circ\tau^n\circ \sigma^n = \sigma^n$. Moreover any
probability on the $\lambda$ such that $\sum_{1}^d \lambda_i = n$ is in the
image of $\sigma^n$, so that $\sigma^n \circ\tau^n (p^{\glob,n})= p^{\glob,n}$. 
\begin{lem}
\label{sigma}
With the above definitions, for any $\epsilon$, 
%for $n > (C/\delta)^{\frac1{1 - \alpha}} +  (C/\delta)^{\frac2{1 - 2\gamma}} $, for a constant $C$ depending only on the dimension and $\epsilon$, 
we have 
\[
\sup_{ \| \clas\| \leq n^{\gamma} }\left\|\sigma^n \mathcal{N}(\clas,V_{\mu}) -
p^{\clas,n}\right\|_1  =
\leq C 
O\left( n^{-1/2+\epsilon}/\delta ,  n^{-1/4 + \gamma}/\delta\right) .
\]
\end{lem}
{\it Proof.} See end of section \ref{proofs.classical.lemmas}.

\qed

The channel $S_n$ is given by the following sequence of operations acting on the two spaces of the product $ L^{1}( \mathbb{R}^{d-1})\otimes \mathcal{T}_1(\mathcal{F}) $. Given a sample from the probability distribution $N(\clas,V_{\mu})$, we use the Markov kernel $\sigma^{n}$ to produce a Young diagram $\lambda$. Conditional on $\lambda$ we send the quantum part through the  channel 
\[
S_{\lambda}: \phi 
\mapsto 
\tilde{S}_{\lambda}(\phi)\otimes \frac{{\bf 1}_{\mathcal{K}_{\lambda}} }{M_n(\lambda)}
\]
with 
\[
\tilde{S}_{\lambda}:  \phi 
\mapsto 
T_{\lambda}^* \phi  +
(1-\Tr(T_{\lambda}^*(\phi))) |{\bf
0}, {\lambda}\rangle\langle{\bf 0}, {\lambda}|.
\]
The second term is rather arbitrary and insures that $\tilde{S}_{\lambda}$ is trace preserving map. 
What is important is that for any density operator $\rho_{\lambda}$ on the block 
$\lambda$, the operator $\tilde{S}_{\lambda}$ reverts the action of 
$T_{\lambda}$:
\begin{align*}
\tilde{S}_{\lambda} T_{\lambda}
(\rho_{\lambda})  & =
T^*_{\lambda}T_{\lambda}(\rho_{\lambda}) + (1 -
\Tr(T_{\lambda}^*T_{\lambda}(\rho_{\lambda})))|{\bf
0},{\lambda}\rangle\langle{\bf 0},{\lambda}| \\
& = \rho_{\lambda} +  (1 -
\Tr(\rho_{\lambda}))|{\bf
0},{\lambda}\rangle\langle{\bf 0},{\lambda}| \\
& = \rho_{\lambda}.
\end{align*} 
Now
\begin{align*}
S_n( \mathcal{N}(\clas,V_{\mu}) \otimes\phi^{\quant}) &= 
\bigoplus_{\lambda}
[\sigma^{n} \mathcal{N}(\clas,V_{\mu})](\lambda) \tilde{S}_{\lambda}(\phi^{\vec{\zeta}})\otimes
\frac{{\bf 1}_{\mathcal{K}_{\lambda} }}{M_n(\lambda)}.
\end{align*}
and with the notation $\sigma^{n} \mathcal{N}^{\clas}_{\lambda} : =
[\sigma^{n}\mathcal{N}(\clas,V_{\mu}))](\lambda)$ and $q^{\clas,n}_{\lambda} :=
\min(\sigma^{n} \mathcal{N}^{\clas}_{\lambda}, p^{\clas,n}_{\lambda})$
we have 
\begin{align*}
& S_n(\phi^{\quant}\otimes \mathcal{N}(\clas,V_{\mu})) - \rho^{\glob,n} 
\\
&\quad =
\bigoplus_{\lambda}
\left\{q^{\clas,n}_{\lambda}
(\tilde{S}_{\lambda}(\phi^{\quant})-\rho^{\glob,n}_{\lambda}) + (\sigma^{n}
\mathcal{N}^{\clas}_{\lambda} -
q^{\clas,n}_{\lambda})\tilde{S}_{\lambda}(\phi^{\quant}) -
(p^{\clas,n}_{\lambda} - q^{\clas,n}_{\lambda})\rho^{\glob,n}_{\lambda} \right\}\otimes
\frac{{\bf 1}_{
\mathcal{K}_{\lambda}}}{M_n(\lambda)}.
\end{align*} 
Taking $L^1$ norms, and using that all $\phi$'s and $\rho$'s have trace $1$
and that channels (such as $\tilde{S}_{\lambda}$) are trace preserving, we
get the bound:
\begin{align*}
\left\| S_n(\phi^{\quant}\otimes \mathcal{N}(\clas,V_{\mu})) - \rho^{\glob,n}
\right\|_1 & \leq \sum_{\lambda} \left\|q^{\clas,n}_{\lambda}
(\tilde{S}_{\lambda}(\phi^{\quant})-\rho^{\glob,n}_{\lambda})   \right\|_1 +
\sum_{\lambda} \left| \sigma\mathcal{N}^{\clas}_{\lambda} - p^{\clas,n}_{\lambda}
\right|
\\
& \leq 2\sum_{\lambda\not\in \Lambda_{n,\alpha}} q^{\clas,n}_{\lambda} +
\sup_{\lambda\in\Lambda_{n,\alpha}}  \left\|
\tilde{S}_{\lambda}(\phi^{\quant})-\rho^{\glob,n}_{\lambda}   \right\|_1 +
\left\|\sigma^{n} \mathcal{N}(\clas,V_{\mu}) -
p^{\clas,n}\right\|_1
\\
& \leq 2\sum_{\lambda\not\in \Lambda_{n,\alpha}} q^{\clas,n}_{\lambda} +
\sup_{\lambda\in\Lambda_{n,\alpha}}  \left\|
\phi^{\quant} - T_{\lambda}(\rho^{\glob,n}_{\lambda})  \right\|_1 +
\left\|\sigma^{n} \mathcal{N}(\clas,V_{\mu}) -
p^{\clas,n}\right\|_1
. 
\end{align*}
Now the first term is 
 smaller than the
remainder term of the gaussian outside a ball whose radius is $n^{\alpha}$. Hence this term is going to zero faster than any polynomial, independently on $\delta$ and $\clas$ for $\| \clas\| \leq n^{\gamma}$. 
The second term is treated in Lemma \ref{toutquantique} (recalling that $\phi^{\glob,n}_{\lambda} = T_{\lambda}(\rho^{\glob,n}_{\lambda})$), and the third term is treated in Lemma \ref{sigma}. 

\smallskip

This ends the proof of (\ref{Sn}). 

\qed
%IL FAUDRA PRECISER CETTE VITESSE DANS LE THEOREME
%%%%%%%%%%%%%%%%%%%%%%%%%%%%%%%%%%%%%%%%%%%
\section{Technical proofs}
\label{sec.technical.proofs}
%%%%%%%%%%%%%%%%%%%%%%%%%%%%%%%%%%%%%%%%%%%%
\subsection{Combinatorial and representation theoretical tools}
\label{technical_tools}

Here we continue the analysis of the $SU(d)$ irreducible representations $(\pi_{\lambda}, \mathcal{H}_{\lambda})$ started in section \ref{subsec.irrep.1}. The purpose of this section is to provide good estimates of quantities of the type $\langle \mathbf{m}, {\lambda} \mid \pi_{\lambda}(U) \mid \mathbf{l}, {\lambda} \rangle$ which will be needed in the proofs of Lemmas \ref{non-orth} and \ref{ldisplacement}. 

%We shall usually drop the explicit reference to the representation and write $U$ instead of $\pi_{\lambda}(U)$.
We shall use  the following form of a general $SU(d)$ element and the 
shorthand notations 
\begin{eqnarray}
\label{generalU}
U(\quant, \cartan) & := & 
\exp \left[ i \left( \sum_{i=1}^{d-1} \cartansub_i H_i + \sum_{1\leq j<k\leq d} \frac{{\rm Re}(\quantsub_{j,k}) T_{j,k} + {\rm Im}(\quantsub_{j,k}) T_{k,j}}{\sqrt{\mu_j - \mu_k} }\right)  \right] ,\nonumber\\
U(\quant, \cartan, n) & :=& U(\quant/\sqrt{n}, \cartan/\sqrt{n}),\quad 
U(\quant) := U(\quant,\cartanzero),\quad 
U(\quant,n) := U(\quant/\sqrt{n}),
\end{eqnarray}
where $H_{i}$ and $T_{i,j}$ are the generators of $SU(d)$ defined by 
\begin{eqnarray}\label{generators_algebra}
 H_{j}  &= & E_{j,j} - E_{j+1,j+1} \qquad  \text{for}~j\leq d-1  ; \nonumber \\
T_{j,k} &= & i E_{j,k} - i E_{k,j}  \qquad \text{for} ~ 1\leq j<k \leq d; \nonumber \\
T_{k,j} &=  & E_{j,k} + E_{k,j}   \qquad \text{for}     1\leq j<k \leq d.
\end{eqnarray}  
with $E_{i,j}$ the matrix with entry $(i,j)$ equal to $1$, and all others equal to $0$.

%We have begun studying these scalar products in Section \ref{Sud}. We now proceed with the analysis, starting from \eqref{formdet}.

% We first recall a few notations:  $\mid \mathbf{m} \mid = \sum_{i < j} m_{i,j}$.

\medskip

We first introduce some new notations and remind the reader about the already existing ones. 

1) We write $\lc$ for the length of the column $c$ in the Young diagram $\lambda$. There are then $\lambda _i - \lambda _{i+1} $ columns such that $\lc = i$. An alternative definition is $l(c) = \inf\{ i : \lambda_i\geq c\}$. 

2) Recall that we denote by $\Va$ the basis vectors $f_{a(1)} \otimes \dots \otimes f_{a(n)}$, and to each vector we associate a Young tableau $t_{\bf a}$ where the indices $a(i)$ fill the boxes of a diagram $\lambda$ in a particular way. 
We denote by $\tca$ the column $c$ of $t_{\bf a}$, i.e. the function 
$t_{\bf a}^{c}: \{1,\dots,\lc\} \to\{1,\dots, d\}$ that associates to the row number $r$ the value of the entry of that Young tableau in column  $c$, row $r$. For example, if $t_{\bf a} = \tiny{\young(221,21)} $  we get the values:
\begin{align*}
t_{\bf a}^{1}(1) & = 2, & t_{\bf a}^{1}(2) & = 2, & t_{\bf a}^{2}(1) & = 2, & t_{\bf a}^{2}(2) & = 1, & t_{\bf a}^{3}(1) & = 1. 
\end{align*}
We shall often be interested in the image $\tca(\{1,\dots, \lc\})$ as unordered set, 
or compare $\tca$ to ${\rm Id}^c$, the identity function on the integers $\{1,\dots, l(c)\}$.

3) Recall also that $\mathcal{H}_{\lambda}$ is spanned by the vectors 
$y_{\lambda} f_{\bf a}$ for which $t_{\bf a}$ is a semistandard Young tableau, and 
$y_{\lambda} = q_{\lambda}p_{\lambda}$ is the Young symmetriser 
(cf. Theorem \ref{th.basis.irrep}). If $t_{\bf a}$ is semistandard then we can use the alternative notation $f_{\bf m}$ for $f_{\bf a}$ since ${\bf a}$ is in one-to-one correspondence  with  ${\bf m}= \{m_{i,j} :  1\leq i<j\leq d\} $, where $m_{i,j}$ is the number of $j$'s in the row $i$ of $t_{\bf a}$. The normalised vectors are 
$$
\left\vert {\bf m}, \lambda \right\rangle := 
y_{\lambda} f_{\bf m} /\| y_{\lambda} f_{\bf m} \|. 
$$

4) Let $ \mathcal{O}_{\lambda}({\bf m})$ be the orbit of $f_{\bf m}$ under the 
subgroup $\mathcal{R}_{\lambda}$ of row permutations. This consists of vectors 
$\Vb$ which have exactly $m_{i,j}$ boxes with $j$ in row $i$, and the rest are $i$. In particular, row $i$ has no entries smaller than $i$. Since the action of permutations is transitive, we have 
\begin{equation}\label{plambda.orbit.decomposition}
p_{\lambda} \Vm = \sum_{\sigma \in
\mathcal{R}_{\lambda}}  f_{{\bf a}\circ \sigma}=
\sum_{\Vb\in \mathcal{O}_{\lambda}({\bf m})}
\frac{\# \mathcal{R}_{\lambda}}{\# \mathcal{O}_{\lambda}({\bf m})} \Vb.
\end{equation}
%with $ \mathcal{R}_{\lambda}$ the subgroup of $S_n$ letting invariant the rows of $\lambda$ and $ \mathcal{O}_{\lambda}({\bf m})$ is the orbit in 
%$(\mathbb{C}^{d})^{\otimes n}$ of $\Vm$ under $\mathcal{R}_{\lambda}$. Note that 
%$ \mathcal{O}_{\lambda}({\bf m})$ is the set of $\Vb$'s which have exactly $m_{i,j}$ boxes with $j$ in row $i$, and the rest are $i$. 

5) Since we antisymmetrize with $q_{\lambda}$, we are only interested in the 
$t_{\bf a}$ (not necessarily semistandard) which do not have two equal entries in the same column. Such  tableaux $t_{\bf a}$ (or vectors $f_{\bf a}$) shall be called \emph{admissible} and their set is denoted $\mathcal{V}$.

6) For any $f_{\bf a}\in\mathcal{O}_{\lambda}({\bf m})$ we define 
$$
\Gamma(\Va)  := |{\bf m}| - \#\{ 1\leq c \leq\lambda_{1} : \tca \neq {\rm Id}^{c} \},
$$
and denote by $\mathcal{V}^{\Gamma}({\bf m})$ the set of vectors 
$f_{\bf a}\in \mathcal{O}_{\lambda}({\bf m})\bigcap \mathcal{V}$ with 
$\Gamma(f_{a})=\Gamma$. Then we have
%the set of admissible vectors in $\mathcal{O}_{\lambda}({\bf m})$ splits according to the values of $\Gamma$ 
$$
%\{ f_{\bf a} \in\mathcal{O}_{\lambda}({\bf m})~ {\rm admissible} \} =
\mathcal{O}_{\lambda}({\bf m})\bigcap \mathcal{V} =
\bigcup_{\Gamma\in \mathbb{N}} \mathcal{V}^{\Gamma}({\bf m}).
%= : 
%\bigcup_{\Gamma\in \mathbb{N}}\left( \{ \Va \in\mathcal{O}_{\lambda}({\bf m}): 
%\Gamma(\Va) = \Gamma\} \bigcap \mathcal{V} \right).
$$
%Note that we do not make explicit the dependence on $\mathbf{m}$, 
%and that $\Gamma \geq 0$. 
Note that $\Gamma(f_{\bf a})\geq 0$  and is zero if and only if each column 
$\tca$ is either ${\rm Id}^{c}$ or of the form $\tca(r) = j\delta_{r=i} + r\delta_{r\neq i}$ for some $i \leq l(c) < j$. A $\tca$ of this form will be called an \emph{$(i,j)$-substitution}.

%With these definitions, we prove in Lemma \ref{lemtools} many formulas that we shall use for  proving Lemmas \ref{ldisplacement} and \ref{non-orth}.

The following `algorithm' shows how to build all the possible $\Va\in\mathcal{V}^{\Gamma}({\bf m})$, thus enabling us to estimate the size of $\mathcal{V}^{\Gamma}({\bf m} )$.

\bigskip

\emph{Algorithm}
\nopagebreak

Let $({\bf m}, \lambda)$ be fixed but otherwise arbitrary. In order to generate a particular admissible $f_{\bf a} \in \mathcal{O}_{\lambda}({\bf m})$ we need to select the $m_{i,j}$  boxes on row $i$ which are filled with $j$, for all $i<j$. The rest of the boxes are filled automatically with $i$'s. The constraint is that no column should have two boxes filled with the same number.

Generating a diagram can be described intuitively as follows. We start with the `vacuum' 
vector (tableau) $f_{\bf 0}:= f_{{\bf m}={\bf 0}}$ (row $i$ is filled exclusively with 
$i$'s), and with a set of $|{\bf m}|$ bricks containing 
$m_{i,j}$ identical bricks labelled $(i,j)$, for each pair $i<j$. 
To change the content of a box from $i$ into $j$ we place an $(i,j)$-brick in that 
box. This procedure is repeated until all bricks have been used, each box being 
modified at most once. 

At this stage each column $c$ may contain several bricks placed in the appropriate boxes, so that its configuration is uniquely defined by the set of bricks $\colmod$  which shall be called a \emph{column-modifier}. 
%By placing several bricks in different boxes of the {\it same} column we get a general \emph{column-modifier} $\colmod$ which can be identified with the set of bricks. 
For example if $\colmod =\{ (i ,j), (f,l)\}$ then the column has entries 
$$
t_{\bf a}^{c} (k)= 
\left\{ 
\begin{array}{lll}
j &  \textrm{if $k = i$} ;\\
l &  \textrm{if $k = f$} ;\\
k & \textrm{otherwise}. 
\end{array}
\right.
$$
Note that a column-modifier is not an arbitrary collection of bricks but one that can be used to produce a column with different entries. In the previous example, if $i<f$ 
this means either ($j\neq f$ and $j,l>l(c)$) or ($j=f$ and $l>l(c)$). The elementary one-brick column-modifier denoted $\colmod(i,j)$ can only be used in a column with $i\leq l(c)<j$, otherwise the entry $j$ would appear twice.
%Then $\Gamma(\Va)$ is the total number of bricks in columns with non-elementary modifiers minus the number of such columns. 

%We can see this as having $m_{i,j}$ bricks labelled $(i,j)$ for each pair $i<j$, and the question is where to place them inside the diagram $\lambda$. 

%We can see this as having $m_{i,j}$ bricks labelled $(i,j)$ for each pair $i<j$, and the question is where to place them inside the diagram $\lambda$. The constraints are 
%that no two numbers in a column are the same, i.e $t_{\bf a}$ is admissible. 
%The value $\Gamma(\Va)$ is the total number of 
%{\it additional} bricks we put in columns where there was already one brick, if we
%set them sequentially. 

%Here is a slightly different view of the process. We start with the `vacuum' 
%vector (tableau) $f_{{\bf m}={\bf 0}}:=f_{\bf 0}$ where row $i$ is filled exclusively with 
%$i$'s. 

%Consider the notion of \emph{column-modifier} $\colmod$, that is something we apply on a column to change it. An $(i,j)$ brick is an elementary column-modifier that changes an $i$ of row $i$ into a $j$. We shall denote it $\colmod(i,j)$. We can compose column-modifiers by changing simultaneously a number of boxes in the same column, without repeated modifications of the same box. 
%For example $\colmod =\{ (i ,j), (f,l)\}$ produces changes on rows $i$ and $f$. 

Now, since the length of a column is at most $d$ and all entries must be different, there are less than $d!$ different types of column-modifiers. 
Another important remark is that a column-modifier always increases the value of the modified cells, so that in this case  $\tca(\{1, \dots, l(c)\}) \neq \{1,\dots, l(c)\}$.

Alternatively to the above scenario where the bricks are inserted sequentially, 
we can first cluster them into $|{\bf m}| - \Gamma$ column-modifiers, and then apply each column-modifier to a particular column. A given collection of column-modifiers is uniquely determined by $\{ m_{\colmod} : \colmod\}$ where $m_{\colmod}$ is the multiplicity of $\colmod$. This procedure is detailed in the following 3 stages:

%Then $\Va$ is obtained by applying all our $|{\bf m}|$ bricks clustered
%in $|{\bf m}| - \Gamma $ column-modifiers  (there are $m_{\colmod}$
%times the column-modifier $\colmod$), and each column-modifier being
%applied to a different column, according to the following sequence: 

%We then give the following ``algorithm''.
\begin{enumerate}
{\item[I.]
\label{1} Choose $\Gamma$ bricks among our $|{\bf m}|$. As we have $d(d-1)/2$
different types of bricks (recall that $i>j$), and we do not distinguish between identical bricks, there are at most $[d(d-1)/2]^{\Gamma} $ possibilities. 
For $\Gamma = 0$, we have only one choice.}
{\item[II.]
\label{2} Consider the remaining bricks as a set of elementary column-modifiers. 
Starting from these, we sequentially add each of the $\Gamma$ bricks selected in the first stage, to one of these elementary column-modifiers to form non-elementary ones. 
At each step we have at most $d!$ different {\it types} of column modifiers to which we can attach the new brick. Note that we do not distinguish between column modifiers of the same type, but rather consider them as an unordered set. Hence, we have less that  $(d!)^{\Gamma}$ possibilities. 

Note that at the end of stage II at least $\max \{ 0, |{\bf m}| -2\Gamma\}$ of the column-modifiers are elementary, and that $m_{\colmod(i,j)}\leq m_{i,j}$.    }
{\item[III.]
\label{3} Apply the column-modifiers to the columns of $\Vzero$, so
that no two modifiers are applied to the same column and the resulting 
$\Va\in\mathcal{O}_{\lambda}({\bf m})$ is admissible. By construction 
$\Gamma(f_{\bf a})= \Gamma$ and all admissible tableaux can be generated in this way.

}
\end{enumerate}
For counting the number of possibilities for the third stage we apply the column 
modifiers sequentially, but since some of them may be identical we need to
divide by the combinatorial factor $\prod_{\colmod} m_{\colmod}!$, where 
$m_{\colmod}$ is the number of column modifiers of type $\colmod$.

We distinguish between elementary column modifiers of type $\colmod(i,j)$ and composite ones. There are less than n possibilities of inserting a composite 
column-modifier $\colmod$. An elementary one of type $\colmod(i,j)$ can only be inserted in a column with at least $i$ rows, and since the resulting vector has to be admissible, the column cannot contain another $j$, so its length is smaller than $j$. There are $\lambda_i-\lambda_j$ such columns.  Hence the number of possibilities at stage three of the algorithm is upper bounded by 
\begin{align}
\label{up3}
\prod_{\colmod \neq \colmod({i,j})}  \frac{n^{ m_{\colmod}}}{m_{\colmod}!} \cdot
\prod_{i<j} \frac{(\lambda_i - \lambda_j)^{m_{\colmod(i,j)}}} {m_{\colmod(i,j)}!}.
\end{align}
When $\Gamma = 0$, for each elementary column modifier $\colmod(i,j)$ 
the number of available columns is at least 
$(\lambda_i -\lambda_j - \lvert{\bf m}\vert)_{+}:= \max\{ 0, \lambda_i -\lambda_j - \lvert{\bf m}\vert \} $.
Thus we have the following lower bound
\begin{align}
\label{do3}
  \prod_{i<j} \frac{ (\lambda_i - \lambda_j - |\mathbf{m}|)_{+}^{m_{i,j}}}{m_{i,j}!}.
\end{align}

%Among them, we must suppress the columns already modified, which are less than
%$|{\bf m}| - \Gamma$. We have then between $\lambda_i -\lambda_j -
%\lvert{\bf m}\rvert$ and $\lambda _i - \lambda _j$ possibilities when inserting each $\colmod(i,j)$ elementary
%column modifier. 

Notice that the upper bound \eqref{up3} depends on the set of multiplicities 
$\{m_{\colmod}\}$.
 
We now return to our list of notations and definitions. 

7) To each column of $t_{\bf a}$ we associated a column modifier which completely determines its content. If $m_{\colmod}^{\bf a}$ is the number of columns with 
column-modifer $\colmod$, we collect all multiplicities in  
$
E := \{ m_{\colmod}^{\bf a} :\colmod\}.
$
In particular $\Gamma$ is a function of $E$ 
$$
\Gamma(f_{\bf a}) = |{\bf m}| - \sum_{\colmod} m_{\colmod}^{\bf a}.
$$
Vectors for which $\Gamma(f_{\bf a})=0$ have the same multiplicity 
set $E^{0}$ where $m_{\colmod(i,j)} = m_{i,j}$ for all $i<j$ and the other 
$m_\colmod = 0$. 
Similarly to $\mathcal{V}^{\Gamma}({\bf m})$, we denote by 
$\mathcal{V}^{E}({\bf m})$ the set of tableaux in 
$\mathcal{O}_{\lambda}({\bf m})\bigcap \mathcal{V}$ with 
$E(f_{\bf a})= E$, in particular
$$
\mathcal{V}^{\Gamma} ({\bf m})= \bigcup_{ E } \mathcal{V}^{E}({\bf m})
$$ 

%To summarise, we have defined the maps 
%$$
%f_{\bf a} \longmapsto E (f_{\bf a}) \longmapsto \Gamma(f_{\bf a}).
%$$

8) To each column $c$ of $t_{\bf a}$ we associate two disjoint sets: 
the added entries 
$\{ t_{\bf a}^{c}(1), \dots , t_{\bf a}^{c}(l(c))\}  \setminus \{ 1,\dots, l(c)\}$ 
and the deleted entries 
$\{ 1,\dots, l(c)\} \setminus \{ t_{\bf a}^{c}(1), \dots , t_{\bf a}^{c}(l(c))\}$. This data is placed into a single set by attaching a $\pm$ sign to each entry, indicating if 
it is added or deleted. It is easy to verify that if $t_{\bf a}$ is admissible, the set of added and deleted entries is uniquely determined by the column-modifer $\colmod$ associated to $c$, and hence shall be denoted by $S(\colmod)$. For example $S(\colmod(i,j))  =
\{ (i,-), (j,+)\} $ and for  $\colmod= \{ (i,j), (j,k)\}$ we have
$S(\colmod)  = \{ (i,-), (k,+)\}$. 
We define the multiplicities $m_S^{\bf a}= \sum_{\colmod : S(\colmod) =S} m_{\colmod}^{\bf a}$ and $F(f_{\bf a}):= \{m_S^{\bf a} :S \}$ .
To summarise, we have defined the maps 
$$
f_{\bf a} \longmapsto E(f_{\bf a})  \longmapsto  F(f_{\bf a}).
$$

We now state our estimates. The first point of the following lemma is an exact formula serving as the main tool to prove some of the bounds below. 

\begin{lem}
\label{lemtools}
\noindent
\begin{enumerate}
\item{For any unitary operator $U\in M(\mathbb{C}^{d})$, 
for any basis vectors $\Va$ and $\Vb$, we have
\begin{align}
\label{formdet}
\langle\Va\vert q_{\lambda } U^{\otimes n} \Vb \rangle & = \prod_{1\leq c \leq \lambda _1} \det(U^{\tca,\tcb}),
\end{align}
where $U^{\tca,\tcb}$ is the $\lc \times \lc$ minor of $U$ given by $[U^{\tca,\tcb}]_{i,j} = U_{\tca(i),\tcb(j)}$.
}
\end{enumerate}

%We now get bounds for $\langle {\bf m}, {\lambda } \vert U \vert {\bf l}, {\lambda } \rangle$ on the interesting range of parameters. 
Under the assumptions
\begin{align}
\label{hyp}
 |\mathbf{m}| & \leq n^{\eta}, \\
\lambda & \in \Lambda_{n,\alpha}, \notag \\
\inf_{i} |\mu_i - \mu_{i+1}| & \geq \delta ,\notag \\
\mu_d & \geq \delta ,\notag \\
\|\quant\|_1 & \leq C n^{\beta}, \qquad \beta\leq 1/2 \notag \\ % for any positive C
\|\cartan\|_1 & \leq n^{-1\!/2 + 2 \beta}/\delta ,\notag \\
n & > \left(\frac2{\delta}\right)^{1/(1-\alpha)}. \notag
\end{align}
%Let us define the general unitaries in $SU(d)$
%\begin{eqnarray}
%\label{generalU}
%U(\quant, \cartan) & := & 
%\exp \left[ i \left( \sum_{i=1}^{d-1} \cartansub_i H_i + \sum_{1\leq j<k\leq d} \frac{{\rm Re}(\quantsub_{j,k}) T_{j,k} + {\rm Im}(\quantsub_{j,k}) T_{k,j}}{\mu_j - \mu_k }\right)  \right] ,\nonumber\\
%%U(\quant) = U(\quant,\cartanzero), \quad 
%U(\quant, \cartan, n) & :=& U(\quant/\sqrt{n}, \cartan/\sqrt{n}),
%% \quad U(\quant,n) = U(\quant/\sqrt{n}),
%\end{eqnarray}
%where $H_{i}$ and $T_{i,j}$ are the generators of $SU(d)$ defined in \eqref{generators_algebra}.
we have the following estimates  with remainder terms uniform in the eigenvalues $\mu_{\bullet}$:

\begin{enumerate}[resume]
%%%%%%%%%%%%%%%%%%%%%%%%%%%%%%%%%%
\item  { The number of admissible $\Va\in \mathcal{O}_{\lambda}({\bf m})$ 
with $\Gamma(\Va) = 0$ is 
\begin{equation}
\label{cg0}
\# \mathcal{V}^{0} ({\bf m})= 
\prod_{j>i}\frac{(\lambda_i-\lambda_j)^{m_{i,j}}}{m_{i,j}!} (1+ O(n^{-1 + 2\eta} / \delta)).
\end{equation}
}
%%%%%%%%%%%%%%%%%%%%%%%%%%%%%%%%%%
\item {
Let $E:= \{m_{\colmod}:\colmod\} $ with $\Gamma(E)=\Gamma$. 
The number %$\mathcal{V}^{E_m}$ 
of admissible $\Va\in \mathcal{O}_{\lambda}({\bf m})$ with $E(\Va) = E$ is bounded by:
\begin{equation}
\label{cggam}
\# \mathcal{V}^{E} ({\bf m}) \leq  n^{-\Gamma  + \sum_{i<j} (m_{i,j} - m_{\colmod({i,j})})} \prod_{j>i}\frac{(\lambda_i-\lambda_j)^{m_{\colmod({i,j})}}}{m_{\colmod({i,j})}!}.
\end{equation} 
}
%%%%%%%%%%%%%%%%%%%%%%%%%%%%%%%%%%
\item{
The number of admissible $\Va\in \mathcal{O}_{\lambda}({\bf m})$ with 
$\Gamma(\Va) = \Gamma$ is bounded by:
\begin{equation}
\label{cgG}
\# \mathcal{V}^{\Gamma}({\bf m}) \leq C^{\Gamma} n^{-\Gamma} \delta^{-2 \Gamma} |\mathbf{m}|^{2\Gamma} \prod_{j>i}\frac{(\lambda_i-\lambda_j)^{m_{i,j}}}{m_{i,j}!},
\end{equation}
for a constant $C= C(d)$.
}
%%%%%%%%%%%%%%%%%%%%%%%%%%%%%%%%%%
\item { 
Let $\Va \in \mathcal{V}^{\Gamma^{a}}(\mathbf{l})$, and consider $ \mathcal{V} ^{\Gamma ^b} ({\bf m}) \subset\mathcal{O}_{\lambda}(\mathbf{m})$ for some fixed 
$\Gamma^b$. Then:
\begin{align}
\label{prod_id}
\left| \left\langle \Va \Bigg| q_{\lambda} \sum_{\Vb \in \mathcal{V}^{\Gamma^b} ({\bf m})} \Vb \right\rangle \right| \leq \left\{
\begin{array}{cc}
0 & \mbox{if } \Gamma^b \neq |\mathbf{m}| - |\mathbf{l}| + \Gamma^a \\
(C |\mathbf{m}|)^{\Gamma^b} & \mbox{otherwise} \\
\end{array}
\right.,
\end{align}
with $C= C(d)$.
}
%%%%%%%%%%%%%%%%%%%%%%%%%%%%%%%%%%
\item { If $\Va \in \mathcal{V}^{0}(\mathbf{m})$, then
\begin{align}
\label{prod_id0}
\left\langle \Va \Bigg| q_{\lambda} \sum_{\Vb \in \mathcal{O}_{\lambda}(\mathbf{m})} \Vb \right\rangle  = 1.
\end{align}
}
%%%%%%%%%%%%%%%%%%%%%%%%%%%%%%%%%%
\item\label{zem0} {
If $f_{\bf a} \in \mathcal{V}^{0}({\bf m})$ so that its set of elementary column-modifiers is  
$E^{0}= \{ m_{\colmod(i,j)}=m_{i,j} \}$, then 
\begin{align}
% Z(E^0) & \stackrel{\mathrm{def}}{=} 
\langle \Va | q_{\lambda} U(\quant, \cartan, n)^{\otimes n} \Vzero\rangle% \notag \\
\label{interg0}
 & = \exp\left(i\phi -\frac{\|\quant\|^2_2}{2}\right) \prod_{i < j} \left(\frac{\quantsub_{i,j}}{\sqrt{n}\sqrt{\mu_i - \mu_j}}\right)^{m_{i,j}} r(n),
\end{align}
with the phase and error factor
\begin{align*}
\phi & =  \sqrt{n} \sum_{i = 1}^{d-1} (\mu_i - \mu_{i+1}) \xi_i, \\
r(n) & = 1 + O\left(n^{-1 + 2\beta + \eta}\delta^{-1}, n^{-1/2 +2 \beta} \delta^{-1}, 
n^{-1 + 2\beta+ \alpha }\delta^{-1} \right).
\end{align*}
}
%%%%%%%%%%%%%%%%%%%%%%%%%%%%%%%%%%
\item { If $\Va \in \mathcal{V}^{E}({\bf m})$, so that its  set of column-modifiers is 
$E= \{m_{\colmod}: \colmod\}$ and $\Gamma(E)= \Gamma$, then
\begin{eqnarray}
%| Z(E) | & \stackrel{\mathrm{def}}{=} 
&&
\hspace{-0.5cm}
\left| \langle \Va | q_{\lambda} U(\quant, \cartan, n)^{\otimes n} \Vzero\rangle\right| 
\notag \\\label{interg}
&& 
\hspace{-0.5cm}
\leq \exp\left(-\frac{\|\quant\|^2_2}{2}\right) \left(\frac{ C\|\quant\|}{\sqrt{n\delta}}\right)^{- \Gamma+\sum_{i<j} (m_{i,j} - m_{\colmod({i,j})}) }\prod_{i < j} \left(\frac{ \quantsub_{i,j}}{\sqrt{n}\sqrt{\mu_i - \mu_j}}\right)^{m_{\colmod({i,j})}} r(n)
\end{eqnarray}
with $C=C(d)$ a constant and $r(n)$ as in point \ref{zem0} above.
%error factor
%\begin{align*}
%r(n) & = 1 + O\left(n^{-1 + 2\beta + \eta}\delta^{-1}, n^{-1/2 + \beta} \delta^{-1/2}, n^{-1 + \alpha + \beta}\delta^{-1} \right).
%\end{align*}
}
%%%%%%%%%%%%%%%%%%%%%%%%%%%%%%%%%%
\item {Under the further hypotheses  that $\|\vec{z}\| \leq n^{\beta}$, $m_{i,j} \leq 2 |\quantsub_{i,j}+ z_{i,j}|n^{\beta + \epsilon}$ for some $\epsilon > 0$,
% and $  n^{-1/2 + 3 \beta + 2 \epsilon} \leq  C\delta^{-3/2}/ 2$ where $C$ is a constant depending only on the dimension $d$, 
 we have:
\begin{align}
\label{numer}
& \left\langle \sum_{\Va \in \mathcal{O}_{\lambda}(\mathbf{m})}\Va \Bigg| q_{\lambda} U(\quant + \vec{z}, \cartan, n) \Vzero \right\rangle \notag \\
 & =  \exp\left(i\phi -\frac{\|\quant + \vec{z}\|^2_2}{2}\right) \prod_{i<j}\frac{\big((\zeta_{i,j} + z_{i,j})(\sqrt{n}\sqrt{\mu_i - \mu_j})\big)^{m_{i,j}}}{m_{i,j}!} r(n),
\end{align} 
with 
\begin{equation*}
r(n) =  1 + O\left(n^{-1 + 2\beta + \eta}\delta^{-1}, 
%n^{-1/2+2\beta}\delta^{-1},
n^{-1 + 2 \beta+\alpha}\delta^{-1}, n^{-1 + 2\eta}\delta^{-1}, n^{-1 + \alpha + \eta}\delta^{-1}, \delta^{-3/2} n^{-1/2 + 3 \beta + 2 \epsilon} \right).
\end{equation*}
}
%%%%%%%%%%%%%%%%%%%%%%%%%%%%%%%%%%
\item {Under the further hypotheses that  $|\mathbf{l}| \leq |\mathbf{m}|$ and $n^{1 - 3 \eta} > 2 C / \delta^{2}$, where $C=C(d)$,
\begin{align}
\label{denom}
\left\vert\left \langle \sum_{\Va \in \mathcal{O}_{\lambda}(\mathbf{l} )}\Va \Bigg| q_{\lambda}  \sum_{\Vb \in \mathcal{O}_{\lambda}(\mathbf{m})}\Vb \right\rangle\right \vert
& \leq   (C |\mathbf{m}|)^{|\mathbf{m}| - |\mathbf{l}|}\prod_{i<j} \frac{(\lambda_i - \lambda_j)^{l_{i,j}}}{l_{i,j}!} \left(\frac{C|\mathbf{l}|^2|\mathbf{m}|}{n\delta^2}\right)^{\Gamma^a_{\min}(\mathbf{l}, \mathbf{m})}
\end{align} 
with 
\begin{align}
\Gamma^a_{\min}(\mathbf{l}, \mathbf{m}) \geq \frac{\big(|\mathbf{l} - \mathbf{m}|  + 3 |\mathbf{l}| - 3 |\mathbf{m}| \big)_{+}}{6} .
\end{align}
}
%%%%%%%%%%%%%%%%%%%%%%%%%%%%%%%%%%
\item{
%With $n^{1 - 3\eta} > 2 C / \delta^{2}$, where $C=C(d)$ , 
We have
\begin{align}
\label{denom54}
\left\langle \sum_{\Va \in \mathcal{O}_{\lambda}(\mathbf{m})}\Va \Bigg| q_{\lambda}  \sum_{\Vb \in \mathcal{O}_{\lambda}(\mathbf{m})}\Vb \right\rangle 
& =
\prod_{i<j} \frac{(\la_i - \la_j)^{m_{i,j}}}{m_{i,j}!} \big(1 + O(n^{3\eta - 1}/\delta)\big).
\end{align} 
}
%%%%%%%%%%%%%%%%%%%%%%%%%%%%%%%%%%
\end{enumerate}

\end{lem}

\bigskip

\begin{proof}
~

\emph{Proof of \eqref{formdet}.} We first express $\langle \Va\vert U^{\otimes n} \Vb\rangle $ as a product of matrix entries of $U$:
\begin{align*}
\langle \Va\vert U^{\otimes n} \Vb\rangle & = \prod_{1\leq c \leq \lambda _1} \prod_{1 \leq r \leq \lc} \langle f_{\tca(r)} \vert U f_{\tcb(r)}\rangle \notag \\
%\label{intermformdet}
& =\prod_{1\leq c \leq \lambda _1} \prod_{1 \leq r \leq \lc} U_{\tca(r), \tcb(r)}.
\end{align*}

Since the subgroup of column permutations $\mathcal{C} _{\lambda }$ is the product of the permutation groups of each column, each $\sigma\in\mathcal{C} _{\lambda }$ is
$
\sigma = s _1\dots s_{\lambda_{1}}
$
with $\ssc $ a permutation of column $c$ which transforms $\tcb(r)$ into $\tcb(\ssc(r))$.  Then
\begin{align*}
\langle\Va\vert q_{\lambda }  U^{\otimes n}  \Vb \rangle =\langle\Va\vert U^{\otimes n}  q_{\lambda }  \Vb \rangle & = \sum_{\sigma \in \mathcal{C} _{\lambda }}  \epsilon(\sigma) \prod_{1\leq c \leq \lambda _1}  \prod_{1\leq r\leq  \lc}  
U_{\tca(r), \tcb(\ssc(r))} \\
& = \prod_{1\leq c \leq \lambda _1} \sum_{\ssc \in S_c} \epsilon (\ssc) \prod_{1\leq r \leq  \lc}  U_{\tca(r), \tcb(\ssc(r))} \\
%& = \prod_{1\leq c \leq \lambda _1} \sum_{\ssc \in S_c} \epsilon (\ssc) \prod_{1\leq r \leq \lc}  [U^{\lc,\tca, \tcb}]_{r,\ssc(r)} \\
& = \prod_{1\leq c \leq \lambda _1} \det(U^{\tca,\tcb}).
\end{align*}
%Remembering that $U$ commutes with $q_{\lambda }$ and that $U^{\lc,\tca,\tcb}$ is the $\lc \times \lc$ submatrix of $U$ given by $[U^{\lc,\tca,\tcb}]_{i,j} = U_{\tca(i),\tcb(j)}$, we have proved formula \eqref{formdet}.

\emph{Proof of \eqref{cg0}.}
%\nopagebreak
The number of admissible $\Va$ such that $\Gamma(\Va)=0$ is given by the products of the possibilities at each stage of the algorithm. For the first two stages, there is exactly one possibility when $\Gamma = 0$. Hence $\#\mathcal{V}^0$ is the number of possibilities at the third stage. Here the upper bound \eqref{up3} reads as $\prod_{j>i}(\lambda_i-\lambda_j)^{m_{i,j}} / m_{i,j}! $. On the other hand, we may use \eqref{do3} as a lower bound, recalling that   $\lambda_i - \lambda_j \geq \delta n / 2$ and 
$ |{\bf m}| \leq n^{\eta}  $ (cf. \eqref{hyp}). This yields the result \eqref{cg0}.

\bigskip

\emph{Proof of \eqref{cggam}.}
%\nopagebreak
The number of $\Va$ in $\mathcal{V}^{E}$ is given by the third stage of the algorithm (the two first stages yield a particular $E$). We then obtain \eqref{cggam} by applying \eqref{up3} and neglecting the $m_{\colmod }!$ factors, while noticing that $\sum_{\colmod} m_{\colmod} = |\mathbf{m}| - \Gamma$.

\bigskip

\emph{Proof of \eqref{cgG}.}
%\nopagebreak
The set $\mathcal{V}^{\Gamma}$ is the union of all $\mathcal{V}^{E}$ with $\Gamma(E) = \Gamma$. Now the first two stages of the algorithm imply that there are at most $C^{\Gamma}$ different $E$ with the latter property, with $C= C(d)$.

Now we use  \eqref{cggam} to upper-bound  $\mathcal{V}^{E}$ as follows.
Since $\sum m_{\colmod({i,j})} \geq |\mathbf{m}| - 2 \Gamma$, we may write $\prod_{\colmod} m_{\colmod (i,j)}! \geq \prod_{i<j} m_{i,j}! \sup_{i<j}m_{i,j}^{-2\Gamma}$. Moreover $\lambda_i -\lambda_j \geq \delta n / 2$. By putting together we obtain   
%the largest $\# \mathcal{V}^{E}$  is smaller than
\[
\# \mathcal{V}^{E} \leq n^{-\Gamma} \delta^{-2 \Gamma} |\mathbf{m}|^{2\Gamma} \prod_{j>i}\frac{(\lambda_i-\lambda_j)^{m_{i,j}}}{m_{i,j}!} , \qquad \forall E ~{\rm with}~
\Gamma(E)= \Gamma .
\]

Multiplying by the number of possible $E$ yields the result.

\bigskip

\emph{Proof of \eqref{prod_id}.}
%\nopagebreak
We applying \eqref{formdet} with $U = \mathbf{1}$. 
Since both $\Va$ and $\Vb$ are product of basis vectors, the scalar product $\langle \Va \mid q_{\lambda} \Vb \rangle$ is equal to $-1$ or $1$ if $\tca([1,l(c)]) = \tcb([1,l(c)])$ for all columns, and $0$ otherwise. Here we denote by $\tca([1,l(c)])$ 
the set of entries $\{ \tca(1), \dots, \tca(l(c))\}$.

Now, since a modified  column cannot satisfy $\tca([1,l(c)]) = [1,l(c)]$ 
(and the same for ${\bf b}$), the vectors $\Va$ and $\Vb$ are orthogonal unless they have the same number of modified columns. Finally, that number is $|\mathbf{l}| - \Gamma^a$ for $\Va$ and $|\mathbf{m}| - \Gamma^b$ for $\Vb$. This yields the first line of \eqref{prod_id}.

We now concentrate on the case when $\Gamma^b = |\mathbf{m}| - |\mathbf{l}| + \Gamma^a$. Since  
 $|\langle \Va \mid q_{\lambda} \Vb \rangle|\leq 1$, we can bound the sum of scalar products by the number of non-zero inner products. The question is how many diagrams $f_{\bf b}$ have the same content (seen as an unordered set) in each column as $f_{\bf a}$: $t_{\bf a}^{c}([1,l(c)]) =t_{\bf b}^{c}([1,l(c)])$, or equivalently $S(\colmod_{\bf a}^c) = S(\colmod_{\bf b}^c)$.

For building the relevant $\Vb$, we can follow the algorithm with the further condition that, at stage three, all the column-modifiers are applied in such a way that the unordered column content is identical to that of $\Va$. 

The first two stages of the algorithm are the same so they yield a $C^{\Gamma^b}$ factor. We now have a collection $\{m_{\colmod}\}$ of column modifiers which have to be placed such that they match the column content of $f_{\bf a}$. 
For each $S$ we identify the column modifiers 
$\colmod_{1},\dots, \colmod_{r(S)}$ such that $S(\colmod_{i})=S$ for all 
$1\leq i\leq r(S)$. The total number of such objects is $m_{S}:= \sum_{i\leq r(S)} m_{\colmod_{i}}$ and the number of ways in which they can be inserted to produce {\it distinct} diagrams is 

%At the following stage of the algorithm,  we must ensure $\tca([1,l(c)]) = \tcb([1,l(c)])$, that is $S(\colmod_{\bf a}^c) = S(\colmod_{\bf b}^c)$, where we denote by $\colmod_{\{{\bf a},{\bf b}\}}^c$ the column-modifier applied on column $c$ of $\Va$ and $\Vb$,  respectively. 
$$\left(\begin{array}{c}
m_S \\ m_{\colmod_1} \dots m_{\colmod_{r(S)}} \end{array}\right).
$$ 
%choices for each $S$, where $S(\colmod_i) = S$ for each $1\leq i \leq k$.

Recall that the number of elementary column-modifiers $\sum_{i<j} m_{\colmod (i,j)}$ is at least $|{\bf m}|- 2\Gamma^{b}$. Moreover, each elementary column-modifier 
$\colmod (i,j)$ corresponds to a different $S(\colmod (i,j))= \{ (i,-), (j,+)\}$. Thus
$$
|{\bf m}|- 2\Gamma^{b} \leq \sum_{i<j} m_{\colmod(i,j)} \leq 
\sum_{S}  \max_{\colmod : S(\colmod) = S} m_{\colmod} .
$$
Since
$$
\sum_{S} m_{S} = \sum_{\colmod} m_{\colmod} = |{\bf m}| - \Gamma^{b},
$$
we obtain
$$
\sum_{S} \left( m_S - \max_{\colmod : S(\colmod) = S}  m_{\colmod}\right) \leq \Gamma^b.$$ 
This implies
\[
\prod_S \left(
 \begin{array}{c}
 m_S \\ 
 m_{\colmod_1} \dots m_{\colmod_{r(S)}} 
 \end{array}
\right) \leq |\mathbf{m}|^{\Gamma^b}.
\]
Multiplying by the $C^{\Gamma^{b}}$ of the first stages, we get \eqref{prod_id}.

\bigskip

\emph{Proof of \eqref{prod_id0}.}
%\nopagebreak
As shown above the only non-zero contributions come from 
$f_{\bf b}\in \mathcal{V}^{0}\subset \mathcal{O}_{\lambda}({\bf m})$. 

%We may use the same strategy as above, noticing first that $\langle \Va \mid q_{\lambda} \Vb \rangle = 0$ if  $\Vb \neq 0$, second that we must have the same modified columns. 
Since $\Gamma^b = 0$, the constant from the two first stages of the algorithm is $1$, $m_S = m_{i,j} = m_{\colmod({i,j})}$ for all $S$ corresponding to an elementary column-modifier, and $0$ otherwise. So the combinatorial factor is again one: we do not have any choice in our placement of column-modifiers. In other words, the only $\Vb$ such that  $\langle \Va \mid q_{\lambda} \Vb \rangle \neq 0$ is $\Va$. Finally, $\langle \Va \mid q_{\lambda} \Va \rangle = 1$.

\bigskip

\emph{Proof of \eqref{interg0}.}
%\nopagebreak
From \eqref{formdet} we deduce
$$
\langle \Va | q_{\lambda} U(\quant, \cartan, n)^{\otimes n} \Vzero\rangle
= \prod_{1\leq c\leq \lambda_{1}} {\rm det} (U^{t_{\bf a}^{c}, {\rm Id}^{c}}) ,
\qquad  U= U(\quant,\cartan,n).
$$
We will use the Taylor expansion of the unitary  $U(\quant,\cartan,n)$ to estimate the above determinants.

 Entry-wise, for all $1\leq i\leq d$ on the first line, and all $1\leq
i <j \leq d$ on the second and third lines:
\begin{eqnarray*}
U_{i,i}(\quant,\cartan,n)  & =& 1 + i \frac{\cartansub_i \delta_{i\neq d} -
\cartansub_{i-1}\delta_{i\neq1}}{\sqrt{n}} - \frac1{2n}\sum_{j\neq i}
\frac{|\quantsub_{i,j}|^2}{|\mu_i
-\mu_j|}
\\
& &  \qquad +  O(\|\quant\|^3 n^{-3/2} \delta^{-3/2}, \|\quant\| \|\cartan\| n^{-1} \delta^{-1/2}, \| \vec{\xi}\|^{2}n^{-1}) ;
\\
&&\\
U_{i,j}(\quant,\cartan,n)  & =& - \frac1{\sqrt{n}} \frac{\quantsub_{i,j}^*}{\sqrt{\mu_i -\mu_j}} + O(\|\quant\|^2
n^{-1} \delta^{-1}, \|\quant\| \|\cartan\| n^{-1} \delta^{-1/2}) ;
\\
U_{j,i}(\quant,\cartan,n)  & = & \frac1{\sqrt{n}} \frac{\quantsub_{i,j}}{\sqrt{\mu_i -\mu_j}} + O(\|\quant\|^2
n^{-1} \delta^{-1}, \|\quant\| \|\cartan\| n^{-1} \delta^{-1/2}).
\end{eqnarray*}
If $\quant =O(n^{\beta})$, $\|\cartan\| \leq n^{-1/2 + 2\beta}/\delta$, and 
$\beta < 1/2$, 
the remainder terms are $O(n^{-3/2+3\beta}\delta^{-3/2})$ for the first line and 
$O(n^{-1+2\beta}\delta^{-1})$ for the last two lines. 
%Moreover, if $\beta< ???$ then the dominant term is $O(n^{-1+2\beta}\delta^{-1})$.

Therefore, when our parameters are in this range, we can give precise enough
evaluations of the determinants. The idea is to find the dominating terms in the expansion of the determinant 
$$
\det A
=\sum_{\sigma} \prod_{i} \epsilon(\sigma) A_{i,\sigma(i)}.
$$

Note that we can use the above Taylor expansions inside the determinant since the number of terms in the product is at most $d$.

Since $\Va \in \mathcal{V}^0$, all $\tca$ are either ${\rm Id}^{c}$, or an 
\emph{$(i,j)$-substitution}.
% that is $\tca(r) =
%j\delta_{r=i} + r \delta_{r\neq i}$ for $i\leq l(c) <j$.
%We only need to estimate the corresponding types of determinants. 
If $\tca = {\rm Id}^{c}$,
the summands with more than two non-diagonal terms are of the same order as the remainder term, so that only the identity and the transpositions count in $\sum_{\sigma } \prod_{i} A_{i,\sigma (i)}$. Let $l=l(c)$, then
\[
\upsilon(l):= \det(U^{{\rm Id}^{c}, {\rm Id}^{c}}(\quant,\cartan, n)) =  
%\det(U^{l(c),{\rm Id}^{c}, {\rm Id}^{c}}(\quant,\cartan, n))
1 + i
\frac{\cartansub_{l}}{\sqrt{n}} - \frac1{2n}\sumtwo{1\leq i \leq l}{l+1\leq j\leq d}
\frac{|\quantsub_{i,j}|^2}{\mu_i -\mu_j} +
O(n^{-3/2 + 3 \beta}\delta^{-3/2} ).%+ n^{-1+2\beta}\delta^{-1} ).
\] 
Note that for $l=d$, we get the usual determinant of $U(\quant,\cartan, n)$ which is $1$.

Consider now the case $\tca\neq {\rm Id}^{c}$. Since $\tca(r)\geq r$ for all $r$,  there 
exists a whole column of $U^{\tca,{\rm Id}^{c}}$ whose entries are smaller in
modulus than $O(\|\quant\|/\sqrt{n\delta}) = O(n^{-1/2+\beta}\delta^{-1})$. 
%The same bound holds for the determinant.
In particular if $\tca$ is an
\emph{$(i,j)$-substitution}, then the only summand that is of this
order comes from the identity. So that 
\begin{equation}
\label{uij}
\upsilon(i,j): =\det(U^{\tca, {\rm Id}^{c}}(\quant,\cartan, n)) = 
\frac{\quantsub_{i,j} }{\sqrt{n}\sqrt{\mu_i -\mu_j}} +
O(n^{-1 + 2\beta}\delta^{-1}).%+n^{-3/2+3\beta}\delta^{-3/2}).
\end{equation}
Note that this approximation does not depend on $l(c)$, but only on $i$ and $j$.

We now put together the estimated determinants in the product \eqref{formdet}. For each $i < j$ there are  $m_{i,j}$ columns of the type $(i,j)$-substitution. Out of the $\lambda_l-\lambda_{l+1}$ columns of length $l=l(c)$ there are 
$\lambda_l-\lambda_{l+1} - R_l$ of the type ${\rm Id}^{c}$, with $0\leq R_l\leq
\lvert{\bf m}\rvert$.

%The reason of these assertions is that there are $m_{i,j}$ boxes with
%a $j$ in row $i$, and if a column has no such substitution, then its entry in
%row $i$ is $i$, and $\tca = Id_c$. 
Hence:
\begin{align}
\label{interg0_et}
\langle \Va | q_{\lambda} U(\quant, \cartan, n)^{\otimes n} \Vzero\rangle 
& = \prod_{l=1}^d\left( \upsilon(l)) \right)^{\lambda_l -\lambda_{l+1}}\prod_{1\leq i<j\leq d}
\left(\upsilon(i,j) \right)^{m_{i,j}} \prod_{l=1}^d (\upsilon(l))^{-R_l}.
\end{align}
Now $\upsilon(l) = 1 + O(n^{-1 + 2\beta}\delta^{-1})$ 
and $R_l\leq | {\bf m}| \leq n^{\eta}$, so the last product is 
$1+O(n^{-1 + 2 \beta + \eta}\delta^{-1} )$.
Similarly, since $\lambda\in\Lambda_{n,\alpha}$ we have $\lambda_l-
\lambda_{l+1} = n (\mu_l -\mu_{l+1}) + O(n^{\alpha})$, and we can use 
Lemma \ref{Taylor} given at the end of this section to estimate the first product  as follows
\begin{align*}
\prod_{l=1}^{d} \upsilon(l)^{\lambda_l -\lambda_{l+1}} 
& = \prod_{l=1}^{d}
\exp\left(i\phi_l  - \frac1{2}\sumtwo{1\leq i\leq l }{l+1\leq j\leq d}
|\quantsub_{i,j}|^2 \frac{\mu_l-\mu_{l+1}}{\mu_i -\mu_j}
\right) r(n) \\
&= \exp\left(i\phi - \frac{\|\quant \|^2_2}{2}\right) r(n),
\end{align*}
with 
\begin{eqnarray*}
\tilde{r}(n)  &=& 1+O(n^{-1 + \alpha + 2\beta} \delta^{-1}, n^{-1/2 + 2\beta}\delta^{-1}), \\
\phi_l   &=& \delta_{l \neq d}\sqrt{n} (\mu_l - \mu_{l+1}) \cartansub_l, \\
\phi &=& \sqrt{n}  \sum_{l=1}^{d-1}  (\mu_l - \mu_{l+1}) \cartansub_l.
\end{eqnarray*}

We now turn our attention to the middle product on the right side of \eqref{interg0_et}
\begin{align*}
\upsilon(i,j)^{m_{i,j}}  & = 
\left(\frac{\quantsub_{i,j}}{\sqrt{n}\sqrt{\mu_i -\mu_j}} \right)^{m_{i,j}}\left(1 +
O\left(n^{-1 + 2 \beta+\eta}\delta^{-1}\right)\right),
\end{align*}
where we have used that $|\mathbf{m}|\leq n^{\eta}$.

Inserting into \eqref{interg0_et} yields \eqref{interg0}.\\ Note that 
$\langle \Va | q_{\lambda} U(\quant, \cartan, n)^{\otimes n} \Vzero\rangle=0$ if there exist $i<j$ such that $\zeta_{i,j}=0$ and $m_{i,j}\neq 0$ .

\bigskip
%%%%%%%%%%%%%%%%%%%%%%%%%%%%%%%%%%%%%%%%%%%%
\emph{Proof of \eqref{interg}.}
%\nopagebreak
We may write, much like in (\ref{interg0_et}),
\[
\langle \Va | q_{\lambda} U(\quant, \cartan, n)^{\otimes n} \Vzero\rangle 
 = \prod_{l=1}^d\left( \upsilon(l)) \right)^{\lambda_l -\lambda_{l+
1}}\prod_{\colmod}
\left(\upsilon(\colmod) \right)^{m_{\colmod}} \prod_{l=1}^d (\upsilon(l))^{-R_l}
\]
where $ 0\leq R_l\leq |{\bf m}| -\Gamma$ and $\upsilon(\colmod)$ is the determinant of the minor of $U$ corresponding to having applied the column-modifier $\colmod$. We can further split the column-modifers into elementary ones $\colmod(i,j)$ and non-elementary ones $\colmod^{\prime}$.

Then $\langle \Va | q_{\lambda} U(\quant, \cartan, n)^{\otimes n} \Vzero\rangle$ can be written as
\[
%\langle \Va | q_{\lambda} U(\quant, \cartan, n)^{\otimes n} \Vzero\rangle 
%=
\prod_{l=1}^d\left( \upsilon(l)) \right)^{\lambda_l -\lambda_{l+1}}  
\prod_{i<j}
\left(\upsilon(i,j) \right)^{m_{\colmod(i,j)}} 
\prod_{l=1}^d (\upsilon(l))^{-R_l}
\prod_{\colmod^{\prime}} 
\left(\upsilon(\colmod^{\prime}) \right)^{m_{\colmod^{\prime}}} .
\]

The first three products on the right side can be treated as above. 
For the fourth product we give a rough upper bound based on the following observation. If the entries in the column have been modified in an admissible way, then 
$\tca (i) =j>l(c)$ for some $i$, so that $|\upsilon(\colmod)|\leq C\|\quant\|/\sqrt{n\delta}$ for any $\colmod$, with some constant $C=C(d)$. 

Thus by using the previous point 
\begin{eqnarray}
&&
\left|\langle \Va | q_{\lambda} U(\quant, \cartan, n)^{\otimes n} \Vzero\rangle \right| \leq 
\nonumber\\
&&
\mbox{}\hspace{2cm} \exp\left(-\frac{\| \quant \|_{2}^{2} }{2}\right) 
\left(\frac{ C\|\quant\|}{\sqrt{n\delta}}\right)^
{\sum_{\colmod^{\prime}} m_{\colmod^{\prime}}}
\prod_{i<j}\left(\frac{|\quantsub_{i,j}|}{\sqrt{n} \sqrt{\mu_{i}-\mu_{j}}}\right)^{m_{\colmod(i,j)}} r(n).
\label{gg} 
\end{eqnarray}
We obtain \eqref{interg} by noting that the number of non-elementary modifiers is 
$$
\sum_{\colmod^{\prime}} m_{\colmod^{\prime}} = - \Gamma +
\sum_{i<j} (m_{i,j} - m_{\colmod(i,j)} ).
$$ 

%Here we have tacitly assumed that $\zeta_{i,j}\neq 0$. A simple argument shows 
%that  if there exist $i<j$ for which $(\zeta_{i,j}=0, m_{\colmod(i,j)}\neq 0)$ then
%$Z(E)=0$, and if $(\zeta_{i,j}=0, m_{\colmod(i,j)}= 0)$ then the corresponding term in the product is one, with the convention $0^{0}=1$.

\bigskip

%%%%%%%%%%%%%%%%%%%%%%%%%%%%%%%%%%%%%%%
\emph{Proof of \eqref{numer}.}
Note that only admissible vectors in $\mathcal{O}_{\lambda}({\bf m})$ can bring 
non-zero contributions. 
We shall split the sum into  sub-sums using 
$\mathcal{O}_{\lambda}({\bf m}) \bigcap \mathcal{V}= \bigcup_{E} \mathcal{V}^{E}({\bf m})$, and compare each sub-sum against the benchmark $\mathcal{V}^{0}=\mathcal{V}^{E^{0}} $. 

%\nopagebreak
From the bounds on $\quant$ and $\vec{z}$ we obtain 
$\|\quant + \vec{z}\| =O(n^{\beta})$, so we can apply the previous points with 
$\quant + \vec{z}$ instead of $\quant$.% ,  and $C$ replaced by $C+1$ .

%We merely combine some of the previous entries of the lemma, after noticing that 
%$(\quant + z)$ plays the same role as $\quant$ with the new constant $C+1$, that is 
%$\|\quant + z\| \leq (C+1) n^{\beta}$. So that all the former bounds in the lemma remain valid with $\quant + z$ instead of $\quant$.

Using \eqref{cg0} and \eqref{interg0} and recalling that $\lambda \in \Lambda_{n,\alpha}$, we get:
\begin{eqnarray*}
&&\left\langle \sum_{\Va \in \mathcal{V}^0} \Va \Bigg| q_{\lambda} U(\quant + \vec{z}
, \cartan, n)^{\otimes n} \Vzero \right\rangle  %\\ 
= 
 \exp\left(i\phi-\frac{\|\quant + \vec{z}\|^2_2}{2}\right) \prod_{i<j}
 \frac{\big((\zeta_{i,j} + {z}_{i,j})\sqrt{n}\sqrt{\mu_i - \mu_j}\big)^{m_{i,j}}}{m_{i,j}!} r(n)
\end{eqnarray*}
with error factor
\[
r(n)  = 1 + O\left(n^{-1 + 2\beta + \eta}\delta^{-1}, n^{-1/2 + 2\beta} \delta^{-1}, 
n^{-1 +2\beta +\alpha}\delta^{-1}, n^{-1 + 2\eta}\delta^{-1}, n^{-1 + \alpha + \eta}\delta^{-1} \right).
\]

For $E \neq E^{0}$ we combine \eqref{interg} and \eqref{cggam}  to obtain
\begin{eqnarray*}
&&
 \left|\left\langle 
\sum_{\Va \in \mathcal{V}^{E}} \Va \Bigg| q_{\lambda} U(\quant + \vec{z}, \cartan, n) \Vzero \right\rangle \right| \cdot
\left|
\left\langle 
\sum_{\Va \in \mathcal{V}^0} \Va \Bigg| q_{\lambda} U(\quant + \vec{z}, \cartan, n) \Vzero 
\right\rangle
 \right|^{-1}  \\ 
&& \leq
n^{-\Gamma} \prod_{i<j} \left(\frac{\lambda_i - \lambda_j}{n}\right)^{m_{\colmod({i,j})} - m_{i,j}}\frac{m_{i,j}!}{m_{\colmod({i,j})}!} \left(\frac{\|\quant + \vec{z}\|}{\sqrt{\delta n}}\right)^{- \Gamma}\prod_{i<j}\left(\frac{\sqrt{\delta n} |\zeta_{i,j} + z_{i,j}|}{\|\quant + z\|\sqrt{n}\sqrt{\mu_i - \mu_j}}\right)^{m_{\colmod({i,j})} - m_{i,j}} r(n) \\
&& \leq O(n^{-\Gamma (1/2 +\beta)})\delta^{-\Gamma/2} 
\prod_{i<j :m_{i,j}\neq 0}\left(
\frac{
|\quantsub_{i,j} +z_{i,j} |\sqrt{\mu_i - \mu_j}}{m_{i,j} \|\quant + \vec{z}\|}\right)^{m_{\colmod(i,j)} - m_{i,j}  }  \\
&& \leq O\big((2\delta^{-3/2} n^{-1/2 + 3 \beta + 2 \epsilon})^\Gamma\big),
\end{eqnarray*}
with $O(\cdot)$ uniform in $\Gamma$. In the second inequality we used  
$$
m_{i,j}!/m_{\colmod(i,j)}! \leq m_{i,j}^{m_{i,j}-m_{\colmod(i,j)}}, \qquad 
\sum_{i<j} (m_{\colmod({i,j})} - m_{i,j}) \geq - 2 \Gamma, \qquad 
\lambda\in \Lambda_{n,\alpha}
%\|\quant + z\| = O(n^\beta), 
$$ 
and in the third inequality we used
$$
%\sum_{i<j} (m_{\colmod({i,j})} - m_{i,j}) \geq - 2 \Gamma, \qquad 
m_{i,j}\leq 2 |\quantsub_{i,j}+ z_{i,j}| n^{\beta + \epsilon} ,
\qquad 
\frac{
|\quantsub_{i,j} +z_{i,j} |\sqrt{\mu_i - \mu_j}}{m_{i,j} \|\quant + \vec{z}\|}\leq 1 .
$$

Furthermore, for a given $\Gamma$, there are at most $C^{\Gamma}$ different $E$ such that $\Gamma(E) = \Gamma$, corresponding to the possible choices in the first two stages of the algorithm, where $C=C(d)$. Hence, if $n$ is large enough such that 
$ 2C\delta^{-3/2} n^{-1/2 + 3 \beta + 2 \epsilon} <1 $, we have:

\begin{eqnarray*}
&&
\left\langle \sum_{\Va \in \mathcal{O}_{\lambda}(\mathbf{m})} \Va \Bigg| q_{\lambda} U(\quant + z, \cartan, n) \Vzero \right\rangle %\\
%&&
 = \sum_{\Gamma}\left\langle \sum_{\Va \in \mathcal{V}^{\Gamma}} \Va \Bigg| q_{\lambda} U(\quant + z, \cartan, n) \Vzero \right\rangle \\
&&
 = \Big(1 + O(\delta^{-3/2} n^{-1/2 + 3 \beta + 2 \epsilon})\Big)
\exp\left(i\phi -\frac{\|\quant + z\|^2_2}{2}\right) \prod_{i<j}\frac{\big((\quant + z)_{i,j}(\sqrt{n}\sqrt{\mu_i - \mu_j})\big)^{m_{i,j}}}{m_{i,j}!} r(n) \\
&& =
\exp\left(i\phi-\frac{\|\quant + z\|^2_2}{2}\right) \prod_{i<j}\frac{\big((\quant + z)_{i,j}(\sqrt{n}\sqrt{\mu_i - \mu_j})\big)^{m_{i,j}}}{m_{i,j}!} r_2(n)
\end{eqnarray*}
where the sum over $\Gamma$ was bounded using a geometric series and   
\[
r_2(n) = 1 + O\left(n^{-1 + 2\beta + \eta}\delta^{-1}, 
%n^{-1/2 +2\beta}\delta^{-1}, 
n^{-1 + \alpha + \beta}\delta^{-1}, n^{-1 + 2\eta}\delta^{-1}, n^{-1 + \alpha + \eta}\delta^{-1}, \delta^{-3/2} n^{-1/2 + 3 \beta + 2 \epsilon} \right).
\]
This is exactly \eqref{numer}.

\bigskip

\emph{Proof of \eqref{denom}.}
%\nopagebreak
We choose $\Gamma^{a}$ and $\Gamma^{b}$ satisfying the condition 
$\Gamma^b - \Gamma^a = |{\bf m}| - |{\bf l}|$ under which the inner products in 
\eqref{prod_id} are non-zero. By multiplying \eqref{cgG} and \eqref{prod_id}, we see that:
\begin{align}  
\label{refgam}
\left\vert\left\langle \sum_{\Va \in \mathcal{V}^{\Gamma^{a}}(\mathbf{l})}\Va \Bigg| q_{\lambda}  \sum_{\Vb \in \mathcal{V}^{\Gamma^{b}}({\bf m})
%\mathcal{O}_{\lambda}(|\mathbf{m}|)
}\Vb \right\rangle \right\vert
& \leq (C |\mathbf{m}|)^{\Gamma^b} \!\prod_{i<j} \frac{(\lambda_i - \lambda_j)^{l_{i,j}}}{l_{i,j}!} \left(\frac{C|\mathbf{l}|^2}{n\delta^2}\right)^{\Gamma^a} \\
& = (C |\mathbf{m}|)^{|\mathbf{m}| - |\mathbf{l}|}\prod_{i<j} \frac{(\lambda_i - \lambda_j)^{l_{i,j}}}{l_{i,j}!} \left(\frac{C|\mathbf{l}|^2|\mathbf{m}|}{n\delta^2}\right)^{\Gamma^a} .\notag
\end{align}

It remains to sum up the upper bounds over all relevant pairs 
$(\Gamma^{a},\Gamma^{b})$. If $n^{1 - 3 \eta} >  2 C / \delta^{2}$, the dominating term in the sum of bounds is that corresponding to the smallest possible $\Gamma^a$. The question is, what is the smallest possible value of $\Gamma^a$ leading to non-zero inner products?

A necessary condition for $\Va$ not to be orthogonal to $\Vb$ is that for each set 
$S$ of suppressed and added values, the two vectors have the same multiplicities $m_S^{\bf a} = m_S^{\bf b}$.

The following argument provides a lower bound for $\Gamma(\Va) + \Gamma(\Vb)$. The idea is to count the minimum number of `horizontal box shuffling' operations necessary 
in order to transform a Young tableau $t_{\bf a^{\prime}}\in \mathcal{O}_{\lambda}({\bf m})$ into the tableau $t_{\bf a}$. Since $|{\bf m}| \leq n^{\eta}$ and  $\lambda_{d} \geq \delta n + O(n^{\alpha})$, the tableau $t_{\bf a^{\prime}}$ can be chosen to have at most one modified box per column (thus $\Gamma(f_{{\bf a}^{\prime}}) =0$), and such that each of the modified columns of $t_{\bf a}$ are also modified in $t_{\bf a^{\prime}}$. We also choose  $t_{\bf b^{\prime}}$ in a similar 
fashion. 

Now at each step we horizontally move one elementary column modifier $\colmod(i,j)$ of $t_{{\bf a}^{\prime}}$ (or $t_{{\bf b}^{\prime}}$) into an already modified 
column, with the aim of constructing $t_{\bf a}$ (or $t_{\bf b}$). 

Each such operation increases $\Gamma(f_{{\bf a}^{\prime}}) + \Gamma(f_{{\bf b}^{\prime}})$ by one. On the other hand the operation has the following effect on the 
$m_{S}^{{\bf a}^{\prime}}$ (or $m_{S}^{{\bf b}^{\prime}}$): the multiplicities
$m_{\{(i,-),(j,+)\}}$ and $m_{S_{0}}$ decrease by one, and $m_{S_{0} +  \{(i,-),(j,+)\}}$ increases by one. Here $S_{0}$ is the signature of the column to which the box $(i,j)$ is moved. Hence the distance 
$\sum_{S} |m_S^{{\bf a}^{\prime}} - m_S^{{\bf b}^{\prime}}|$ 
decreases by at most three. Since initially this quantity was equal to  
$\sum_{i<j} | l_{i,j} - m_{i,j}|$, we need at least $\sum_{i<j} |{l}_{i,j} - {m}_{i,j}| / 3$ such operations before reaching our goal 
$m_S^{\bf a} = m_S^{\bf b}$. This means that $\Gamma(\Va) + \Gamma(\Vb) \geq |\mathbf{l} - \mathbf{m}| / 3$.

%
%%On the one hand, we know that $\Gamma^b - \Gamma^a = |{\bf m}| - |{\bf l}|$.
%On the other hand, we can bound $\Gamma(\Va) + \Gamma(\Vb)$ from below. Indeed, this quantity increases by one if and only if we put another $(i,j)$ brick in a column that was already modified (say with $S_1$). Now such an operation has the following effect on the $m_{S}$ (or $l_{S}$) : the $m_{(i,-),(j,+)}$ and $m_{S_1}$ both decrease by one, and $m_{S_1 +  ((i,-),(j,+))}$ increases by one. Hence the distance $\sum_{S} |\mathbf{l}_S - \mathbf{m}_S|$ decrease by at most three. We thus need at least $\sum_{i<j} |\mathbf{l}_{i,j} - \mathbf{m}_{i,j}| / 3$ such operations before getting the equalities $\mathbf{m}_S = \mathbf{l}_S$. That is,  $\Gamma(\Va) + \Gamma(\Vb) \geq |\mathbf{l} - \mathbf{m}| / 3$.

Together with $\Gamma^b - \Gamma^a = |{\bf m}| - |{\bf l}|$,  this result yields $\Gamma^a \geq (|\mathbf{l} - \mathbf{m}| + 3 |\mathbf{l}| - 3 |\mathbf{m}|) / 6 $. Moreover $\Gamma^a$ is non-negative.

Replacing in the above equation yields \eqref{denom}.

\bigskip

\emph{Proof of \eqref{denom54}.}
%\nopagebreak
Since $\mathbf{l} = \mathbf{m}$,  equations \eqref{cg0} and \eqref{prod_id0} prove that the bound \eqref{refgam} is saturated when $\Gamma^a = 0$, up to the error factor $\left(1 + O(n^{-1 + 2 \eta}/\delta)\right)$. Hence the remainder 
term due to the other $\Gamma$ consist in a geometric series with factor 
$\left(\frac{C|\mathbf{m}|^3}{n\delta^2}\right) = O(n^{1 - 3\eta}/\delta^{2})$. 

\end{proof}

The only part of the proof we have still postponed is the following
technical lemma:

\begin{lem}
\label{Taylor}
If $x_{n}= O(n^{1/2-\epsilon})$, then
\[
\left(1+\frac{x_{n}}{n}\right)^n  = \exp(x_{n})(1+ O(n^{-\epsilon})).
\]
\end{lem}
\begin{proof}
For simplicity we will ignore the dependence on $n$ and write $x=x_{n}$. 

For any $y$ such that $|y|\leq 1$, for any $n\in \mathbb{N}$ 
%(in fact even for any complex number)
, we have the Taylor expansion:
\[
(1+y)^n = \sum_{k=1}^{\infty} {n \choose k} y^k.
\]
Now $(n-k)^k/k! \leq {n \choose k} \leq n^k/k! $ for $n\geq k$. If $k\leq
n^{1/2-\epsilon/2}$, then $(n-k)^k = n^k (1+ O(n^{-\epsilon}))$. If $k\geq
n^{1/2-\epsilon/2}$, then $n^k/k! = O(n^{(1/2+\epsilon/2)k})$.
So that if $y = x/n = O(n^{-1/2-\epsilon})$, 
\begin{align*}
(1+x/n)^n & = (1+O(n^{-\epsilon})) \sum_{k=0}^{n^{1/2-\epsilon/2}} \frac{x^k}{k!}
+ \sum_{k> n^{1/2-\epsilon/2} } O(n^{(1/2+ \epsilon/2)k} (x/n)^k
\\
& = (1+O(n^{-\epsilon})) \exp(x) 
+ \sum_{k> n^{1/2-\epsilon/2} } (O(n^{(1/2+ \epsilon/2)k}-n^{k}/k!) (x/n)^k
\\
&=  (1+O(n^{-\epsilon})) \exp(x) 
+ O(e^{-n^{1/2-\epsilon/2}})
\\
&= (1 + O(n^{-\epsilon})) \exp(x),
\end{align*}
as $\exp(x) \geq  C \exp( - n^{1/2-\epsilon}))$ for some constant $C>0$.
\end{proof}

%%%%%%%%%%%%%%%%%%%%%%%%%%%%%%%%%%%%%%%%%%%
\subsection{Proof of Lemma \ref{V_approx} and non-orthogonality issues}
\label{preuve_quasi_orth}
%%%%%%%%%%%%%%%%%%%%%%%%%%%%%%%%%%%%%%%%%%%%%%

\begin{lem}
\label{non-orth}
Let $(\mathbf{m}, \lambda)$ and $(\mathbf{l}, \lambda)$ be semistandard Young tableaux with diagram $\lambda$ and define $|{\bf m}|:=\sum_{i<j}m_{ij}$ and 
$|{\bf l}-{\bf m}| := \sum_{i<j} | l_{i,j}-m_{ij}|$.

If 
$$
\sum_{j > i} m_{i,j} - \sum_{j < i} m_{j,i} \neq \sum_{j > i} l_{i,j} - \sum_{j < i} l_{j,i}
$$ for some $1\leq i \leq d$, then   
\[
\langle \mathbf{m}, \lambda | \mathbf{l}, \lambda \rangle  = 0.
\]

Otherwise, we derive an upper bound under the following conditions.
We assume that $\lambda_i - \lambda _{i+1} > \delta n $ for all $1\leq i \leq d-1$  and $\lambda _d > \delta n$, for some $\delta > 0$. Furthermore we assume 
$|\mathbf{l}| \leq |\mathbf{m}| \leq n^{\eta}$ for some $\eta < 1/ 3$ and that
 $Cn^{3\eta - 1}/\delta^{2} <1$ where $C=C(d)$ is a  constant.

Then: 
\begin{equation}\label{quasiorth1}
%\left| \langle \mathbf{m},\lambda | \mathbf{l}, \lambda \rangle \right|  \leq 
%(Cn)^{- |\mathbf{m} - \mathbf{l}| / 6 } (Cn)^{\eta (3 |\mathbf{m} - \mathbf{l}| - (|\mathbf{m}| - |\mathbf{l}|)) / 4} \delta^{|\mathbf{m}| -|\mathbf{l}| - |\mathbf{m} - \mathbf{l}|  / 3} (1 + O(n^{-1 + 3\eta}/\delta)).
\left|\langle \mathbf{m}, {\lambda} | \mathbf{l}, {\lambda} \rangle\right| \leq
 (C^{\prime}n)^{-\eta( |{\bf m}|-|{\bf l}|)/4}\,
  (C^{\prime}n)^{(9\eta-2)|{\bf m} -{\bf l}|/12}\,
  \delta^{ (|\mathbf{m}| -|\mathbf{l}|)/2 - |\mathbf{m} - \mathbf{l}|  / 3}\,
   (1+O(n^{-1+3\eta} /\delta)).
\end{equation}
where $C^{\prime}=C^{\prime}(d,\eta)$ and the constant in the remainder term depends only on $d$. The right side is of order less than $n^{(9 \eta -2)|{\bf m} - {\bf l}| / 12 }$ and converges to zero for $\eta < 2 / 9$  when $n\to \infty$.

\end{lem}

{\it Proof.} 
We know that $\left\vert \mathbf{m} ,{\lambda}\right\rangle$ is a linear combination 
of $n$-tensor product vectors in which the basis vector $f_i$ appears exactly 
$\lambda_i - \sum_{j > i} m_{i,j} + \sum_{j < i} m_{j,i}$ times. As two tensor basis vectors are orthogonal if they do not have the same number of $f_i$ in the decomposition, we get that $\langle \mathbf{m}, {\lambda} | \mathbf{l}, {\lambda} \rangle = 0$ if $\sum_{j > i} m_{i,j} + \sum_{j < i} m_{j,i} \neq \sum_{j > i} l_{i,j} + \sum_{j < i} l_{j,i}$ for any $1 \leq i \leq d$.

\medskip

In the general case,
\begin{align}
\langle \mathbf{m}, {\lambda} | \mathbf{l}, {\lambda} \rangle = \frac{\langle q_{\lambda} p_{\lambda} \Vm | q_{\lambda} p_{\lambda} \Vl\rangle }{\sqrt{\langle q_{\lambda} p_{\lambda} \Vm | q_{\lambda} p_{\lambda} \Vm\rangle \langle q_{\lambda} p_{\lambda} \Vl | q_{\lambda} p_{\lambda} \Vl\rangle}}.
\end{align}

We use the fact that $q_{\lambda}$ is a projection, up to a constant factor
(cf. \eqref{eq.p.q.square},\eqref{csq1}), and erase the $q_{\lambda}$ at the left of each scalar product, and we decompose $p_{\lambda} f_{\bf m} $ and $p_{\lambda} f_{\bf l}$ 
on orbits as in \eqref{plambda.orbit.decomposition}. 
Since the multiplicity of the elements in the orbits are the same in numerator and denominator, we end up with:
\begin{align}
\langle \mathbf{m}, {\lambda} | \mathbf{l}, {\lambda} \rangle = \frac{\langle \sum_{\Va \in \mathcal{O}_{\lambda}(\mathbf{m})} \Va 
| q_{\lambda} \sum_{\Vb \in \mathcal{O}_{\lambda}(\mathbf{l})} \Vb \rangle}
{\langle \sum_{\Va \in \mathcal{O}_{\lambda}(\mathbf{m})} \Va | 
q_{\lambda}\sum_{\Vap \in \mathcal{O}_{\lambda}(\mathbf{m}) } \Vap\rangle
\langle \sum_{\Vb \in \mathcal{O}_{\lambda}(\mathbf{l})} \Vb | q_{\lambda} \sum_{\Vbp \in \mathcal{O}_{\lambda}(\mathbf{l})} \Vbp \rangle}
\end{align}

We use \eqref{denom54} for the denominator:
% under the assu $\lambda \in \Lambda_{n,\alpha}$, with $|\mathbf{l}|$ and $|\mathbf{m}| \leq n^{\eta}$ and $n^{1 - 3\eta} > 2C /\delta$ with $C$ depending only on the dimension $d$:
\begin{eqnarray*}
&&
\left\langle \sum_{\Va \in \mathcal{O}_{\lambda}(\mathbf{m})} \Va \Bigg| 
q_{\lambda}\sum_{\Vap \in \mathcal{O}_{\lambda}(\mathbf{m}) } \Vap\right\rangle
\Bigg\langle \sum_{\Vb \in \mathcal{O}_{\lambda}(\mathbf{l})} \Vb \Bigg| q_{\lambda} \sum_{\Vbp \in \mathcal{O}_{\lambda}(\mathbf{l})} \Vbp \Bigg\rangle \\
&&=  
\prod_{1 \leq i < j \leq d}\frac{ (\lambda_i -\lambda_j)^{(m_{i,j} + l_{i,j}) / 2}}{\sqrt{m_{i,j}! \,l_{i,j}!}}
(1+O(n^{3\eta - 1}/\delta))),
\end{eqnarray*}
and the numerator is bounded as in \eqref{denom}. Then, under the assumption 
$|{\bf m} | \geq |{\bf l}|$ we have
\[
\left|\langle \mathbf{m}, {\lambda} | \mathbf{l}, {\lambda} \rangle\right| 
\leq 
(C |\mathbf{m}|)^{|\mathbf{m}| - |\mathbf{l}|} \left(\frac{C |\mathbf{m}|^3}{\delta ^2 n}\right)^{\Gamma_{min}} \cdot \prod_{i < j} (\lambda_i - \lambda_j)^{(l_{i,j} - m_{i,j}) / 2} \sqrt{\frac{m_{i,j}!}{l_{i,j}!}} \cdot \left(1 + \left( O(n^{3\eta-1}/\delta) \right)\right), 
\]
where $\Gamma_{min} = (   (|\mathbf{l} - \mathbf{m}| + 3 |\mathbf{l}| - 3 |\mathbf{m}|) / 6) \wedge 0$.

The factorials can be bounded as 
$$
\prod_{i < j} \sqrt{\frac{m_{i,j}!}{l_{i,j}!}} \leq |\mathbf{m}|^{\sum (m_{i,j} - l_{i,j})_{+}/2} = 
|\mathbf{m}|^{(|\mathbf{m} - \mathbf{l}| + |\mathbf{m}| - |\mathbf{l}|) / 4 }.
$$ 
Since $|{\bf m}|\leq n^{\eta}$ and  $Cn^{3\eta-1}/\delta^{2}< 1$, we have
$$
\left(\frac{C |\mathbf{m}|^3}{\delta ^2 n}\right)^{\Gamma_{min}} \leq
\left(\frac{C n^{3\eta-1}}{\delta ^2}\right)^{  (|\mathbf{l} - \mathbf{m}| + 3 |\mathbf{l}| - 3 |\mathbf{m}|) / 6}.
$$
Since $\lambda_{i}-\lambda_{j} >n \delta $ we have
$$
\prod_{1 \leq i < j \leq d} (\lambda_i -\lambda_j)^{(l_{i,j} - m_{i,j}) / 2} \leq
(n \delta )^{ (|{\bf l}|-|{\bf m}|)/2 }.
$$

%and we have assumed $|\mathbf{l}| \leq |\mathbf{m}| \leq n^{\eta}$ with $\eta \leq 1/3$. Notably, we may forget that $\Gamma_{min}$ is non-negative, since we take an upper bound and $C|{\bf m}|/(\delta^2 n) < 1$. 

The constant $C=C(d)$ can be replaced by another constant 
$C^{\prime}= C^{\prime}(d,\eta)$ such that all powers of $n$ appear in the form $(C^{\prime}n)^{\gamma}$. Putting the bounds together we get
%\begin{align}
%\left|\langle \mathbf{m}, {\lambda} | \mathbf{l}, {\lambda} \rangle\right| 
%& \leq  \delta^{- 2\Gamma_{min}} (Cn)^{(|\mathbf{l}| - |\mathbf{m}|) / 2 - \Gamma_{min}} (C|\mathbf{m}|)^{(|\mathbf{m} - \mathbf{l}| + 5(   |\mathbf{m}| - |\mathbf{l}|)) / 4 + 3 \Gamma_{min})} (1 + O(n^{-1 + 3\eta}/\delta))
%\notag
%\\
%\label{quasiorth1}
%& \leq  \delta^{|\mathbf{m}| -|\mathbf{l}| - |\mathbf{m} - \mathbf{l}|  / 3} (Cn)^{- |\mathbf{l} - \mathbf{m}| / 6 } (Cn)^{\eta (3 |\mathbf{l} - \mathbf{m}| - (|\mathbf{l}| - |\mathbf{m}|)) / 4} (1 + O(n^{-1 + 3\eta}/\delta)),
%\end{align}
\begin{eqnarray*}
\left|\langle \mathbf{m}, {\lambda} | \mathbf{l}, {\lambda} \rangle\right| \leq
 \delta^{ (|\mathbf{m}| -|\mathbf{l}|)/2 - |\mathbf{m} - \mathbf{l}|  / 3}
 (C^{\prime}n)^{-\eta( |{\bf m}|-|{\bf l}|)/4} (C^{\prime}n)^{(9\eta-2)|{\bf m} -{\bf l}|/12} (1+O(n^{-1+3\eta} /\delta))
\end{eqnarray*}

\qed
%%%%%%%%%%%%%%%%%%%%%%%%%%%%%%%%%%%%%%
\medskip

A consequence of this lemma is the following.
\begin{cor}
\label{gqo}
Let $\eta < 2 / 9$ and let $(\mathbf{m},{\lambda})$ be such that $|\mathbf{m}| \leq n^{\eta}$. Assume as in Lemma \ref{non-orth} that $\lambda_i - \lambda _{i+1} > \delta n $ for all $1\leq i \leq d-1$  and $\lambda _d > \delta n$, for some $\delta > 0$, and that
 $Cn^{3\eta - 1}/\delta^{2} <1$ where $C=C(d) is a constant$.

Then there exists a constant $C^{\prime\prime}= C^{\prime\prime}(d,\eta)$ such that
\begin{equation}\label{bound.sum.inner.prod}
\sumtwo{|\mathbf{l}| \leq n^{\eta}}{\mathbf{l} \neq \mathbf{m}} 
\left| \langle \mathbf{m},{\lambda} | \mathbf{l},{\lambda} \rangle  \right| \leq (C^{\prime\prime} n)^{(9 \eta - 2) / 12}\delta^{-1/3}.
\end{equation}
%for $n$ such that $C^{\prime}n > \delta^{2}/\eta$ with $C^{\prime}$ the constant of Lemma \ref{non-orth}. 
\end{cor}

\begin{proof}
We break the sum into two parts ($|{\bf l}| \leq |{\bf m}|$ and $|{\bf l}| >  |{\bf m}|$), and by triangle inequality it suffices to prove the statement under the additional condition 
$|{\bf l}| \leq |{\bf m}|$.
  
We use \eqref{quasiorth1} neglecting the terms containing $|{\bf m}|-|{\bf l}|$ in the exponent which are less than $1$. Then the expression on the left side of 
\eqref{bound.sum.inner.prod} is bounded from above by
$$
2 \sum_{k\geq 1} N(k) \left[ (C^{\prime}n)^{(9\eta-2)/12} \delta^{-1/3} \right]^{k}
%= \sum_{p\geq 1} (\sum_{k}=1^{p} N(k)) (a^{p}-a^{p+1})
$$
where $N(k)$ is the number of ${\bf l}$'s for which $|{\bf m} - {\bf l}|= k$. 

Since there are $d(d-1)/2$ pairs $1\leq i < j \leq d$, there are at most 
$(k+1)^{d(d-1)/2}$ different choices for the  values $\{| l_{i,j} - m_{i,j}| : i<j\}$ satisfying 
$\sum  |l_{i,j} - m_{i,j}| = k$. Moreover, there are $2^{d(d-1)/2}$ sign choices which fix ${\bf l}= \{l_{i,j}\}$ completely. Thus $N(k) \leq (2(k+1))^{d(d-1)/2} \leq c^{k}$ for some constant $c$ which can be incorporated  in the geometric series starting at $k=1$, hence the desired estimate.

 % $(d(d-1)/2 - 1)^{k}$ different choices for the  values $\{| l_{i,j} - m_{i,j}| : i<j\}$ satisfying 
%$\sum  |l_{i,j} - m_{i,j}| = k$. Moreover, there are $2^{d(d-1)/2}$ sign choices which fix ${\bf l}= (l_{i,j})$ completely. Thus $N(k) \leq 2^{d(d-1)/2} (d(d-1)/2 - 1)^{k}$ and the sum above is bounded by a geometric series starting at $k=1$, hence the desired estimate.

\end{proof}

We use this quasi-orthogonality to build an isometry
$V_{\lambda}:\mathcal{H}_{\lambda} \to \mathcal{F}$ which maps the relevant 
finite-dimensional vectors $\left\vert {\bf m}, \lambda\right\rangle$ `close' to their Fock counterparts $\left\vert {\bf m}\right\rangle$. This is the aim of Lemma \ref{V_approx}.

\begin{lem}
\label{isocomp}
Let $A$ be a contraction (i.e. $A^* A \leq \bf{1}$) from a finite space $\mathcal{H}$  to an infinite space $\mathcal{K}$. Then there is an $R:\mathcal{H}\to \mathcal{K}$ such that $A + R$ is an isometry and ${\rm Range}(A) \perp {\rm Range}(R)$.

As a consequence, for any unit vector $\phi$, we have $\| R \phi \|^2 = 1 - \| A \phi \|^2$.
\end{lem}

\begin{proof}
As $\mathcal{K}$ is infinite-dimensional, we may consider a subspace $\mathcal{H}'$ of $\mathcal{K}$, orthogonal to ${\rm Range}(A)$, and the same dimension as 
$\mathcal{H}$, so that we can find an isomorphism $I$ from $\mathcal{H}$ to $\mathcal{H}'$. We then take $R = I \sqrt{\mathbf{1} - A^*A}$. 

\end{proof}

%We can now prove Lemma \ref{V_approx}.

%%%%%%%%%%%%%%%%%%%%%%%%%%%%%%%%%%%%%%%%%%%%
{\it Proof of Lemma \ref{V_approx}.}
Let $A_{\lambda}:\mathcal{H}_{\lambda}\to \mathcal{F}$ be defined by
\[
A_{\lambda} := \frac1{\sqrt{1 + (Cn)^{(9\eta -2) / 12}/\delta ^{1/3}}} \sum_{|\mathbf{l}| \leq n^{\eta}} \left| \mathbf{l} \right\rangle \left\langle \mathbf{l}, {\lambda} \right|.
\]
Then,  
\begin{align*}
A_{\lambda}^{*} A_{\lambda} & =  \frac1{1 + (Cn)^{(9\eta -2) / 12}/\delta ^{1/3}} \sum_{|\mathbf{l}| \leq n^{\eta}} \left| \mathbf{l},{\lambda} \right\rangle \left\langle \mathbf{l},{\lambda} \right| \\
& \leq \mathbf{1}_{\mathcal{H}_{\lambda}}. 
\end{align*}
where the last inequality follows from Corollary \ref{gqo} and the following argument. 
It is enough to show that all eigenvalues of $A_{\lambda}^{*}A_{\lambda}$ are smaller than $1$. Let $\sum_{\bf m} c_{\bf m} \left| \mathbf{m}, {\lambda} \right\rangle$ be an eigenvector of $A_{\lambda}^{*}A_{\lambda}$, and $a$ the corresponding eigenvalue. Then by the linear independence of $\left\vert \mathbf{m}, {\lambda} \right\rangle $ we get that for each~${\bf l}$
$$
 \frac1{1 + (Cn)^{(9\eta -2) / 12}/\delta ^{1/3}} 
\sum_{|\mathbf{m}|\leq n^{\eta}}  
\left\langle \mathbf{l},\lambda  | \mathbf{m}, {\lambda} \right\rangle  c_{\bf m} 
 = 
a   c_{\bf l}.
$$
If ${\bf l}_{0}$ is an index for which $| c_{\bf l}|$ is maximum, then by taking absolute values on both sides we obtain
$$
a \leq  \frac1{1 + (Cn)^{(9\eta -2) / 12}/\delta ^{1/3}} \sum_{|\mathbf{m}|\leq n^{\eta}}  
 \left\vert \langle \mathbf{l},\lambda  | \mathbf{m}, {\lambda}\rangle  \right\vert \leq 1.
$$

Now we may apply Lemma \ref{isocomp}, and find an $R_{\lambda}$ such that 
$A_{\lambda} + R_{\lambda}$ is an isometry, and ${\rm Range}(R_{\lambda}) \perp 
{\rm Range}(A)$, so that $\left\langle  \mathbf{m} \right| R_{\lambda} = 0$. 
We define $ V_{\lambda} : = A_{\lambda} + R_{\lambda}$.
Then
\begin{eqnarray*}
\left\langle  \mathbf{m} \right| V_{\lambda} & = & \left\langle  \mathbf{m} \right| (A_{\lambda} + R_{\lambda}) \\
& = & \left\langle  \mathbf{m} \right| A_{\lambda} \\
& = &  \frac1{\sqrt{1 + (Cn)^{(9\eta -2) / 12}/\delta ^{1/3}}} \left\langle  \mathbf{m} \right| \sum_{|\mathbf{l}| \leq n^{\eta}} \left| \mathbf{l} \right\rangle \left\langle \mathbf{l}, {\lambda} \right| \\
& = & \frac1{\sqrt{1 + (Cn)^{(9\eta -2) / 12}/\delta ^{1/3}}} \left\langle \mathbf{m},{\lambda} \right|.
\end{eqnarray*}

\qed

%%%%%%%%%%%%%%%%%%%%%%%%%%%%%%%%%%%%%%%%

\subsection{Proof of Lemma \ref{ldisplacement} on mapping rotations into displacements}
\label{pdisplacement}

We first recall a few definitions and notations. We denote by $D^{\vec{z}}$ the displacement operation (super-operator) acting on observables in the multimode Fock space $\mathcal{F}$ as
$$
D^{\vec{z}} (W( {\vec{y}})) := {\rm Ad}[ W( \vec{z})] \left(W( \vec{y})\right) 
= e^{2i \sigma (\vec{y} ,\vec{z})} \,  W (\vec{z}+\vec{y}) , 
\qquad \vec{y},\vec{z}\in \mathbb{C}^{d(d-1)/2}. 
$$
The operation acts as displacement on coherent states, in particular
$$
D^{ \vec{\zeta}+ \vec{z}} (|\mathbf{0}\rangle\langle\mathbf{0}|) = 
\vert \vec{\zeta}+ \vec{z} \rangle\langle \vec{\zeta}+ \vec{z} \vert.
$$
Similarly, on the finite dimensional space $\left(\mathbb{C}^{d}\right)^{\otimes n}$ 
we have the action (cf. \eqref{generalU})
$$
\Delta^{\quant,\vec{\xi},n}(A) 
= {\rm Ad}[U(\quant, \vec{\xi}, n)](A):= 
U(\quant /\sqrt{n},\vec{\xi}/\sqrt{n} )^{\otimes n} \,A \, U^*(\quant/\sqrt{n},\vec{\xi}/\sqrt{n})^{\otimes n},
$$
whose restriction to the block $\lambda$ is 
$\Delta^{\quant,\vec\xi, n}_{\lambda} ={\rm Ad}[U_{\lambda}(\quant,\vec{\xi}, n)] $. 
%First we know that $D^{\quant +z}(|\mathbf{0}\rangle\langle\mathbf{0}|) $ is the
%density matrix of a (coherent) pure state $|\quant + z\rangle$ whose
%decomposition on the Fock basis is given by (\ref{coherent}). 

The isometric embedding $T_{\lambda}(\cdot) := V_{\lambda} \cdot V_{\lambda}^{*} $ and its `adjoint' $T_{\lambda}^{*}(\cdot) := V_{\lambda}^{*} \cdot V_{\lambda} $ satisfy
$$
T_{\lambda}\Delta^{\quant+\vec{z}, \vec{\xi}, n}_{\lambda} T_{\lambda}^*
(\left\vert \mathbf{0}\right\rangle\left\langle\mathbf{0}\right\vert)  = 
V_{\lambda} 
\vert  \quant + \vec{z}, \vec{\xi},  \lambda \rangle
\langle \quant + \vec{z}, \vec{\xi}, \lambda \vert  
V_{\lambda}^{*}
$$
where $\vert  \quant + \vec{z}, \vec{\xi},  \lambda \rangle := 
U_{\lambda}(\quant+\vec{z}, \vec{\xi},n)\left\vert {\bf 0}, \lambda\right\rangle$ are the 
`finite dimensional coherent states'.

%This is a pure state \mbox{$V_{\lambda}U(\quant+z,\gamma,n)\Vzero$} (recall that
%$\Vzero$ is the semi-standard Young tableau with only $i$ in row $i$).
According to Lemma \ref{V_approx}, the coordinates of 
$V_{\lambda} \vert  \quant + \vec{z}, \vec\xi, \lambda \rangle$ 
in the Fock basis are described  by:
\begin{align}
\label{etape1}
\langle {\bf m} | V_{\lambda} \vert \quant+\vec{z},\cartan,\lambda \rangle =
\left\{ \begin{array}{l} 
%0 \mbox{ if } {\bf m}\not\in\lambda; \\
\langle {\bf m},{\lambda}| U_{\lambda}(\quant+\vec{z},\cartan,n) |{\bf 0},{\lambda}\rangle 
(1 + O(n^{(9\eta -2) / 12}\delta ^{-1/3}))
\mbox{ if }
|{\bf m}| \leq n^{\eta};\vspace{2mm}\\
\mbox{something not important if } |{\bf m}| > n^{\eta} .
\end{array} \right.
\end{align}
%Note that we may recast $(1 + (Cn)^{(9\eta -2) / 12}/\delta ^{1/3})^{-1/2}$ as $1 + O(n^{(9\eta -2) / 12}\delta ^{-1/3})$.

Using the relation $\|  |\psi\rangle\langle \psi| - |\psi^{\prime} \rangle \langle \psi^{\prime}| \|_{1}= 2\sqrt{ 1 - |\langle\psi
|\psi^{\prime}\rangle|^2}$, which holds for unital vectors $\psi,\psi^{\prime}$, the statement of the lemma is equivalent to
\begin{equation}
\label{lemplus2}
\sup_{\| \vec{z}\|\leq n^{\beta}}
\sup_{\quant\in\Theta_{n,\beta}}\,
\sup_{\|\cartan\|\leq n^{-1/2 + 2 \beta}\!/\delta}\,
\sup_{\lambda\in\Lambda_{n, \alpha}} 
1 - 
\left|  \langle \vec{z}+\quant | V_{\lambda} | \quant+\vec{z},\cartan, {\lambda}\rangle\right|  
 = R(n)^2,
\end{equation}
with $R(n)$ the original remainder term.

We shall prove formula \eqref{lemplus2} by decomposing these vectors in the Fock basis, that is 
\begin{equation}
\label{decomposition}
\langle \quant+ \vec{z} | V_{\lambda}  |\quant+\vec{z},\cartan,{\lambda}\rangle = 
\sum_{\bf m} 
\langle \quant + \vec{z} | {\bf m}\rangle 
\langle {\bf m} | V_{\lambda} |\quant+\vec{z},\cartan,{\lambda}\rangle.
\end{equation}

The estimates are based on the following observations.

1) The coherent states have significant coefficients $\langle \vec\zeta + \vec{z} |{\bf m}\rangle$ only for `small' ${\bf m}$'s, i.e. those in the set
\begin{equation}
\label{adaptedm}
\mathcal{M}:= \{{\bf m} : m_{i,j}\leq |(\quant + \vec{z})_{i,j}|^2 n^{\epsilon} 
% \leq 2 |(\quant + \vec{z})_{i,j}| n^{\beta + \epsilon} 
, \quad i<j \}.
\end{equation} 
In particular, since $2\beta+ \epsilon<\eta$ we have $\mathcal{M}\subset \{{\bf m} : |{\bf m}| \leq n^{\eta}\}$.

2) The coefficients 
$\langle {\bf m} | V_{\lambda} |\quant+\vec{z},\cartan,{\lambda}\rangle$ are uniformly close to  $\exp(i\phi)\langle \quant + \vec{z} | {\bf m}\rangle$ where $ \phi$ is a fixed real phase, in particular uniformly over ${\bf m}\in \mathcal{M}$. 

3) If $a_{\bf m}$ and $b_{\bf m}$ are the two sets of coefficients, such that 
$\sum_{\bf m} |a_{\bf m}|^2 = \sum_{\bf m} |b_{\bf m}|^2 = 1$, then  
\begin{equation}
\label{boundsubset}
 1 - \left| \sum_{\bf m} a_{\bf m} b_{\bf m} \right|   \leq  1- \left| \sum_{{\bf m}\in\mathcal{M}} a_{\bf m} b_{\bf m} \right|  + 
\left| \sum_{{\bf m}\notin \mathcal{M}} a_{\bf m} b_{\bf m}\right|  \leq 2 
\left(1-  \left|\sum_{{\bf m}\in\mathcal{M}} a_{\bf m} b_{\bf m}\right| \right).
\end{equation}

\medskip 

The precise statement in point 1) is %if $\epsilon n^{\beta} \geq \beta$ then
\begin{equation}
\label{tail}
\sum_{\mathbf{m} \not \in \mathcal{M}}  | \langle\quant + \vec{z} | {\bf m}\rangle |^2 \leq d^2 n^{-\beta}.
\end{equation}
Indeed, the inner products can be written as a product over the $(i,j)$ oscillators and we have the bound
$$
\sum_{\mathbf{m} \not \in \mathcal{M}}  | \langle\quant + \vec{z} | {\bf m}\rangle |^2\leq
\sum_{i<j}  \exp(-x_{i,j}) \sum_{k> x_{i,j} n^{\epsilon}} \frac{x_{ij}^{k}}{k!}, \qquad
x_{i,j}=  |(\quant + \vec{z})_{i,j}|^{2}.
$$
Each of the terms in the sum is a tail of Poisson distribution and is bounded by 
$ n^{-\epsilon n^{\beta}}$ if $x_{i,j} \geq 1$ and by $n^{-\beta}$  if $x_{i,j} <1$.

%We consider separately the $\mathbf{m}$ on which  there is weight, that is those satisfying for all $(i,j)$:
%\begin{equation}
%\label{adaptedm}
%m_{i,j}\leq |(\quant + z)_{i,j}|^2 n^{\epsilon}  \leq 2 |(\quant + z)_{i,j}|
% n^{\beta + \epsilon}.
%\end{equation} 
%We shall use the second form, the condition for applying formula \eqref{numer}. We denote this set by $\mathcal{M}$. Notice that 
%\begin{equation}
%\label{tail}
%\sum_{\mathbf{m} \not \in \mathcal{M}}  | (\quant + z | {\bf m}\rangle |^2 \leq d^2 n^{-\beta}
%\end{equation}
%as long as $\epsilon n^{\beta} \geq \beta$. Indeed, we end up with $\exp(-x)\sum_{k > x n^{\epsilon}} x^{k}/k! \leq n^{-\epsilon n^{\beta}}$ if $x = |(\quant + z)_{i,j}| \geq 1$ and, if $|(\quant + z)_{i,j}| < 1$, the remainder series is directly less than $n^{-\beta}$.

\medskip

We turn now to point 2). From the third line of \eqref{etape1} we get
%First, recalling that $\eta \geq 2 \beta + \epsilon$, we may use third line of \eqref{etape1}: 
\begin{align*}
\langle {\bf m}| V_{\lambda}  U_{\lambda}(\quant+\vec{z},\cartan,n) |{\bf
0}, {\lambda}\rangle 
& = \frac{ \langle y_{\lambda} \Vm | y_{\lambda}
U(\quant+\vec{z},\cartan,n) | f_{\bf 0}\rangle}{\sqrt{\langle y_{\lambda} \Vzero |
y_{\lambda}\Vzero\rangle}\sqrt{\langle y_{\lambda} \Vm |
y_{\lambda}\Vm\rangle}} (1 + O(n^{(9\eta -2) / 12}\delta ^{-1/3}))
\\
& = \frac{\langle p_{\lambda} \Vm  | q_{\lambda} U(\quant
+\
\vec{z},\cartan,n)
\Vzero \rangle }{\sqrt{\langle p_{\lambda} \Vm |
q_{\lambda} p_{\lambda} \Vm \rangle}} (1 + O(n^{(9\eta -2) / 12}\delta ^{-1/3}) )
\end{align*} 
where we have used (\ref{csq1}) and (\ref{fcoherent}).

We recall that  $ \mathcal{O}_{\lambda}({\bf m})$ is the orbit in $
(\mathbb{C}^{d})^{\otimes n}$ of $\Vm$ under $\mathcal{R}_{\lambda}$ and that we have the decomposition
$$
p_{\lambda} \Vm = \sum_{\Va\in \mathcal{O}_{\lambda}({\bf m})}
\frac{\# \mathcal{R}_{\lambda}}{\# \mathcal{O}_{\lambda}({\bf m})} \Va.
$$
%The multiplicative constant is the same on the numerator and denominator, so
%that we can write, with $Id_c$ denoting  the identity of $[1,l(c)]$,
Then, by employing  formulas \eqref{numer} and \eqref{denom54}, we can write
\begin{align}
\label{dev}
 \langle {\bf m}| V_{\lambda}  U_{\lambda}(\quant+\vec{z},\cartan,n) |{\bf 0},{\lambda}\rangle 
&
= \frac{\sum_{\Va\in \mathcal{O}_{\lambda} ({\bf m})}\langle \Va | q_{\lambda}
U(\quant + \vec{z}, \cartan,n) \Vzero\rangle}{\sqrt{\sum_{\Va,\Vb \in
\mathcal{O}_{\lambda} ({\bf m})} \langle  \Va | q_{\lambda} \Vb \rangle}}(1 + O(n^{(9\eta -2) / 12}\delta ^{-1/3}) \\
&
= e^{i\phi -\|\quant + \vec{z}\|^2_2/2}\prod_{i\leq j}\frac{(\quant + \vec{z})_{i,j}^{m_{i,j}}}{\sqrt{m_{i,j}!}}\left(\frac{n(\mu_i - \mu_j)}{\la_i - \la_j}\right)^{m_{i,j}/2}r(n). \notag
\end{align} 
The corresponding remainder term is
\[
r(n) =  1 + O\left(n^{(9\eta -2) / 12}\delta ^{-1/3},n^{-1 + 2\beta + \eta}\delta^{-1}, n^{-1/2 + 3\beta+2\epsilon} \delta^{-3/2}, n^{-1 + \alpha + 2\beta}\delta^{-1}, n^{-1 + \alpha + \eta}\delta^{-1}, n^{-1 + 3\eta}\delta^{-1} \right)
\]
and the phase is:
\[
\phi =  \sqrt{n} \sum_{i = 1}^{d-1} (\mu_i - \mu_{i+1}) \cartansub_i.
\]
Since $\lambda \in \Lambda_{n,\alpha}$ and the eigenvalues are separated by $\delta$ we have $\left(\frac{n(\mu_i - \mu_j)}{\la_i - \la_j}\right)^{m_{i,j}/2} = 1 + O(n^{\alpha - 1 + \eta} / \delta) $ and the error can be absorbed in $r(n)$.

In conclusion, for $\mathbf{m}$ satisfying  \eqref{adaptedm}, we have:
\[
 \langle
{\bf m} | V_{\lambda} U(\quant+\vec{z},\cartan,n) |{\bf 0},{\lambda}\rangle =\exp(i\phi) 
\langle {\bf m} | \quant + \vec{z} \rangle r(n) .
\]

Inserting this result into \eqref{decomposition}, and using \eqref{boundsubset} and 
\eqref{tail}, we get
\begin{eqnarray*}
1-\left|\langle \vec{z}+\quant | V_{\lambda}  U(\quant+\vec{z},\cartan,n) 
|{\bf 0} ,{\lambda}\rangle \right|& = &  
O\left( 1- r(n),  \sum_{\mathbf{m} \not\in \mathcal{M}}  |\langle {\bf m} | \quant + \vec{z} \rangle |^2\right) =  R_2(n),
\end{eqnarray*}
with
\begin{multline*}
R_2(n) =  O\left(n^{(9\eta -2) / 12}\delta ^{-1/3},n^{-1 + 2\beta + \eta}\delta^{-1}, n^{-1/2 + 3\beta+2\epsilon} \delta^{-3/2}, n^{-1 + \alpha +2 \beta}\delta^{-1}, \right.\\
\left. n^{-1 + \alpha + \eta}\delta^{-1}, n^{-1 + 3\eta}\delta^{-1}, n^{-\beta}\right).
\end{multline*}
Through expression \eqref{lemplus2}, noticing that $R_2(n) = R(n)^2$, we see that we have proved the lemma.

\qed

%%%%%%%%%%%%%%%%%%%%%%%%%%%%%%%%%%%%%%%%%%%%

\subsection{Proof of Lemma \ref{lconcentration} on typical Young diagrams}

Recall that the state $\rho^{\theta,n}:= \rho^{\otimes n}_{\theta/\sqrt{n}}$ has the decomposition over `blocks' $\lambda$ given by \eqref{prerhon}. The probability distribution over Young diagrams $p_{\lambda}^{\quant,\clas,n}$ depends only on the diagonal parameters $\vec{u}$ and is given by
\[
p_{\lambda}^{\quant,\clas,n}=
c_n^{\lambda}\sum_{{\bf m}\in \lambda} \prod_{i=1}^{d}  (\mu_i^{\clas,n})^{\lambda_i}
\prod_{j=i+1}^{d}\left(\frac{\mu_j^{\clas,n}}{\mu_i^{\clas,n}}\right)^{m_{i,j}} ,
\]
with 
\[
c_n^{\lambda}= {n \choose \lambda_1,\lambda_2,\dots,\lambda_d}
\prod_{l=1}^d\frac{\lambda_l!\prod_{k=l+1}^{d} (\lambda_l - \lambda_k +k -l)}{(\lambda_l +d -l)!}.
\]
The above formula can be understood as follows. By invariance under rotations we can take $\vec{\zeta}=0$ and the state is diagonal in the standard basis
basis $\left(\mathbb{C}^{d} \right)^{\otimes n}$ formed by the vector $f_{\bf a}$. Each eigenprojector carries a weight $\prod_{i=1}^{d} (\mu^{\clas,n})^{m_{i}}$ where $m_{i}$ is the multiplicity of the vector $f_{i}$ in the tensor product $f_{\bf a}$. Thus, we only need to add all multiplicities over vectors that are `inside' the block $\lambda$. Since the irreducible representation has basis $f_{\bf m}$ labelled by semistandard Young tableaux, we get a factor 
$$
\prod_{i=1}^{d} (\mu_{i}^{\clas,n})^{m_{i}} = 
\prod_{i=1}^{d} (\mu_{i}^{\clas,n})^{\lambda_{i}} \prod_{j=i+1}^{d} 
\left(\frac{\mu_j^{\clas,n}}{\mu_i^{\clas,n}}\right)^{m_{i,j}}.
$$
%The fact that $\{f_{\bf m}\}$ is not an orthogonal basis does not matter since subspaces with different multiplicities $\{m_{i}\}$ are mutually orthogonal. 
The additional factor 
$c_{n}^{\lambda}$ is the dimension of $\mathcal{K}_{\lambda}$, on which the state is proportional to the identity.

Recall that $\mu_i^{\clas,n} = \mu_i + u_i /\sqrt{n}$ for $1\leq i\leq (d-1)$ and $\mu_d^{\clas,n} = \mu_d - (\sum_i u_i)/\sqrt{n}$. If $\delta \geq  2 d n^{\alpha-1} \geq 
2 d n^{\gamma - 1/2} $ then 
$\mu_j^{\clas,n}/\mu_i^{\clas,n} \leq 1$ for all $\| \vec{u} \| \leq n^{\gamma}$.
%Thus, for $n > (4/\delta)^{\frac1{1 - \alpha}}$, the $\mu^{\clas,n}_i$ are non-increasing for all $\|\clas\|\leq n^{\gamma}$, recalling $\gamma \leq \alpha$. 
Moreover $m_{i,j}\leq n$ for all $(i,j)$, so the total number of ${\bf m}$'s is smaller than 
$n^{d^{2}}$.  Thus
\[
\sum_{{\bf m}} \prod_{i<j} (\mu^{\vec{u},n})^{\lambda_{i}} \, \left(\frac{\mu_j^{\clas,n}}{\mu_i^{\clas,n}}\right)^{m_{i,j}} \leq
n^{d^2}.
\]
On the other hand ${\bf m}=\mathbf{0}$ is always in the set of possible ${\bf m}$,
so that 
\[
\sum_{{\bf m}}\prod_{i<j}\left(\frac{\mu_j^{\clas,n}}{\mu_i^{\clas,n}}\right)^{m_{i,j}} \geq 1.
\]
One can easily verify that
\[
1 \geq \prod_{l=1}^d\frac{\lambda_l!\prod_{k=l+1}^{d} ( \lambda_l - \lambda_k
+k -  l)}{(\lambda_l +d -l)!} \geq \frac1{(n+d)^{d^2}}.
\]

The remaining factor is the multinomial law. We now show that this is the
dominating part. Let us write $(Y_1,\dots,Y_d)$ for the multinomial random variable. Then we have
\begin{equation}
\label{concmult}
\mathbb{P}[|Y_i- n \mu^{\clas,n}_i |\geq x]\leq 2\exp\left(-\frac{2x^2}{n}\right).
\end{equation}
Indeed each $Y_{i}$ is a sum of independent Bernoulli variables 
$X_{1},\dots, X_{n}$ with $\mathbb{P}(X_{k}=1) = \mu^{\clas,n}_i$ and $\mathbb{P}(X_{k}=0) =1- \mu^{\clas,n}_i$, and by Hoeffding's inequality \cite{vanderVaart&Wellner}
\[
\mathbb{P}[|\sum_{k=1}^{n} X_k- \mathbb{E}[X_k]|\geq x]\leq 2\exp\left(-\frac{2x^2}{n}\right).
\]

%We apply this to the Bernoulli random variable that yields $1$ with probability
%$\mu^{\clas,n}_i$ and $0$ with probability $1-\mu^{\clas,n}_i$ and we get an deviation inequality on the possible results of the multinomial law:
%\begin{equation}
%\label{concmult}
%\mathbb{P}[|Y_i- n \mu^{\clas,n}_i |\geq x]\leq 2\exp\left(-\frac{2x^2}{n}\right).
%\end{equation}
By definition, for any $\lambda\notin \Lambda_{n,\alpha}$ there exists an $i$ such that 
$|\lambda_i - n \mu_i |\geq n^{\alpha}$, which implies 
$|\lambda_i - n \mu^{\clas,n}_i|\geq  n^{\alpha} - d n^{\gamma+1/2}$. With 
$n^{\alpha-\gamma-1/2} >2d$, the upper bound is simply $n^{\alpha}/2$ and we have
%Now, for $n > (4/\delta)^{\frac1{1- \alpha}}$, for all $\|\clas\|\leq n^{\gamma}$, and all $\lambda\not\in
%\Lambda_{n,\alpha}$, there is a $i$ such that $|\lambda_i-n \mu^{\clas,n}_i|\geq
%(1/2d) n^{2/3}$, so that
\begin{align*}
\sum_{\lambda\notin \Lambda{n,\alpha}} \| b_{\lambda}^{\theta,n} \|_{1} = 
\mathbb{P} [\lambda\not\in \Lambda_{n,\alpha} ]& \leq n^{d^2}\sum_{i=1}^d
\mathbb{P}[|Y_i- n \mu^{\clas,n}_i |
\geq  n^{\alpha}/2]
\\ &
\leq
2d n^{d^2}\exp(- n^{2 \alpha - 1}/2) .
\end{align*}

\qed

%%%%%%%%%%%%%%%%%%%%%%%%%%%%%%%%%%%%%%%

\subsection{Proof of Lemma \ref{lclassical} and Lemma \ref{sigma} on classical LAN}
\label{proofs.classical.lemmas}
We shall use multinomials as an intermediate step. 
Recalling that $b_{\lambda}^{\glob,n} = p_{\lambda}^{\glob,n} \tau_{\lambda}^n$, we can write:
\begin{multline}
\label{multientre}
\left\| \mathcal{N}(\clas,V_{\mu}) -
\sum_{\lambda}
b_{\lambda}^{\glob,n}\right\|_1 \leq 
\left\| 
p^{\glob,n} -  M^n_{\mu^{\clas,n}_1,\dots,\mu^{\clas,n}_d}\right\|_1
+ 
\left\| \mathcal{N}(\clas,V_{\mu}) -
\sum_{\lambda}  M^n_{\mu^{\clas,n}_1,\dots,\mu^{\clas,n}_d}(\lambda) \tau_{\lambda}^n
\right\|_1, 
\end{multline}
where $ M^n_{\mu^{\clas,n}_1,\dots,\mu^{\clas,n}_d}$ is the $d$-multinomial with coefficients $\mu_i^{\clas,n}$. 

Concisely,  what we really prove in this lemma is the equivalence of the following classical experiments, together with an explicit rate:
\begin{align*}
\mathcal{P}_n & =\left\{ p^{\clas,n}, \|\clas\|\leq n^{\gamma}  \right\}, \\
\mathcal{M}_n & = \left\{ M^n_{\mu^{\clas,n}_1,\dots,\mu^{\clas,n}_d},
\|\clas\|\leq n^{\gamma} \right\}, \\
\mathcal{G}_n & = \left\{\mathcal{N}(\clas, V_{\mu}), \|\clas\|\leq
n^{\gamma} \right\}.
\end{align*} 

Recall  that $p^{\glob,n}$ does not depend on $\vec{\zeta}$ and is denoted 
$p^{\clas,n}$. We shall use the shorthand notation 
$M^{n,\clas} := M^n_{\mu^{\clas,n}_1,\dots,\mu^{\clas,n}_d}$.

We first bound the first term on the left side of \eqref{multientre} as follows:
\begin{equation}
\label{mult}
\sup_{\|\clas\|\leq n^{\gamma}} \left\| p^{\clas,n} - M^{n,\vec{u}} \right\|_1 \leq
 C \frac{n^{-1 / 2 + \gamma} + n^{\alpha - 1}}{\delta}.
\end{equation}
To show this, we rewrite:
\begin{align*}
\left\| p^{\clas,n} -
M^{n,\vec{u}} \right\|_1 &= \sum_{|\lambda|=n}
|p^{\clas,n}_{\lambda} - 
M^{n,\vec{u}} (\lambda)| \\
& \leq \sum_{\lambda\in \Lambda_{n, \alpha}}|p^{\clas,n}_{\lambda} - 
M^{n,\vec{u}} (\lambda)| 
+\sum_{\lambda\not\in \Lambda_{n, \alpha}} 
\left[ p^{\clas,n}_{\lambda} + 
M^{n,\vec{u}}(\lambda)\right].  
\end{align*} 
 Lemma \ref{lconcentration} and (\ref{concmult}) imply that for all
$\|\clas\|\leq n^{\gamma}$, and $n > (2d/ \delta)^{\frac{1}{1 - \alpha}} + (2d)^{\frac{1}{\alpha-\gamma-1/2}}$,
\[
\sum_{\lambda\not\in \Lambda_{n}} p^{\clas,n}_{\lambda} + 
M^{n,\vec{u}}(\lambda)\leq
C_1\exp(-(C_2 n^{2 \alpha - 1})),
\]
with $C_1$ and $C_2$ depending only on the dimension.
We end the proof of (\ref{mult}) by recalling that 
\[
p^{\clas,n}_{\lambda} =\prod_{l=1}^d\frac{\lambda_l!\prod_{k=l+1}^{d} 
(\lambda_l - \lambda_k +k - l)}{(\lambda_l +d
-l)!}\sum_{{\bf m}\in \lambda}\prod_{i<j}\left(\frac{\mu_j^{\clas,n}}{\mu_i^{\clas,n}}\right)^{m_{i,j}}
M^{n,\vec{u}}(\lambda).
\]
Now, for all $\|\clas \|\leq n^{\gamma}$ and all $\lambda \in \Lambda_{n,\alpha}$, the right hand side without the multinomial is
\[
\prod_{l<k}
\left(1-\frac{\mu_{k}}{\mu_{l}} + O(n^{\alpha-1}/\delta)\right)
%\frac{ n\mu_l - n\mu_k + O(n^{\alpha})}{n\mu_l + O(n^{\alpha})}
\sum_{{\bf m}\in\lambda}\prod_{i<j}\left(\frac{\mu_j}{\mu_i}+
O(n^{-1/2 + \gamma}/\delta)\right)^{m_{i,j}}.
\]

If  $\lambda_{i}-\lambda_{i+1} > (d-1) n^{1/2}$ then $\lambda$ contains all the multiplicities ${\bf m}$ in the `cube' $\{ 0,1,\dots, n^{1/2} \}^{d(d-1)/2}$. 
Since $\mu_{i}-\mu_{i+1}>\delta$, the condition holds for all  
$\lambda\in \Lambda_{n, \alpha}$, with $n$ satisfying $n\delta>2dn^{\alpha}$. Thus

\[
\prod_{i<j} \frac{1- (\frac{\mu_j}{\mu_i} +
O(n^{-1 / 2 + \gamma}/\delta))^{n^{1/2}}}{1 -
\frac{\mu_j}{\mu_i}+
O(n^{-1 / 2 + \gamma}/\delta)}
  \leq 
\sum_{{\bf m}\in\lambda}\prod_{i<j}\left(\frac{\mu_j}{\mu_i}+
O(n^{-1 / 2 + \gamma} /\delta)\right)^{m_{i,j}}\leq \prod_{i<j} \frac{1}{1 -
\frac{\mu_j}{\mu_i}+
O(n^{-1 / 2 + \gamma} /\delta)}.
\]
Putting together yields
\[
\left|\prod_{l=1}^d\frac{\lambda_l!\prod_{k=l+1}^{d} \lambda_l - \lambda_k
+k -       l}{(\lambda_l +d
-l)!}\sum_{{\bf m}\in\lambda}\prod_{i<j}\left(\frac{\mu_j^{\clas,n}}{\mu_i^{\clas,n}}\right)^{m_{i,j}}
- 1\right| \leq C \frac{n^{-1 / 2 + \gamma} + n^{\alpha - 1}}{\delta}.
\]
We have thus proved (\ref{mult}).

\medskip

The second term in \eqref{multientre} can be treated by  `classical' (albeit technical) methods and we refer to \cite{Kahn_phd} for the details of the proof. The result is 
\[
\sup_{\| \vec{u} \leq n^{\gamma}} 
\| \mathcal{N}(\vec{u} ,V_{\mu}) - M^{n,\vec{u}}\|_{1} \leq C (n^{- 1 / 4 + \epsilon} + n^{- 1 / 2 + \gamma})
/ \delta,
\]
for $n^{-1/2 + \gamma} > C \delta / 2 $ with $C=C(d)$. Together with \eqref{mult}, and noticing that $ \alpha - 1 > \epsilon - 1/2$ for small enough $\epsilon$, we get the desired rate of convergence for Lemma \ref{lclassical}.

From here, proving Lemma \ref{sigma} (that is the inverse direction) is easy 
enough. Indeed, recall that $\sigma^n\tau^n p^{\glob,n} = p^{\glob,n}$ and that
$\sigma^n$ is a contraction. Then 
\begin{align*}
\left\|\sigma^n \mathcal{N}(\clas,V_{\mu}) -  p^{\quant,\clas,n}  \right\|_1 &=
\left\|\sigma^n \mathcal{N}(\clas,V_{\mu}) -  \sigma^n\tau^n p^{\quant,\clas,n}
\right\|_1 \\
& \leq \left\|\mathcal{N}(\clas,V_{\mu}) - \tau^n p^{\quant,\clas,n}  \right\|_1.
\end{align*} 
So that we have the same speed and conditions as those of
 Lemma \ref{lclassical}.

\qed

\subsection{Proof of Lemma \ref{len0} on convergence to the thermal equilibrium state}

We recall that the state $\Phi$ on $CCR(L^{2}(\rho),\sigma)$ was defined in 
\eqref{eq.decom.state} and is the product of a classical Gaussian distribution 
and $d(d-1)/2$ Gaussian states $\Phi_{i,j}$ of quantum harmonic oscillators, one for each pair $i<j$. $\Phi_{i,j}$ are thermal equilibrium states with inverse temperature $\beta= \ln(\mu_{i}/\mu_{j})$ (cf. \eqref{phi_integrale_gaussienne}). The joint state 
$\Phi^{\quantzero}:=\bigotimes_{i<j} \Phi_{i,j}$ is then displaced to obtain 
$\Phi^{\quant}$ but Lemma \ref{len0} is only concerned with $\Phi^{\quantzero}$.

It is well known that thermal equilibrium states are diagonal in the number basis and in our case
\begin{equation}
\label{phi0}
\Phi^{{\quantzero}} = \sum_{{\bf m}\in \mathbb{N}^{d(d-1)/2}} \, \prod_{i<j}
\frac{\mu_i-\mu_j}{\mu_i}\left(\frac{\mu_j}{\mu_i}\right)^{m_{i,j}}|{\bf
m}\rangle \langle{\bf m}|.
\end{equation}

%\medskip

As shown in \eqref{jesaispas}, a similar formula holds for the finite dimensional 
block states $\rho^{{\quantzero},\clas,n}_{\lambda}$: 

\begin{equation}
\langle \mathbf{m}, \lambda  |  \rho^{\quantzero, \clas, n}_{\lambda} | \mathbf{m}, \lambda\rangle = 
%\prod_{i=1}^d
%(\mu_i^{\clas, n})^{\lambda_i} 
C^{\vec{u}}_{\lambda}\, 
        \prod_{i<j}^d \left(\frac{\mu_j^{\clas,n}}{\mu_i^{\clas,n}}\right)^{m_{i,j}},
        \end{equation} 
        where $C^{\vec{u}}_{\lambda}$ is a normalisation constant, 
        $\mu_i^{\clas,n} = \mu_i + u_i /\sqrt{n}$ for $1\leq i\leq (d-1)$ and $\mu_d^{\clas,n} = \mu_d - (\sum_i u_i)/\sqrt{n}$.

However there is a caveat: although $\left\vert {\bf m} ,\lambda\right\rangle$ are eigenvectors of $ \rho^{\quantzero, \clas, n}_{\lambda}$, they are not orthogonal to each other so we cannot directly use $ \left\vert {\bf m},\lambda\right\rangle \left\langle {\bf m} , \lambda\right\vert$ as eigenprojectors in the  spectral decomposition. However, Lemma 
\ref{V_approx} gives us an estimate of the error that we incur by doing just that. 

Not first that the eigenvalues of  $\rho^{{\quantzero},\clas,n}_{\lambda}$ are labelled by the {\it total} multiplicities $m_{i}$ of the index $i$ in the semistandard Young tableaux :
$$
m_{i} := \lambda_{i} - \sum_{j>i} m_{i,j} + \sum_{j<i} m_{j,i}.
$$ 
Let us denote by $\mathcal{H}( \{m_{i}\})$ and $P(\{m_{i}\})$ the corresponding eigenspace and respectively eigenprojection. Then
$$
 \rho^{\quantzero, \clas, n}_{\lambda} = 
C^{\vec{u}}_{\lambda} \sum_{\{m_{i}\}} \prod_{i=1}^d
(\mu_i^{\clas, n})^{m_i-\lambda_{i}} P(\{m_{i}\}).
$$ 
As in Lemma \ref{V_approx} we have
$$
P({\{m_{i}\}})= 
\frac{1}{1+ Cn^{(9\eta-2)/12}\delta^{-1/3}} \sum_{{\bf m} : \{m_{i}\}}  
\left\vert {\bf m},\lambda\right\rangle \left\langle {\bf m } , \lambda\right\vert + 
E({\{m_{i}\}})
$$
where the sum runs over those ${\bf m}$ with total multiplicities $\{m_{i}\}$. 
The (positive) reminder has trace norm 
$$
{\rm Tr}( E({\{m_{i}\}})) =O(n^{(9\eta -2)/12} \delta^{-1/3}) \cdot {\rm dim} (\mathcal{H}(\{m_{i}\})).
$$
By summing over all $\{m_{i} \}$ we get
$$
 \rho^{\quantzero, \clas, n}_{\lambda} =  
\frac{1}{1+ Cn^{(9\eta-2)/12}\delta^{-1/3}} \tilde{\rho}^{\quantzero, \clas, n}_{\lambda} +  
C^{\vec{u}}_{\lambda} \sum_{\{m_{i}\}} \prod_{i=1}^d (\mu_i^{\clas, n})^{m_i-\lambda_{i}}E({\{m_{i}\}}) ,
$$
where $ \tilde{\rho}^{\quantzero, \clas, n}_{\lambda}$ is the approximate state
$$
\tilde{\rho}^{\quantzero, \clas, n}_{\lambda} := 
C^{\vec{u}}_{\lambda}   \sum_{\{m_{i}\}} \,\prod_{i=1}^d
(\mu_i^{\clas, n})^{m_i-\lambda_{i}} \sum_{{\bf m}: \{m_{i}\}}  \left\vert {\bf m},\lambda\right\rangle \left\langle {\bf m } , \lambda\right\vert.
$$
The error term has trace norm of the order
$$
O(n^{(9\eta -2)/12} \delta^{-1/3}) \cdot 
C^{\vec{u}}_{\lambda} \sum_{\{m_{i}\}} \prod_{i=1}^d (\mu_i^{\clas, n})^{m_i-\lambda_{i}} {\rm dim} (\mathcal{H}(\{m_{i}\})) = O(n^{(9\eta -2)/12} \delta^{-1/3}) ,
$$
where we have used the normalisation of the block state 
$ \rho^{\quantzero, \clas, n}_{\lambda} $.

In conclusion 
\begin{equation}\label{estimate.0}
\|  \rho^{\quantzero, \clas, n}_{\lambda} - \tilde{\rho}^{\quantzero, \clas, n}_{\lambda}\|_{1} = O(n^{(9\eta-2)/12}\delta^{-1/3}).
\end{equation}

The next step is to show that the block states $\tilde{\rho}^{{\quantzero},\clas,n}_{\lambda}$ are 
mapped by $T_{\lambda}$ close to $\Phi^{{\quantzero}}$. 
%approximate precisely enough $T_{\lambda}(\rho^{{\quantzero},\clas,n}_{\lambda})$. 
Using (\ref{jesaispas}), we can write
\begin{eqnarray}
\label{Tlrl}
T_{\lambda}(\tilde{\rho}^{{\quantzero},\clas,n}_{\lambda})& =&
C_{\lambda}^{\clas}
\sum_{\mathbf{m}\in \lambda}\prod_{i<j}\left(\frac{\mu^{\clas,n}_j}{\mu^{\clas,n}_i}\right)^{m_{i,j}} T_{\lambda}(|{\bf
m} , {\lambda}\rangle\langle{\bf m}, {\lambda}|).
\end{eqnarray}

If $n^{\alpha - 1} \leq \delta / 2$ and $\alpha > 1/2 > \eta$, we know that all $\mathbf{m}$ such that $|\mathbf{m}|\leq n^{\eta}$ `fit into' $\lambda$. 

Since $\mu_i^{\clas,n} = \mu_i + O(n^{-1/2 + \gamma})$, when $|\mathbf{m}|\leq n^{\eta}$,
\begin{equation}\label{estimate.1}
\left( \frac{\mu^{\clas,n}_j}{\mu^{\clas,n}_i}\right)^{m_{i,j}} =
\left(\frac{\mu_j}{\mu_i}\right)^{m_{i,j}} (1 + O(n^{-1/2 + \gamma + \eta}/\delta)).
\end{equation}
%We can then compute $C_{\lambda}^{\clas}$, on the one hand, and divide the left hand side of equation \eqref{Tlrl} in two parts.
For the normalisation constant  we can write:
\[
(C^{\clas}_{\lambda})^{-1} = \sum_{|{\bf m}|\leq n^{\eta}}\prod_{i<j}\left(\frac{\mu^{\clas,n}_j}{\mu^{\clas,n}_i}\right)^{m_{i,j}} +  \sum_{{\bf m}\in
\lambda:|{\bf m}|\geq n^{\eta}}\,\prod_{i<j}\left(\frac{\mu^{\clas,n}_j}{\mu^{\clas,n}_i}\right)^{m_{i,j}}. 
\]
If $2dn^{\gamma-1/2}<\delta/2$ then the second  part is less than 
%$\delta n^{\eta d^2} (1 - \delta)^{n^{\eta}}$
$ n^{d^2} (1 - \delta/2)^{n^{\eta}}$ 
%for $n^{\eta} > C\ln(n)/\delta$, where $C$ depends only on $d$. 
which is negligible compared to the other error terms. 
Hence:
\begin{align}
(C^{\clas}_{\lambda})^{-1} & =   \sum_{|{\bf m}|\leq n^{\eta}}
\prod_{i<j}\left(\frac{\mu_j}{\mu_i}\right)^{m_{i,j}} (1 + O(n^{-1/2 + \gamma + \eta}/\delta) )\nonumber\\
& =   \sum_{{\bf m}\in
\mathbb{N}^{d(d-1)/2}}\prod_{i<j}\left(\frac{\mu_j}{\mu_i}\right)^{m_{i,j}} 
(1+O(n^{-1/2 + \gamma + \eta}/\delta) \nonumber\\
& = \prod_{i<j} \frac{\mu_{i}}{\mu_i-\mu_j} (1+ O(n^{-1/2 + \gamma + \eta}/\delta)).
\label{estimate.2}
\end{align}

We then recall that for unit vectors, we have $\| |\psi\rangle\langle\psi| -  |\phi\rangle\langle\phi| \|_1 = 2\sqrt{1     - |\langle \psi |\phi \rangle|^2}$. So that, using Lemma \ref{V_approx}, we get that for $|\mathbf{m}|\leq n^{\eta}$
\begin{equation}\label{estimate.3}
\| T_{\lambda}( |\mathbf{m}, {\lambda}\rangle\langle\mathbf{m}, {\lambda}|)  - |\mathbf{m}\rangle\langle\mathbf{m}| \|_1 =
\| V_{\lambda} |\mathbf{m}, {\lambda}\rangle\langle\mathbf{m}, {\lambda}| V_{\lambda}^* - |\mathbf{m}\rangle\langle\mathbf{m}| \|_1 = O(n^{(9\eta - 2)/ 24}/\delta ^{1/6}).
\end{equation}

Putting the estimates \eqref{estimate.1}, \eqref{estimate.2}, \eqref{estimate.3}  
back into formula \eqref{Tlrl}, we obtain
$T_{\lambda}(\tilde{\rho}^{{\quantzero},\clas,n}_{\lambda})$, so that 
\begin{equation}
\label{Tlrl2}
T_{\lambda}(\tilde{\rho}^{{\quantzero},\clas,n}_{\lambda}) =  
\sum_{{\bf m}\in \mathbb{N}^{d(d-1)/2}} \prod_{i<j}
\frac{\mu_i-\mu_j}{\mu_i}\left(\frac{\mu_j}{\mu_i}\right)^{m_{i,j}}|{\bf
m}\rangle \langle{\bf m}| + O(n^{-1/2 + \gamma + \eta}/\delta, n^{(9\eta - 2)/ 24}/\delta ^{1/6}).
\end{equation}

Comparing with \eqref{phi0}, and using  \eqref{estimate.0} we get the desired result.

\qed

%%%%%%%%%%%%%%%%%%%%%%%%%%%%%%%%%%%%%%%%%%%%

\subsection{Proof of Lemma \ref{lgrouplimit} on local linearity of $SU(d)$}

The key is to notice that, as we are dealing with a group, there is a $r$ such
that 
$$
U^{-1}(\quant+ \vec{z}, \cartanzero, n)U(\quant, \cartanzero,n)U(\vec{z},\cartanzero,n) =
U(-\quant-\vec{z}, \vec{0}, n)U(\quant, \vec{0},n)U(\vec{z},\vec{0},n) = 
U(\vec{r}, \vec{s}, n),
$$
and similarly for the operation $\Delta$. 
We shall prove below that under the condition that both $\quant$ and $\vec{z}$ are
smaller than $n^{\beta}$, then $\|\vec{r}\| + \|\vec{s}\| = O(n^{-1/2 + 2\beta}/\delta)$. Let us
call this the \emph{domination hypothesis} for further reference.
 
Now, as the actions are unitary, we may rewrite the norm in Lemma as
\ref{lgrouplimit}:
\begin{align*}
A & = \left\|
[\Delta^{\quant +\vec{z},n}_{\lambda}
-
\Delta^{\quant,n}_{\lambda}\Delta^{\vec{z},n}_{\lambda}]
(|{\bf 0}, {\lambda}\rangle \langle{\bf 0}, {\lambda}|)
 \right\|_1 \\
&=  \left\|
\Delta^{-(\quant + \vec{z}),n}_{\lambda}
[\Delta^{\quant +\vec{z},n}_{\lambda}
-
\Delta^{\quant,n}_{\lambda} \Delta^{\vec{z},n}_{\lambda}]
(|{\bf 0},{\lambda}\rangle\langle{\bf 0},{\lambda}|)
 \right\|_1 \\
& = \left\|
[{\rm Id} - \Delta^{\vec{r},\vec{s},n}_{\lambda}
 ](|{\bf 0},{\lambda}\rangle\langle{\bf 0},{\lambda}|)
\right\|_1.
\end{align*}
As $T_{\lambda}$ is an isometry, we may also let it act the left and
$T_{\lambda}^*$ on the right and get:
\begin{eqnarray*}
A & = &
\left\| 
T_{\lambda}(
\left\vert {\bf 0}, \lambda \right\rangle
\left \langle {\bf 0}, \lambda \right\vert
) -
T_{\lambda}\Delta^{\vec{r},\vec{s},n}_{\lambda}T_{\lambda}^* 
( 
\left\vert {\bf 0}\right\rangle \left\langle {\bf 0}\right\vert 
)
\right\|_1 \\
& \leq&
 \left\| \left\vert {\bf 0}\right\rangle \left\langle {\bf 0} \right\vert  -
\left| \vec{r}\right\rangle \left\langle\vec{r}\right| \right\|_1  +  
\left\| \left|\vec{r}\right\rangle\left\langle\vec{r}\right| -
T_{\lambda}\Delta^{\vec{r}, \vec{s},n}_{\lambda}
T_{\lambda}^* ( \left|{\bf 0}\right\rangle\left\langle{\bf 0}\right\vert)\right\|_1+
\| T_{\lambda}
(
\left \vert{\bf 0}, \lambda \right\rangle \left \langle {\bf 0}, \lambda \right\vert
) - 
\left \vert {\bf 0} \right \rangle \left \langle {\bf 0} \right\vert \|_{1}.   
\end{eqnarray*}

By the domination hypothesis,  the norm of $\vec{r}$ is smaller than
$n^{- 1/2 + 2 \beta}/\delta$, hence $ \langle\vec{r} | \mathbf{0}\rangle 
= 1 - O(n^{- 1 + 4 \beta}/\delta^{2})$. Using  $\| |\psi\rangle\langle\psi| -  |\phi\rangle\langle\phi| \|_1 = 2\sqrt{1 - |\langle \psi |\phi \rangle|^2}$ we get that 
the first term on the right side of the inequality is $O(n^{- 1/2 + 2 \beta}/\delta)$. Notice that this is dominated by $R(n)$ given in equation \eqref{Rn} since $\eta > 2 \beta$.

For the second term, we apply Lemma \ref{ldisplacement}, 
with $\vec{z} = {0}$. By the domination hypothesis, $\|\vec{s}\| \leq n^{-1/2 + 2 \beta}/\delta$,  so we may apply Lemma \ref{ldisplacement}, and the remainder is given by $R(n)$ in equation \eqref{Rn}.  

The last term is $O(n^{(9\eta - 2)/ 24}/\delta ^{1/6})$ as shown in \eqref{estimate.3} which is dominated by $R(n)$.

\medskip 

We finish the proof of the lemma, and simultaneously that of Theorem \ref{main},
by proving the domination hypothesis. 
Recall that an arbitrary element in $SU(d)$ can be written in the exponential form 
\begin{equation*}
U(\vec{r}, \vec{s})  :=  
\exp 
\left[
 i \left(
 \sum_{i=1}^{d-1} s_i H_i + 
 \sum_{1\leq j<k\leq d} 
\frac{{\rm Re} (r_{j,k}) T_{j,k} + {\rm Im}(r_{j,k}) T_{k,j} }{\sqrt{\mu_j - \mu_k }}
\right)  \right] ,
\end{equation*}
where $(\vec{r} ,  \vec{s}) \in\mathbb{C}^{d(d-1)/2} \times \mathbb{R}^{d-1}$, and 
$T_{i,j}, \,H_{i}$ are the generators of $SU(d)$ defined in \eqref{generators_algebra}. 
A special case of this is $U(\vec{r}) := U(\vec{r},\cartanzero)$.  
In general, the map $(\vec{r}, \vec{s}) \mapsto  U(\vec{r}, \vec{s}) $ is not 
injective but becomes so if we restrict to a small enough neighbourhood 
$\mathcal{C}$ of the origin $(0,0)\in \mathbb{C}^{d(d-1)/2} \times \mathbb{R}^{d-1}$. On this neighbourhood it makes sense to define the inverse as a sort of  `logarithm' 
$$
\log U(\vec{r}, \vec{s}) := (\vec{r}, \vec{s}),
$$ 
which is a $C^{\infty}$ function.

By continuity of the product, if $\vec{x}, \vec{y}\in \mathbb{C}^{d(d-1)/2}$ 
are small enough,  then $U(-\vec{x}-\vec{y}) U(\vec{x}) U(\vec{y})  \in \mathcal{C}$. 
Since $\|\quant\| + \|\vec{z}\| / \sqrt{n} \leq n^{\beta - 1/2}/\delta$, we can apply this to 
$\vec{x} =\quant/\sqrt{n}, \vec{y}= \vec{z}/\sqrt{n}$ for $n> (C/\delta)^{\frac{1}{1/2-\beta}}$ with the constant $C$ depending only on the dimension, and get
%$U(-(\quant+\vec{z})/\sqrt{n})U(\quant/\sqrt{n})U(\vec{z}/\sqrt{n}) \in \mathcal{C}$, and 
\begin{equation*}
(\vec{r} / \sqrt{n}  , \vec{s}/\sqrt{n})=
 f( \vec{\zeta}/\sqrt{n}, \vec{z}/\sqrt{n}): = 
 \log\left[ U(-(\quant+\vec{z})/\sqrt{n})U(\quant/\sqrt{n}) U(\vec{z}/\sqrt{n})  \right] .
\end{equation*}
Since $f$ is a $C^{\infty}$ function we can expand in Taylor series and it is easy to 
show that $f(0,0)=0$, the first order partial derivatives are zero as well, and the second order derivatives are uniformly bounded in a neighbourhood of the origin. 
Thus we get 
$$
\vec{r} = 
\sqrt{n}\,  O \left( 
\frac{\| z_{i,j} \|^{2} }{ n(\mu_{i} -\mu_{j}) } , 
\frac{\| \zeta_{i,j} \|^{2} }{ n(\mu_{i} -\mu_{j}) }  
\right) 
= O( n^{-1/2+2\beta}/\delta).
$$

\qed

 {\it Acknowledgements.} We thank Richard Gill for suggesting the research topic and 
 for many useful discussions. J. Kahn acknowledges the support of the 
 French Agence Nationale de la Recherche (ANR), under grant StatQuant (JC07 205763) ``Quantum Statistics''. M. Gu\c{t}\u{a}'s research is supported by the Engineering and Physical Sciences Research Council (EPSRC).

%
%\bibliographystyle{mada}
%\bibliography{bibsud}

\end{document}